\documentclass{article}
\usepackage[T1]{fontenc}
\usepackage[utf8]{inputenc}
\usepackage{lmodern}
\usepackage[letterpaper,margin=1in]{geometry}
\usepackage{microtype}
\usepackage{amsmath,amsthm,amssymb,mathtools}
\usepackage{braket}
\usepackage{ulem}
\usepackage{enumitem}
\usepackage{tikz}
\usepackage{pgfplots}
\usepgfplotslibrary{groupplots}
\pgfplotsset{compat=1.18}
\usetikzlibrary{calc}
\usepackage{caption}

\pgfplotsset{
  free energy axis/.style={
    width=0.405\textwidth,
    height=0.235\textwidth,
    scale only axis,
    axis lines=left,
    axis line style={
      line width=0.45pt,
      -
    },
    tick style={
      line width=0.45pt
    },
    tick align=outside,
    major tick length=2pt,
    xmin=0.3333333333,
    xmax=1,
    domain=0.3333333334:0.9995,
    samples=240,
    xtick={0.3333333333,1},
    xticklabels={$1/3$,$1$},
    ytick=\empty,
    clip=true,
    title style={
      font=\small,
      align=center,
      yshift=1pt
    },
    every axis plot/.append style={
      line cap=round,
      line join=round
    },
  },
  free energy curve/.style={
    black,
    line width=0.95pt,
    no markers
  },
  uniform level/.style={
    black!45,
    line width=0.65pt,
    densely dashed,
    no markers
  },
  critical point/.style={
    only marks,
    mark=*,
    mark size=1.35pt,
    black
  }
}
\usepackage[numbers,sort&compress]{natbib}
\usepackage[colorlinks=true,citecolor=blue,linkcolor=blue,urlcolor=blue]{hyperref}
\usepackage{aliascnt}
\usepackage[nameinlink,noabbrev,capitalise]{cleveref}
\hypersetup{
  pdftitle={On the Metastability of the Mean-Field Interchange Model for Local Dimension >= 3},
  pdfauthor={Sergio Escobar, Lin Lin, Michael Ragone, Kevin D. Stubbs},
  pdfkeywords={Davies generators, interchange model, Schur--Weyl duality,
    spectral gap, metastability, mixing time}
}

\usepackage{ytableau}
\ytableausetup{boxsize=1.25em}

\definecolor{purp}{HTML}{9933ff}
\definecolor{darkcyan}{HTML}{008b8b}

\DeclareMathOperator{\Tr}{Tr}

\theoremstyle{plain}
\newtheorem{theorem}{Theorem}[section]

\newaliascnt{proposition}{theorem}
\newtheorem{proposition}[proposition]{Proposition}
\aliascntresetthe{proposition}

\crefname{proposition}{Proposition}{Propositions}
\Crefname{proposition}{Proposition}{Propositions}

\newaliascnt{lemma}{theorem}
\newtheorem{lemma}[lemma]{Lemma}
\aliascntresetthe{lemma}

\crefname{lemma}{Lemma}{Lemmas}
\Crefname{lemma}{Lemma}{Lemmas}

\newaliascnt{corollary}{theorem}
\newtheorem{corollary}[corollary]{Corollary}
\aliascntresetthe{corollary}

\crefname{corollary}{Corollary}{Corollaries}
\Crefname{corollary}{Corollary}{Corollaries}

\newaliascnt{fact}{theorem}
\newtheorem{fact}[fact]{Fact}
\aliascntresetthe{fact}

\crefname{fact}{Fact}{Facts}
\Crefname{fact}{Fact}{Facts}

\theoremstyle{definition}
\newaliascnt{definition}{theorem}
\newtheorem{definition}[definition]{Definition}
\aliascntresetthe{definition}

\crefname{definition}{Definition}{Definitions}
\Crefname{definition}{Definition}{Definitions}

\theoremstyle{remark}
\newtheorem{remark}[theorem]{Remark}

\newcommand{\F}{\mathbb{F}}

\newcommand{\R}{\mathbb{R}}
\newcommand{\Z}{\mathbb{Z}}

\newcommand{\C}{\mathbb{C}}

\newcommand{\calA}{\mathcal{A}}
\newcommand{\calB}{\mathcal{B}}

\newcommand{\calE}{\mathcal{E}}

\newcommand{\calH}{\mathcal{H}}

\newcommand{\calJ}{\mathcal{J}}
\newcommand{\calK}{\mathcal{K}}
\newcommand{\calL}{\mathcal{L}}

\newcommand{\calR}{\mathcal{R}}

\newcommand{\calT}{\mathcal{T}}

\newcommand{\calV}{\mathcal{V}}

\newcommand{\eps}{\varepsilon}

\newcommand{\loc}{\mathrm{loc}}
\newcommand{\tot}{\mathrm{tot}}

\newcommand{\down}{\downarrow}

\newcommand{\Lloc}{\calL_\loc}
\newcommand{\Jloc}{\calJ_\loc}
\newcommand{\Lrest}{\calL_\mathrm{rest}}
\newcommand{\gap}{\mathrm{gap}}

\newcommand{\betamin}{\beta_\mathrm{min}}
\newcommand{\Bohr}{\mathrm{Bohr}}
\newcommand{\idty}{\mathbf 1}

\newcommand{\ihat}{\hat{\imath}}

\numberwithin{equation}{section}

\begin{document}

\title{On the Metastability of the Mean-Field Interchange Model \\ for Local Dimension $\geq 3$}

\author{
Sergio Escobar\textsuperscript{1},
Lin Lin\textsuperscript{1,2,3},
Michael Ragone\textsuperscript{1},
and Kevin~D. Stubbs\textsuperscript{1,4}\\[0.6em]
\parbox{0.95\textwidth}{\centering\small
\textsuperscript{1}Department of Mathematics, University of California, Berkeley\\
\textsuperscript{2}Applied Mathematics and Computational Research Division, Lawrence Berkeley National Laboratory\\
\textsuperscript{3}Department of Computing and Mathematical Sciences, California Institute of Technology\\
\textsuperscript{4}School of Mathematics, University of Minnesota}}
\date{}
\maketitle

\begin{abstract}
    Many quantum systems are believed to thermalize slowly at low temperatures due to the existence of metastable phases. However, there are cases where further cooling can restore polynomial-time thermalization: we demonstrate that the Davies dynamics of the mean-field interchange model with local dimension $d\geq 3$ and single-site generalized Pauli couplings exhibits this behavior. The dynamical phase diagram is beholden to two distinguished inverse temperatures, $\beta_{\min}(d)$ and $\beta=d$, with the static first-order phase transition lying between them. At low and high temperatures surrounding the interval $(\betamin(d),d)$, the spectral gap of the Davies generator is bounded below by an inverse polynomial in system size $n$. Meanwhile, in this interval, the gap vanishes exponentially in $n$, indicating metastability. The proof uses Schur-Weyl duality to decompose the Davies generator into a classical Markov generator on Young diagrams, and the restriction to the remaining symmetry sectors whose gap is $\Omega(n^{-1})$. Thus, exponential slowdown arises entirely from the Young diagram chain.
\end{abstract}

\pagebreak

\tableofcontents

\pagebreak

\section{Introduction}

Understanding how quickly an open quantum system approaches thermal equilibrium is a fundamental problem in quantum statistical mechanics. This rich field of inquiry has recently developed a new branch: a host of new quantum algorithms for Gibbs state preparation have been discovered, many of which ultimately amount to simulation of certain dissipative Markovian dynamics~\cite{chen2025efficientexactnoncommutativequantum, Ding_2025}. Assessing their efficiency is then equivalent to determining the rate of convergence to equilibrium for these physical processes. A serious barrier to efficiency comes in the form of \textit{metastability}, wherein states are approximately stationary for long times. Characterizing when one should expect these quantum dynamics to exhibit metastability is thus of central importance to determining the viability of these Gibbs sampling algorithms.

These quantum dynamics have a classical counterpart in the Glauber dynamics, which are Markov chains used to prepare Gibbs states for classical spin systems such as the Ising model. They function by applying local energy-dependent updates to an initial distribution in a manner reminiscent of Markovian system-bath interactions, and indeed, one may dually think of these Markov chains as thermalization processes converging to a stationary Gibbs state. These well-studied dynamics hold valuable insights, including clear examples of and mechanisms for metastability. An emblematic case study appears in the Glauber dynamics of the mean-field Ising (aka Curie-Weiss) model with zero magnetic field~\cite{LevinPeresWilmer2006,DingLubetzkyPeres2009}. The mean-field Ising model exhibits a static phase transition driven by spontaneous $\Z_2$-symmetry breaking: at high temperatures, one sees only a single equilibrium phase, while at low temperatures, this fractures into a pair of equilibrium phases, one corresponding to mostly spin up and another mostly spin down. This static picture has dynamic consequences. At high temperatures, the Glauber dynamics converges fast, as the spectral gap of its Markovian generator scales inverse polynomially in system size. But at low temperatures, the Glauber dynamics converges slowly, exhibiting an exponentially vanishing spectral gap. The culprit is that the finite-volume Gibbs state is an equal-weight mixture of the mostly positive and mostly negative phases. When the system is initialized with support wholly in the positive phase, each Glauber spin update attempting to mix into the negative phase incurs a free energy penalty. These penalties stack up---escape requires the coordinated success of local updates to cross a free energy barrier proportional to the system size. Indeed, the positive phase is metastable, and generically, a system initialized near a metastable phase may take an exponentially long time to escape.

The mean-field Ising model prototypes a first mechanism for metastability. There, the existence of multiple phases with a free energy barrier between them--a static feature of the model--leads to the metastability of these same phases under Glauber dynamics. Indeed, in several physical models the mechanism for metastability is captured by a free energy landscape for a few macroscopic variables, such as magnetization or the proportions of different spin values. Competing metastable phases are separated by free energy barriers, which local updates can cross only through unlikely fluctuations. When this barrier grows proportionally to the system size, escape can take exponentially long within a range of temperatures (see e.g.~\cite{BovierDenHollander2015}).

In contrast to the rather well-developed landscape for metastability in classical systems, the theory is poorly understood for quantum systems, including the dissipative Markovian dynamics powering the aforementioned Gibbs sampling algorithms. Existing examples of this phenomenon are primarily obtained through perturbations of classical Hamiltonians~\cite{GamarnikKianiZlokapa2024,rakovszky2024bottlenecksquantumchannelsfinite}, with some notable exceptions including~\cite{alicki2008thermalstabilitytopologicalqubit,bergamaschi2026} We turn to a guiding example of quantum magnetism, the mean-field interchange model: its Hamiltonian is given by 
\begin{equation}\label{eq:introduction-hamiltonian}
    H
    =
    -\frac1n
    \sum_{1\leq i<j\leq n}
    T_{ij}
    \qquad\text{on}\qquad
    \calH=(\mathbb C^d)^{\otimes n},
\end{equation}
where $T_{ij}$ exchanges tensor factors $i$ and $j$. 
The interchange model is a natural higher-dimensional analogue of the mean-field Heisenberg ferromagnet. Indeed, when $d=2$, the identity $T_{ij} = 2\mathbf S_i\mathbin\cdot\mathbf S_j+\frac12I$, where $\mathbf{S}$ denotes the Pauli spin-1/2 vector, 
establishes this equivalence up to an additive constant.

Taking inspiration from classical static-dynamic correspondences, we turn our attention to the model's static phase diagram, which is fairly well understood. 
In~\cite{Bjornberg2016}, Bj\"ornberg determined the limiting free energy of the mean-field interchange model by first writing a variational problem on ordered probability vectors and finding an expression for the partition function in terms of the irreducible representations of the symmetric group. In our normalization, the free energy being maximized is
\begin{equation}\label{eq:introduction-free-energy}
    F_\beta(x)=\frac\beta2\sum_{i=1}^d x_i^2-\sum_{i=1}^d x_i\log x_i,
    \qquad
    x\in\Delta_d^\downarrow
    :=\left\{x_1\geq\cdots\geq x_d\geq0:\sum_i x_i=1\right\},
\end{equation}
with $0\log0=0$.
Here $F_\beta$ is the negative dimensionless physical free energy, so equilibrium states are maxima. It agrees with the mean-field Potts free energy restricted to the ordered simplex $\Delta_d^\down$, although it should be pointed out that this restriction has a variety of nontrivial consequences~\cite{EllisWang1990,Bjornberg2016}. For $d\geq3$, the uniform distribution $u = (1/d,\ldots,1/d)$ is the unique maximizer for $\beta < \beta_c(d)$ where
\begin{equation}\label{eq:introduction-beta-c}
    \beta_c(d)=\frac{2(d-1)}{d-2}\log(d-1),
\end{equation}
whereas for $\beta > \beta_c(d)$ the maximizer has one larger coordinate and $d-1$ equal smaller coordinates. 
% Both distributions maximize $F_\beta$ at $\beta_c(d)$, and the maximizing distribution changes discontinuously there. This is the first-order equilibrium transition.
At $\beta_c(d)$, $F_\beta$ has two local maxima with equal energy and the free energy maximizer changes discontinuously there---a first-order equilibrium transition.

Turning to the dynamical picture, a canonical model of open-system thermalization is provided by the Davies dynamics. The dynamics corresponding to the Heisenberg ($d=2$) model were recently analyzed in~\cite{BassoBergamaschiLinRagoneStubbs2026}. One of the key observations in that work is that the symmetry of the Hamiltonian and of the jumps implies that the Lindbladian is equivariant with respect to the joint actions of $SU(2)$ and $S_n$. Since the observable representations of $SU(2)$ and $S_n$ are unitary with respect the KMS-inner product associated with the Gibbs state $\rho$, this in turn leads to a KMS-orthogonal decomposition of $\calB(\calH)$, the algebra of observables, as 
\begin{equation}
    \calB(\calH) = \underbrace{\mathrm{comm}(SU(2) \times S_n)}_{\calA^{(0)}} \oplus_\rho \underbrace{\mathrm{comm}(SU(2) \times S_n)^{\perp_\rho}}_{\calA_{\mathrm{rest}}}.
\end{equation} The restriction of the dynamics to the space $\calA^{(0)}$ admits an interpretation as a \textit{classical} Markov chain. Their work then introduced the group mixer comparison argument to bound $\gap(\calL\vert_{\calA_\text{rest}})$ and uses classical Markov chain techniques to get the bound $\gap(\calL\vert_{\calA^{(0)}})=\Omega(1)$ for all $\beta\neq 2$.
The authors then proved that the corresponding local Davies dynamics has a gap of constant order with respect to system size at fixed temperatures above the transition and of order $n^{-1}$ below it. The low-temperature slowdown is witnessed by the total magnetization order parameter and is associated with the continuous rotational symmetry of the ordered phase.

In this work, we study the Davies dynamics of the interchange model for $d\geq 3$. Crucially, while this problem is amenable to similar representation theoretic techniques, it differs from the Heisenberg case in that for $d\geq3$ the transition is first-order, whereas for $d=2$ it is continuous~\cite{Bjornberg2016}. This too has a classical analogue in the mean-field Potts model, a generalization of the mean-field Ising model to $d$ colors. The $d=2$ case has a continuous phase transition, while the $d\geq 3$ case has a first-order phase transition. These static phase transitions have dynamical consequences: their Glauber dynamics have a phase transition coinciding with the static phase transition when $d=2$~\cite{LevinPeresWilmer2006}, while the dynamics for the $d\geq 3$ model is known to additionally exhibit a spinodal threshold distinct from the static phase transition~\cite{CuffEtAl2012}.

\subsection{Main results}

We study the Davies generator $\Lloc$ (\cref{def:davies-generator}) with normalized single-site Gell-Mann jump operators and Metropolis weights.

We first revisit Bj\"ornberg's analysis of the interchange model from another perspective. Our approach recovers the free energy and phase diagram directly from the stationary measure of the Young diagram chain, cf. \Cref{app:stationary-measure}. More importantly for the dynamics, we find that the free energy landscape contains two additional distinguished inverse temperatures. Define
\begin{equation}\label{eq:introduction-beta-min}
    \betamin(d)
    :=
    \min_{r>1}
    \frac{r+d-1}{r-1}\log r.
\end{equation}
Then, for all $d \geq 3$, we have $\betamin(d)<\beta_c(d)<d$ where $\beta_c(d)$ is the location of the first-order phase transition, \cref{eq:introduction-beta-c}. At $\betamin(d)$, the free energy develops a degenerate non-uniform critical point. Just above this threshold, it splits into a local maximum and a saddle (see \cref{fig:free-energy-profiles-d3}). At $\beta=d$, the free energy landscape returns to having a single stable point. Thus the interval
\begin{equation}\label{eq:introduction-metastable-window}
    \betamin(d)<\beta<d
\end{equation}
is precisely the region in which the free energy has two locally stable phases separated by a free energy barrier.

This distinction between the critical temperature at the phase transition and the spinodal inverse temperatures is the central phenomenon of the paper. The static phase transition occurs at the coexistence point $\beta_c(d)$, where the uniform and ordered distributions have equal free energy. The dynamics remains exponentially slow on both sides of $\beta_c(d)$, throughout the entire interval \cref{eq:introduction-metastable-window}. The dynamical thresholds are instead the two spinodal points $\betamin(d)$ and $d$. For this reason, we refer to the interval $(\betamin(d),d)$ as the \textit{metastable window}.

\begin{theorem}[Dynamical Phase Diagram]
\label{thm:full-dynamics}
    Fix $d\geq3$ and $\beta>0$. Then there exist positive constants $c_1$ and $c_2,$ depending only on $\beta$ and $d$, such that the spectral gap of the Davies generator scales as
    \begin{equation}
        \gap(\calL_{\loc}) = \begin{cases}
            \Omega(n^{-c_1}) & \text{if }\beta <  \beta_{\min}(d) \qquad \quad \;\;\; \textsf{(high temperature)} \\
            \exp(-\Theta_{\beta,d}(n)) & \text{if }\beta_{\min}(d) <\beta < d \qquad \textsf{(metastable window)} \\
            \Omega(n^{-c_2}) & \text{if }\beta >  d \qquad \qquad \qquad  \; \,  \textsf{(low temperature)}.
        \end{cases} 
    \end{equation}
\end{theorem}

Moreover, we obtain an asymptotic expression for the lower spinodal inverse temperature:
\begin{equation}\label{eq:introduction-beta-min-asymptotic}
    \betamin(d)
    =
    \log d+\log\log d+1+o_d(1).
\end{equation}
Thus, as the local dimension grows, the lower inverse-temperature threshold grows logarithmically in $d$, while its upper threshold is $d$ and so the size of the metastable window grows with the local dimension. The proof of the main theorem is given in \cref{sec:main-results} while the asymptotics for the spinodal threshold are done in~\cref{app:beta-min-asymptotics}. For any fixed trace-distance accuracy, these bounds give polynomial mixing time in the high- and low-temperature regimes and exponential mixing time throughout the metastable window; see \cref{eq:gap_and_mixing_time}.

Our main result also gives the dynamical phase diagram of the local Davies generator away from the boundaries of the metastable window. The free energy picture makes this mechanism quite concrete. If $\betamin(d)<\beta<\beta_c(d)$, the uniform distribution is the global maximum of $F_\beta$, while a lower non-uniform local maximum describes a metastable ordered phase. If $\beta_c(d)<\beta<d$, the roles reverse: the ordered distribution is globally stable, while the uniform distribution survives as a metastable local maximum. At $\beta=\beta_c(d)$, both phases have equal free energy. In every case inside the window, escaping one local maximum requires crossing a region whose free energy is lower by a constant independent of $n$. Since the stationary measure scales as $\pi_\beta(\lambda) \propto \exp\left(n F_\beta(\lambda/n)\right)$, a constant free energy barrier becomes an exponentially small stationary bottleneck. Outside the interval \cref{eq:introduction-metastable-window}, this obstruction disappears.

\subsection{Proof Strategy}

We briefly describe the main ingredients of the proof. See~\cref{sec:preliminaries} for the relevant definitions.

\paragraph{Decomposition of $\calB(\calH)$}
The first step is to decompose the space of observables as
\begin{equation}\label{eq:pf strat decomp}
    \calB(\calH)
    =
    \calA_d \oplus_\rho \calA_d^{\perp_\rho}
\end{equation}
where the direct sum is orthogonal with respect to the KMS inner product. Here, $\calA_d$ is spanned by the Schur--Weyl isotypic projectors $\Pi_\lambda$ and is the direct analogue of what was called $\calA^{(0)}$ in~\cite{BassoBergamaschiLinRagoneStubbs2026}. This decomposition reduces the problem of finding the spectral gap on $\calB(\calH)$ to the independent, smaller problems of finding the spectral gap on each KMS-orthogonal component.

\paragraph{The Young diagram master equation.}
The space $\calA_d$ is a commutative algebra, and so the restriction of the Davies generator to this algebra may be viewed as a \textit{classical} Markov generator. Our next step is to then exploit Schur--Weyl duality to compute this restriction. The Davies generator is expressed in terms of a local twirling operator which has the simple expression
\begin{equation}
    \calT(X)
    =
    \frac1d
    \sum_{i=1}^n
    \Tr_i(X)\otimes I_i.
\end{equation}
Using Young's branching rule and Pieri's rule, we compute the partial trace of an isotypic projector and show that $\calT(\Pi_\lambda)$ has support only on diagrams related to $\lambda$ by removing and re-adding one box. We may then derive the exact transition rates and stationary distribution of the Young diagram Markov chain, which thanks to the representation theory allows only one-box moves.
%In this way, the one-box geometry of the chain emerges directly from the representation theory.

\paragraph{The spectral free energy landscape.}

The exact stationary measure of this chain is given by \cref{eq:stationary-measure}. Combining the Weyl dimension formula, the hook-length formula, and Stirling's approximation gives the uniform estimate
\begin{equation}
    n^{-C}
    e^{n f_\beta(\lambda/n)}
    \leq
    \pi_\beta(\lambda)
    \leq
    n^C
    e^{n f_\beta(\lambda/n)},
\end{equation}
where $f_\beta(x) = F_\beta(x) - \max F_\beta$. We then analyze the critical points of $F_\beta$ explicitly. Every non-uniform local maximum has the form
\begin{equation}
    x_r
    =
    \left(
    \frac{r}{r+d-1},
    \frac1{r+d-1},
    \ldots,
    \frac1{r+d-1}
    \right), \qquad \text{where} \quad  \beta
    =
    \beta_d(r)
    =
    \frac{r+d-1}{r-1}\log r.
\end{equation}
Here $\beta_d(r)$ is the inverse temperature at which $x_r$ is a critical point. The shape of this relationship gives the two spinodal inverse temperatures, while comparison of $F_\beta(x_r)$ with the free energy of the uniform distribution gives the coexistence inverse temperature $\beta_c(d)$. This recovers the static phase diagram of \cref{thm:spectral-free-energy-landscape}: the uniform distribution maximizes $F_\beta$ below $\beta_c(d)$, the ordered distribution maximizes above it, and both maximize at coexistence.

\paragraph{Metastability and conductance.}

Inside the metastable window, choose a small neighborhood $\mathcal W_n$ of the lower local maximum. At coexistence, choose the neighborhood with smaller stationary mass. Its stationary mass may itself be exponentially small; this causes no difficulty because conductance depends on the relative escape flow
\begin{equation}
    \Phi(\mathcal W_n)
    =
    \frac{Q(\mathcal W_n,\mathcal W_n^c)}{\pi_\beta(\mathcal W_n)},
\end{equation}
where $Q(A,B)$ denotes the stationary probability flow from $A$ to $B$. The boundary of $\mathcal W_n$ lies a fixed amount lower in free energy than its local maximum. Consequently,
\begin{equation}
    \frac{\pi_\beta(\partial \mathcal W_n)}
    {\pi_\beta(\mathcal W_n)}
    \leq
    e^{-c_{\beta,d}n}.
\end{equation}
Since the individual transition rates are at most polynomial in $n$, the conductance remains exponentially small. The variational characterization of the spectral gap and Cheeger's inequality then give
\begin{equation}
    \gap(L)
    \leq
    e^{-c_{\beta,d}n}.
\end{equation}
A canonical-path argument supplies the matching exponential lower bound.

\paragraph{Fast mixing in the one-well regimes.}
When $\beta<\betamin(d)$ or $\beta>d$, the free energy has a unique local maximum. We prove that every non-empty superlevel set of $F_\beta$ on the ordered simplex is path connected. Discretizing such paths gives, for every Young diagram $\lambda$, a legal path to the maximizer along which the free energy never falls more than $O_{\beta,d}\left(\log (n)/n\right)$
below its initial value. This corresponds to only a polynomial loss in stationary probability. Every cut has a crossing along one of these paths. This gives an inverse-polynomial lower bound on conductance, and Cheeger's inequality yields
\begin{equation}
    \gap(L)\geq n^{-C_{\beta,d}}.
\end{equation}

\paragraph{The remaining symmetry sectors.}
Finally, we further decompose the second summand in \cref{eq:pf strat decomp} as $\calA_d^{\perp_\rho} = \calK_d \oplus_\rho \mathrm{comm}(S_n)^{\perp_\rho}$. Here $\calK_d = \mathrm{comm}(S_n) \cap \mathrm{comm}(SU(d))^{\perp_\rho}$, the set of permutation-invariant observables that are not $SU(d)$-invariant. We control the two non-commutative sectors by a group mixer comparison strategy analogous to that employed in~\cite{BassoBergamaschiLinRagoneStubbs2026}. For $\calK_d$, unitary freedom in the jump operators allows us to isolate a Lindbladian associated to $SU(d)$, giving
\begin{equation}
    -\Lloc
    \geq
    -\frac1n\calL_{SU(d)}.
\end{equation}
The auxiliary generator $-\calL_{SU(d)}$ is a quadratic Casimir with $\mathrm{gap}(\calL_{SU(d)}) = 1$ and kernel $\mathrm{comm}(SU(d))$, giving the $n^{-1}$ lower bound. For the permutation-non-symmetric sector $\mathrm{comm}(S_n)^{\perp_\rho}$, we compare $\Lloc$ with the $S_n$ group mixer
\begin{equation}
    \calL_{S_n}(O)
    =
    \frac2n
    \sum_{u<v}
    (T_{uv}OT_{uv}-O).
\end{equation}
This generator has kernel $\mathrm{comm}(S_n)$ and gap exactly two. A Dirichlet-form comparison then gives a system-size-independent lower bound on $\Lloc$ over $\mathrm{comm}(S_n)^{\perp_\rho}$.

\vspace{3mm}\noindent
Together, these arguments reduce the exponentially slow part of the quantum dynamics to an explicit classical metastability problem for the Young diagram Markov chain.

\subsection{Discussion}
We find that the Potts--interchange analogy extends beyond the static picture to a Glauber--Davies analogy. The mean-field Potts model shares several features with our classical Young diagram chain. For Glauber dynamics of the mean-field Potts model, there are two types of free energy barrier: one separates ordered distributions favoring different colors, and one which separates the uniform distribution from an ordered distribution when both are locally stable~\cite{CuffEtAl2012}. The latter produces a spinodal threshold at which a metastable well first appears determines a dynamical transition for Glauber mixing. 

However, there is an important distinction: the full Potts dynamics remains exponentially slow at low temperature because it must move between phases favoring different colors. Ordering the coordinates identifies those phases. In our setting, the Young diagram state space can be viewed as the ordered, or permutation-quotiented, version of the Potts state space. This removes the barrier between different choices of favored color, while retaining the barrier between uniform and ordered distributions. The Young diagram chain also has transition rates that differ polynomially from Potts Glauber dynamics, even though the limiting free energy landscapes agree. The quantum problem further contains additional degrees of freedom that are absent from the classical quotient chain; these other symmetry sectors can be controlled by comparison arguments and do not cause an exponential slowdown.

The mean-field interchange model provides a setting in which the complete dynamical effect of a first-order transition can be analyzed while retaining genuinely non-commutative quantum dynamics. Its symmetry reduces the slowest part of the evolution to a tractable free energy landscape, while the remaining observable sectors can be controlled by representation-theoretic comparison arguments. We hope this collection of methods---symmetry decomposition, free energy bottlenecks, and auxiliary group mixers---will be useful more broadly for analyzing quantum Gibbs samplers with large symmetry groups.

Several questions remain open. Most immediately, our estimates exclude the two spinodal inverse temperatures $\betamin(d)$ and $\beta=d$, where the free energy Hessian becomes degenerate and one expects nontrivial critical scaling. It would be interesting to determine the precise gap at these thresholds and in temperature windows that shrink toward them with $n$. Additionally, it would also be useful to understand which features of the present argument survive away from the complete graph, where Schur--Weyl symmetry is partially broken and the Young diagram chain no longer closes exactly.

\subsection{Organization}

The remainder of the paper is organized as follows. In \cref{sec:preliminaries}, we review Schur--Weyl duality, introduce the interchange Hamiltonian and Davies generator, and establish the $SU(d)\times S_n$ equivariance of the local dynamics. In \cref{sec:main-results}, we decompose the observable algebra into the three invariant symmetry sectors, state the corresponding gap bounds, and combine them into the full dynamical phase diagram.

The appendices contain the proofs. \Cref{app:rep-theory-background} collects the representation-theoretic background. \Cref{appdx:young-diagram-master-proof} derives the Young diagram master equation. \Cref{app:stationary-measure} computes its stationary measure and proves the static free energy phase diagram. \Cref{app:young-diagram-mixing} establishes fast mixing in the one-well regimes and exponential slowdown throughout the metastable window. \Cref{app:beta-min-asymptotics} determines the asymptotic behavior of the lower spinodal inverse temperature. Finally, \cref{app:sud-sector,app:Sn-sector} prove the gap bounds on the non-trivial $SU(d)$ and permutation-non-symmetric sectors, respectively.

\section{Preliminaries}\label{sec:preliminaries}
Fix integers $n\geq2$ and $d\geq2$, and consider the Hamiltonian $H$ on $\calH=(\mathbb C^d)^{\otimes n}$ defined in \cref{eq:introduction-hamiltonian}. We keep $d$ and $\beta$ fixed when taking $n\to\infty$; all constants may depend on these fixed parameters. The Hamiltonian is invariant under both collective $SU(d)$ rotations and permutations of the tensor factors, and these commuting symmetries organize its spectrum and dynamics. We begin by recalling the Schur--Weyl decomposition associated with these actions. We then introduce Davies generators and the KMS inner product, before describing the covariance properties that will later allow us to decompose the Lindbladian into invariant symmetry sectors.

\subsection{Schur--Weyl Duality}\label{subsec:schur-weyl}

Here we briefly review Schur--Weyl duality and the important consequences it has for the interchange model. We also fix the representation-theoretic notation used throughout the paper. We will use standard definitions from e.g. Sagan~\cite{Sagan2001} or Fulton-Harris~\cite{FultonHarris1991}. Let
\begin{equation}
    \Omega_{n,d}
    :=
    \left\{
        \lambda=(\lambda_1,\ldots,\lambda_d)\in\mathbb Z_{\geq0}^d:
        \lambda_1\geq\cdots\geq\lambda_d\geq0,\
        \sum_{i=1}^d\lambda_i=n
    \right\}
\end{equation}
be the set of partitions of $n$ with at most $d$ parts. We identify a partition $\lambda\in\Omega_{n,d}$ with its Young diagram on $n$ boxes and at most $d$ rows. For $\lambda \in \Omega_{n,d}$, we use $\lambda^- \subset \Omega_{n-1,d}$ to denote the set
\begin{equation}
    \lambda^- = \{\nu \in \Omega_{n-1,d} : \nu\ \text{is attained from $\lambda$ by removing a box}\}
\end{equation}
For $\nu\in\Omega_{n-1,d}$, let $\nu^+$ denote the set of valid $d$-row diagrams obtained by adding one box. We write $\lambda\overset{\nu}{\sim}\mu$ when $\nu\in\lambda^-\cap\mu^-$, and $\lambda\sim\mu$ when such a $\nu$ exists. For distinct $\lambda,\mu \in \Omega_{n,d}$, the common subdiagram is unique.

Define the diagonal $SU(d)$ action and tensor-permutation action of $S_n$ on $\calH$ by
\begin{equation}
    Q(U)=U^{\otimes n},
    \qquad
    P(\sigma)\ket{i_1,\ldots,i_n}
    =
    \ket{i_{\sigma^{-1}(1)},\ldots,i_{\sigma^{-1}(n)}}. 
\end{equation}
These actions define representations (see \cref{app:rep-theory-background}) of $SU(d)$ and $S_n$. A tremendously important fact is that these representations can be decomposed into irreducible representations -- also referred to as irreps -- which are labeled by partitions $\lambda \in \Omega_{n,d}$. These representations commute, giving us Schur-Weyl duality.

\begin{theorem}[Schur--Weyl duality {\cite{Sagan2001,FultonHarris1991}}]
\label{thm:schur-weyl}
As an $SU(d)\times S_n$-module,
\begin{equation}
    \calH
    \cong
    \bigoplus_{\lambda\in\Omega_{n,d}}
    V^\lambda\otimes W^\lambda,
\end{equation}
where $W^\lambda$ is the irreducible Specht module of $S_n$ indexed by $\lambda$, and $V^\lambda$ is the irreducible polynomial $U(d)$-module indexed by $\lambda$, restricted to $SU(d)$. The irreducible $SU(d)\times S_n$-modules $V^\lambda\otimes W^\lambda$ are pairwise inequivalent, and each occurs with multiplicity one. We denote by $\Pi_\lambda$ the orthogonal projector onto the corresponding summand $V^\lambda\otimes W^\lambda$. 
\end{theorem}

In the Schur--Weyl basis, these group actions admit the following block diagonalizations:
\begin{align}
    Q(U)
    &\cong
    \bigoplus_{\lambda}Q_\lambda(U)\otimes I_{W^\lambda},
    &
    P(\sigma)
    &\cong
    \bigoplus_{\lambda}I_{V^\lambda}\otimes P_\lambda(\sigma).
\end{align}

\begin{proposition}[Interchange Hamiltonian in the Schur--Weyl basis]
\label{prop:interchange-schur-weyl-form}
The actions $Q$ and $P$ commute, and the Hamiltonian $H$ commutes with both. Moreover,
\begin{equation}
    H
    =
    \sum_{\lambda\in\Omega_{n,d}}E_\lambda\Pi_\lambda,
\end{equation}
where the energy eigenvalue $E_\lambda$ has the value
\begin{align}
    E_\lambda
    =
    -\frac1n\sum_{b\in\lambda}c(b)
    =
    \frac{1}{2n} \left( -\sum_{i=1}^d\lambda_i^2 + \sum_{i=1}^d(2i-1)\lambda_i \right).
\end{align}
Here $c(b)=j-i$ is the content of a box $b$ in row $i$ and column $j$. If several diagrams have the same energy, the corresponding energy projector is
\begin{equation}
    \Pi_E
    =
    \sum_{\lambda:E_\lambda=E}\Pi_\lambda.
\end{equation}
\end{proposition}

\begin{proof}
The covariance relations
\begin{equation}
    Q(U)T_{ij}Q(U)^\dagger=T_{ij},
    \qquad
    P(\sigma)T_{ij}P(\sigma)^\dagger=T_{\sigma(i)\sigma(j)}
\end{equation}
show that $H$ commutes with both actions. Furthermore, $\sum_{i<j}T_{ij}$ is the class sum of transpositions and is central in $\mathbb C[S_n]$. By Schur's Lemma, it therefore acts as a scalar on each $W^\lambda$. The scalar is the content sum $\sum_{b\in\lambda}c(b)$; see, for example, \cite{Sagan2001}. The displayed formula is computed explicitly in
\cref{lem:energy-eigenvalues}.
\end{proof}

For $d=2$, the energy eigenspaces of $H$ are labeled by the total spin $j$, or equivalently by two-row Young diagrams
\begin{equation}
    \lambda
    =
    \left(\frac n2+j,\frac n2-j\right),
    \qquad
    j=\frac{\lambda_1-\lambda_2}{2},
    \qquad
    E_j = -\frac{j(j+1)}n-\frac{n-4}{4}
\end{equation}
where $E_j$ is the corresponding energy. In particular, the fully symmetric diagram $(n)$ is the maximal-spin sector $j=n/2$ and is the ground-state sector. For general $d$, Young diagrams similarly label the collective $SU(d)$ representation sectors and therefore play the role of generalized total-spin quantum numbers.

\begin{remark}
For $d=2$, the map $\lambda\mapsto E_\lambda$ is injective on $\Omega_{n,2}$; for $d\geq3$, it need not be. When the map is injective, a Young diagram $\lambda$ and its corresponding energy eigenvalue $E_\lambda$ are perfectly interchangeable.
\end{remark}

\begin{lemma}[Jointly commutative observable algebra]\label{lem:commutant-is-span-projectors}
The joint commutant of both the $SU(d)$ and $S_n$ actions is exactly the set of observables that are scalar on each Schur--Weyl block. Equivalently,
\begin{equation}
    \mathrm{comm}(SU(d)\times S_n)
    =
    \mathrm{span}\{\Pi_\lambda:\lambda\in\Omega_{n,d}\}.
\end{equation}
\end{lemma}

\begin{proof}
Schur--Weyl duality gives $\calH \cong \bigoplus_{\lambda\in\Omega_{n,d}}V^\lambda\otimes W^\lambda$.
Each tensor product is irreducible for the product group, and distinct summands are inequivalent. For $X\in\mathrm{comm}(SU(d)\times S_n)$, write
\begin{equation}
    X=\sum_{\lambda,\mu\in\Omega_{n,d}}\Pi_\mu X\Pi_\lambda.
\end{equation}
Each block $\Pi_\mu X\Pi_\lambda$ is $SU(d) \times S_n$-equivariant. Schur's lemma makes it zero when $\mu\neq\lambda$ and a scalar multiple of the identity when $\mu=\lambda$. Thus
\begin{equation}
    X=\sum_{\lambda\in\Omega_{n,d}}c_\lambda\Pi_\lambda.
\end{equation}
Conversely, every $\Pi_\lambda$ commutes with the joint action.
\end{proof}

\subsection{Davies Generator}

For a finite dimensional Hilbert space $\calH$ and a system Hamiltonian $H$, the Gibbs state at inverse temperature $\beta>0$ is given by
\begin{equation}
    \rho
    :=
    Z_\beta^{-1}e^{-\beta H},
    \qquad
    Z_\beta:=\Tr(e^{-\beta H}).
\end{equation}
We equip $\calB(\calH)$ with the KMS inner product $\langle X,Y\rangle_\rho := \Tr\left(\rho^{1/2}X^\dagger\rho^{1/2}Y\right)$ and norm $\Vert X \Vert_\rho^2:=\langle X,X\rangle_\rho$. The Davies generators considered below are Lindbladians $\calL: \calB(\calH) \to \calB(\calH)$ that satisfy quantum detailed balance with respect to the Gibbs state $\rho$, i.e. $\langle X, \calL(Y) \rangle_\rho = \langle \calL(X), Y \rangle_\rho$ for all $X,Y \in \calB(\calH)$. A Lindbladian which is self-adjoint with respect to a KMS inner product is also said to be reversible. The reversible Lindbladians we consider will have eigenvalue 0 and all other eigenvalues will be negative; we denote by $\gap(\calL)$ the smallest non-zero eigenvalue of $-\calL$.  If $\calV\subseteq\calB(\calH)$ is a subspace, we write $\calV^{\perp_{\rho}}$ to denote the KMS orthogonal complement of $\calV$.

Write $\mathrm{Bohr}(H) = \{\omega = E_+ - E_-: E_+,E_- \in \mathrm{spec}(H)\}$ for the set of Bohr frequencies of $H$. For an operator $A\in\calB(\calH)$, its Bohr-frequency decomposition is given by
\begin{equation}
    A = \sum_{\omega \in \mathrm{Bohr}(H)} A(\omega), \qquad A(\omega) = \sum_{E_+-E_- = \omega} \Pi_{E_+} A \Pi_{E_-}.
\end{equation}
Observe that this decomposition implies $\rho^{1/2}A(\omega)=e^{-\beta\omega/2}A(\omega)\rho^{1/2}$. Moreover, distinct Bohr sectors are KMS-orthogonal:
\begin{equation}
    \langle X(\omega),Y(\nu)\rangle_\rho = 0
    \qquad
    \text{whenever }\omega\neq\nu. 
\end{equation}
We also will fix a weight function $\gamma:\R\to\R_{\geq0}$ satisfying the classical detailed-balance condition $\gamma(\omega)=e^{-\beta\omega}\gamma(-\omega)$. Unless stated otherwise, we assume the Metropolis weight $\gamma_M(\omega)=\min\{1,e^{-\beta\omega}\}$. The weight function controls the rate at which the environment interactions change the system energy by each Bohr frequency $\omega$.

\begin{definition}[Davies generator]\label{def:davies-generator}
    For a finite family of Hermitian operators $\calJ\subset\calB(\calH)$ define $\calL:\calB(\calH)\to\calB(\calH)$ by
    \begin{equation}
        \calL(X) = \sum_{\omega \in \Bohr(H)} \gamma(\omega) \sum_{A \in \calJ}\left[A(\omega)^\dagger X A(\omega) - \frac{1}{2} \left\{A(\omega)^\dagger A(\omega),X\right\} \right].
    \end{equation}
    We call the operators in $\calJ$ jump operators. $\calL_{\loc}$ is defined by the set $\calJ_{\loc}$ given in~\cref{def:Jloc}.
\end{definition}

Since the jumps are Hermitian, $A(\omega)^\dagger=A(-\omega)$. Pairing the $\omega$ and $-\omega$ terms using detailed balance gives the Dirichlet form in terms of commutators~\cite{basso2025quantumclassicalcomparisondavies}:
\begin{equation}\label{eq:dirichlet-form-commutator}
    \mathcal E_{\calJ}(X)
    :=\langle X,-\calL(X)\rangle_\rho
    =\frac12\sum_{A\in\calJ}\sum_{\omega\in\Bohr(H)}
    h_\omega\bigl\Vert[A(\omega),X]\bigr\Vert_\rho^2,
    \qquad h_\omega:=e^{\beta\omega/2}\gamma(\omega).
\end{equation}
For Metropolis weights, $h_\omega=e^{-\beta|\omega|/2}>0$. This identity shows that $\calL$ is KMS self-adjoint and negative semidefinite, and that its kernel consists of the operators commuting with every $A(\omega)$ for which $h_\omega>0$. In particular, if all jumps commute with $H$ and $\gamma(0)>0$, then
\begin{equation}\label{eq:zero-frequency-davies}
    \calL(X)=-\frac{\gamma(0)}2\sum_{A\in\calJ}[A,[A,X]],
    \qquad
    \ker(\calL)=\{X:[A,X]=0\text{ for every }A\in\calJ\}.
\end{equation}
We will use this observation for the group mixers below.

It will often be useful to recast the Davies generator in terms of its associated twirl.
\begin{definition} The twirl operator $\calT_{\calJ}:\calB(\calH)\to \calB(\calH)$ associated to a set of jump operators $\calJ$ is given by
    \begin{equation}
        \calT_\calJ (X) = \sum_{A \in \calJ} A^\dagger XA.
    \end{equation}
\end{definition}
Observe that the Davies generator can be expressed in terms of the twirl $\calT_\calJ$ as
\begin{equation}\label{eq:davies-twirl}
    \calL(X) = \sum_{\omega \in \mathrm{Bohr}(H)} \gamma(\omega) \left[
    \sum_{\substack{E_+-E_-=\omega\\F_+-F_-=\omega}}
    \Pi_{E_-}\calT_\calJ(\Pi_{E_+}X\Pi_{F_+})\Pi_{F_-} 
    -\frac12\left\{
    \sum_{E_+-E_-=\omega}
    \Pi_{E_-}\calT_\calJ(\Pi_{E_+})\Pi_{E_-},X
    \right\} \right]
\end{equation}

In this work, we use local jumps given by the normalized single-site Gell--Mann matrices. We use the physicists' convention in which $\mathfrak{su}(d)$ is represented by traceless Hermitian generators $\{\lambda^\alpha\}_{\alpha=1}^{d^2-1}$, normalized by
\begin{equation}
    \Tr(\lambda^\alpha\lambda^\beta)=2\delta_{\alpha\beta}.
\end{equation}
For $\alpha = 1, \dots, d^2-1$ define $J^\alpha := \frac{1}{\sqrt{2d}}\lambda^\alpha$, and set $J^0:=I/d$. Then $\{J^\alpha\}_{\alpha=0}^{d^2-1}$ is an orthogonal basis of $\calB(\C^d)$ with respect to the Hilbert--Schmidt inner product. The local Davies generator studied in this paper is the Davies generator associated with the traceless jump set
\begin{equation} \label{def:Jloc}
    \Jloc := \{J_i^\alpha:i\in[n],\ \alpha=1,\ldots,d^2-1\}
\end{equation}
and is denoted by $\Lloc$. Equivalently, the jumps are the single-site Gell--Mann matrices, normalized by the factor $(2d)^{-1/2}$. This normalization only fixes the time scale: using the unnormalized Gell--Mann matrices instead would multiply the generator by the constant $2d$.

For several calculations of $\Lloc$, it is convenient to adjoin the identity element $J_i^0=I/d$ to the traceless jump set $\calJ_\loc$. This is an algebraic device only: $J_i^0$ has the single Bohr frequency component $J_i^0(0)=I/d$, and its contribution to the Lindbladian vanishes since for any operator $X$ 
\begin{equation}
    \frac1{d^2}X
    -
    \frac12\left\{\frac1{d^2}I,X\right\}
    =0.
\end{equation}
We therefore define the auxiliary \textit{Lie algebra twirl}
\begin{equation}\label{eq:lie-algebra-twirl}
    \calT(X)
    :=
    \sum_{i=1}^n\sum_{\alpha=0}^{d^2-1}
    (J_i^\alpha)^\dagger XJ_i^\alpha
\end{equation}
and note that since $J_i^0$ does not contribute to the Davies generator, we can express $\Lloc$ in terms of the Lie algebra twirl $\calT$ in the form of \cref{eq:davies-twirl}.

We write $\mathcal E_{\loc}(X):=\langle X,-\Lloc(X)\rangle_\rho$ for the local Dirichlet form, given by \cref{eq:dirichlet-form-commutator} with $\calJ=\Jloc$.

\begin{proposition}[Primitivity]\label{prop:local-primitivity}
    For every finite $n,d$ and $\beta>0$, the kernel of $\Lloc$ is
    $\mathbb C I$. Hence the Gibbs state is the unique full-rank stationary
    state of the Schr\"odinger-picture semigroup.
\end{proposition}

\begin{proof}
    If $X\in\ker(\Lloc)$, then \cref{eq:dirichlet-form-commutator} and $h_\omega>0$ imply $[J_i^\alpha(\omega),X]=0$ for every $i,\alpha,\omega$. Summing over $\omega$ gives $[J_i^\alpha,X]=0$ for all $i,\alpha$. Moreover, the Gell--Mann matrices (together with the identity) span the full matrix algebra of single site operators. So, an operator commuting with all the $J_i^\alpha$ lies in the center of $\calB((\C^d)^{\otimes n})$ and is thus a scalar multiple of the identity. The converse is immediate.
\end{proof}

Let $t_{\mathrm{mix}}(\eps)$ denote the worst-case trace-distance mixing time of the Schr\"odinger-picture semigroup, using trace distance $\frac12\Vert\sigma-\rho\Vert_1$. For a primitive KMS-self-adjoint Lindbladian and $0<\eps<1/2$, the gap is related to the mixing time via
\begin{equation}\label{eq:gap_and_mixing_time}
    \frac{1}{\gap(\mathcal{L})}
    \log\left(\frac{1}{2 \eps}\right)
    \leq
    t_{\mathrm{mix}}(\eps)
    \leq
    \frac{1}{2\gap(\mathcal L)}
    \log\left(\frac{1}{\lambda_{\min}(\rho)\eps^2}\right).
\end{equation}
The upper bound follows from~\cite[Theorem~22]{KastoryanoTemme2013}. For the lower bound, choose a Hermitian eigenobservable $X$ with $\mathcal L(X)=-\gap(\mathcal L)X$ and $\Vert X\Vert_\infty=1$, so $\Tr(\rho X)=0$. Starting from a pure state on an eigenvector of $X$ with eigenvalue of absolute value one, trace-norm duality gives trace distance at least $\frac12e^{-t\gap(\mathcal L)}$ from $\rho$ at time $t$.
For the interchange Hamiltonian \cref{eq:introduction-hamiltonian}, $\log \lambda_{\min}(\rho)^{-1} = O_{\beta,d}(n)$. Since $\dim(\calH) = d^n$ and $E_{\max}(H)-E_{\min}(H) \leq 2\Vert H\Vert \leq n-1$, we have
\begin{equation}\label{eq:gibbs-min-eigenvalue-bound}
    \log\lambda_{\min}(\rho)^{-1}
    \leq n\log d+\beta\bigl(E_{\max}(H)-E_{\min}(H)\bigr)
    \leq n\log d+\beta(n-1) = O_{\beta,d}(n).
\end{equation}

This makes the spectral gap the central dynamical quantity in our analysis, since its scaling distinguishes rapidly mixing regimes from exponentially slow ones. In other words, estimates of $\gap(\Lloc)$ translate directly, up to a factor linear in $n$, into estimates on the mixing time.

\subsection{Symmetry of the Lindbladian}

\begin{definition}[Observable representation]\label{def:observable-representation}
    Let $R:G\to U(\calH)$ be a unitary representation. Its \textbf{observable representation}, also called the conjugation representation, is the representation on $\calB(\calH)$ defined by
    \begin{equation}
        \calR(g)(X):=R(g)XR(g)^\dagger,
        \qquad g\in G,\quad X\in\calB(\calH).
    \end{equation}
\end{definition}

The fixed point operators of the observable representation are exactly those in $\mathrm{comm}(G)$. Here we apply this definition to $R(U,\sigma)=Q(U)P(\sigma)$ and its restrictions to $SU(d)$ and $S_n$. Since the interchange Hamiltonian $H$ commutes with both the on-site $SU(d)$ action and the permutation action of $S_n$, the Gibbs state $\rho\propto e^{-\beta H}$ commutes with both actions as well. As a consequence, the corresponding observable representations are unitary for the KMS inner product, i.e. for all $X\in \calB(\calH)$ and all $U\in SU(d), \sigma\in S_n$,
\begin{equation}
    \Vert X\Vert_{\rho} = \Vert Q(U)X Q(U)^\dagger\Vert_\rho = \Vert P(\sigma) X P(\sigma)^\dagger\Vert_\rho.  
\end{equation}
Analogously to the $d=2$ case~\cite[Thm. 3.1] {BassoBergamaschiLinRagoneStubbs2026}, we may show that the Davies generator $\calL_{\loc}$ likewise enjoys this symmetry.

\begin{lemma}\label{lem:lie-algebra-twirl-equivariant}
    The Lie algebra twirl $\calT$ is equivariant with respect to the joint
    action of $SU(d)\times S_n$. That is, for every $(U,\sigma) \in SU(d) \times S_n$ and $X\in\calB(\calH)$,
    \begin{equation}
        \calT\left(Q(U)P(\sigma)X P(\sigma)^\dagger Q(U)^\dagger\right) = Q(U)P(\sigma)\calT(X)P(\sigma)^\dagger Q(U)^\dagger.
    \end{equation}
\end{lemma}
\begin{proof}
    The completeness relation for the Hilbert--Schmidt orthogonal basis $\{J^\alpha\}_{\alpha=0}^{d^2-1}$ gives us a useful expression for the twirl in terms of the partial trace
    \begin{equation}\label{eq:lie-algebra-twirl-partial-trace}
        \calT(X)
        =
        \frac1d\sum_{i=1}^n\Tr_i(X)\otimes I_i.
    \end{equation}
    Thus, the claim is immediate from equivariance of the partial trace. Indeed, for the $SU(d)$-action,
    \begin{equation}
        \Tr_i\left(U^{\otimes n}X(U^\dagger)^{\otimes n}\right) = U^{\otimes(n-1)}\Tr_i(X)(U^\dagger)^{\otimes(n-1)}.
    \end{equation}
    Therefore
    \begin{equation}
        \Tr_i\left(U^{\otimes n}X(U^\dagger)^{\otimes n}\right)\otimes I_i = U^{\otimes n}\left(\Tr_i(X)\otimes I_i\right)(U^\dagger)^{\otimes n}.
    \end{equation}
    Summing over $i$ gives equivariance under $SU(d)$. For the permutation action, we have
    \begin{equation}
        P(\sigma)\left(\Tr_i(X)\otimes I_i\right)P(\sigma)^\dagger = \Tr_{\sigma(i)}\left(P(\sigma)X P(\sigma)^\dagger\right) \otimes I^{(\sigma(i))}.
    \end{equation}
    Since $i\mapsto\sigma(i)$ is a permutation of $[n]$, summing over $i$ gives equivariance under $S_n$. Combining the two proves the claim.
\end{proof}

\begin{corollary}\label{cor:local-davies-equivariant}
    The local Davies generator is equivariant under conjugation by $SU(d)\times S_n$.
\end{corollary}

\begin{proof}
    By \cref{lem:lie-algebra-twirl-equivariant}, $\calT$ is equivariant under conjugation by $SU(d)\times S_n$. Moreover, each spectral projector $\Pi_E$ commutes with the corresponding unitary representations. Consequently, composing $\calT$ with left and right multiplication by spectral projectors preserves equivariance. The same is true of the anticommutator terms in \cref{eq:davies-twirl}. Thus every summand in $\Lloc$ is equivariant, and hence so is $\Lloc$.
\end{proof}

\section{Main Results}\label{sec:main-results}

We first state the invariant decomposition of the observable algebra, then give the result for each summand and combine them into the full phase diagram.

\subsection{Symmetry Decomposition and Non-Symmetric Sector Bounds}\label{subsec:symmetry-sector-decomposition}

\subsubsection{Decomposition into Symmetry Sectors}

First, $\calB(\calH)$ splits into the $S_n$-invariant subalgebra and its KMS-orthogonal complement:
\begin{equation}
    \calB(\calH)
    =
    \mathrm{comm}(S_n)
    \oplus_\rho
    \mathrm{comm}(S_n)^{\perp_\rho}.
\end{equation}
Inside $\mathrm{comm}(S_n)$, set
\begin{equation}
    \calA_d
    :=
    \mathrm{comm}(S_n)\cap\mathrm{comm}(SU(d))
    =
    \mathrm{comm}(SU(d)\times S_n).
\end{equation}
By \cref{lem:commutant-is-span-projectors},
\begin{equation}
    \calA_d
    =
    \mathrm{span}\{\Pi_\lambda:\lambda\in\Omega_{n,d}\}.
\end{equation}
This is the commutative Young diagram sector. Since the conjugation representations of $S_n$ and $SU(d)$ on $\calB(\calH)$ are unitary with respect to the KMS inner product, complete reducibility implies that the $SU(d)$-invariant subspace of $\mathrm{comm}(S_n)$ admits a KMS-orthogonal complement, see e.g.~\cite{hall2013lie}. We denote this subspace by
\begin{equation}
    \calK_d
    :=
    \mathrm{comm}(S_n)\cap\mathrm{comm}(SU(d))^{\perp_\rho}.
\end{equation}
Thus, we have the full decomposition
\begin{equation}
    \calB(\calH)
    =
    \calA_d \oplus_\rho \calK_d \oplus_\rho \mathrm{comm}(S_n)^{\perp_\rho}.
\end{equation}

\begin{proposition}\label{prop:invariant-symmetry-sectors}
    The three summands above are invariant under $\Lloc$. Consequently,
    \begin{equation}
        \gap(\Lloc)
        =
        \min\left\{
            \gap(\Lloc\vert_{\calA_d}),
            \lambda_{\min}(-\Lloc\vert_{\calK_d}),
            \lambda_{\min}(-\Lloc\vert_{\mathrm{comm}(S_n)^{\perp_\rho}})
        \right\}.
    \end{equation}
\end{proposition}
\begin{proof}
    By $S_n$- and $SU(d)$-equivariance, $\Lloc$ preserves $\mathrm{comm}(S_n)$ and $\mathrm{comm}(SU(d))$, respectively. Hence it preserves their intersection $\calA_d$.
    Since $\Lloc$ is KMS self-adjoint, the KMS-orthogonal complement of any invariant subspace is also invariant. Thus $\mathrm{comm}(SU(d))^{\perp_\rho}$ and $\mathrm{comm}(S_n)^{\perp_\rho}$ are invariant, and consequently so is $\calK_d = \mathrm{comm}(S_n)\cap\mathrm{comm}(SU(d))^{\perp_\rho}$. Therefore $-\Lloc$ is block diagonal with respect to
    \begin{equation}
        \calB(\calH)
        =
        \calA_d
        \oplus_\rho
        \calK_d
        \oplus_\rho
        \mathrm{comm}(S_n)^{\perp_\rho}.
    \end{equation}
    Finally, $\ker(\Lloc)=\mathbb CI\subseteq\calA_d$, so the smallest non-zero eigenvalue of $-\Lloc$ is the minimum of the spectral gap on $\calA_d$ and the smallest eigenvalues on the other two summands.
\end{proof}

As such, the analysis of $\gap(\Lloc)$ splits into three independent problems:
\begin{enumerate}
    \item the Young diagram sector $\calA_d$,
    \item the non-trivial $SU(d)$ (order parameter) sector $\calK_d$,
    \item and the permutation-non-symmetric sector $\mathrm{comm}(S_n)^{\perp_\rho}$.
\end{enumerate}

The three summands have distinct dynamical interpretations. The Young diagram sector $\calA_d$ records only the Schur--Weyl label, equivalently the normalized Young diagram shape $x=\lambda/n$, and carries the spectral free energy landscape. The sector $\calK_d$ contains permutation-invariant observables that transform non-trivially under global $SU(d)$ rotations and therefore carries the collective order parameter. At low temperatures $\beta > \beta_c(d)$, these observables distinguish orientations with the same Young diagram states. Finally, $\mathrm{comm}(S_n)^{\perp_\rho}$ contains the permutation-non-symmetric observables.

\subsubsection{Bounds on \texorpdfstring{$\calK_d$}{} and \texorpdfstring{$\mathrm{comm}(S_n)^{\perp_\rho}$}{}}

The two non-commutative sectors are controlled by comparison with auxiliary generators that mix the corresponding symmetry actions. Their gap bounds are summarized in the following theorem.

\begin{theorem}[Gap bounds on the symmetry-controlled sectors]
\label{thm:symmetry-controlled-sector-bounds}
    There exist constants $C_1,C_2>0$, depending only on $d$ such that the following hold for every $n\geq2$.
    
    \begin{enumerate}[label=(\roman*)]
        \item On the permutation-invariant, non-trivial $SU(d)$ sector $\calK_d$,
        \begin{equation}\label{eq:main-su-sector-lower-bound}
            -\Lloc\big\vert_{\calK_d}
            \geq
            \frac1n\idty\big\vert_{\calK_d}.
        \end{equation}
        Equivalently,
        \begin{equation}
            \lambda_{\min}\left(-\Lloc\big\vert_{\calK_d}\right)
            \geq
            \frac1n.
        \end{equation}
    
        \item On the permutation-non-symmetric sector $\mathrm{comm}(S_n)^{\perp_\rho}$,
        \begin{equation}\label{eq:main-Sn-sector-lower-bound}
            -\Lloc
            \big\vert_{\mathrm{comm}(S_n)^{\perp_\rho}}
            \geq
            \frac{1}{C_1e^{C_2\beta}}
            \idty
            \big\vert_{\mathrm{comm}(S_n)^{\perp_\rho}}.
        \end{equation}
        In particular, for fixed $d$ and $\beta$,
        \begin{equation}
            \lambda_{\min}\left(
                -\Lloc
                \big\vert_{\mathrm{comm}(S_n)^{\perp_\rho}}
            \right)
            =
            \Omega_{\beta,d}(1).
        \end{equation}
    \end{enumerate}
\end{theorem}

\begin{proof}[Proof sketch]
    For part~(i), set $\Lambda_\alpha^{\tot}:=\sum_{i=1}^n(\lambda^\alpha)^{(i)}$ and introduce the normalized global $SU(d)$ mixer
    \begin{equation}
        \calL_{SU(d)}(O)
        =
        -\frac{1}{4d}
        \sum_{\alpha=1}^{d^2-1}
        [\Lambda_\alpha^{\tot},[\Lambda_\alpha^{\tot},O]].
    \end{equation}
    Its kernel is $\mathrm{comm}(SU(d))$, and its spectral gap on $\mathrm{comm}(SU(d))^{\perp_\rho}$ is exactly one. A change of basis among the single-site jumps gives a decomposition
    \begin{equation}
        \Lloc
        =
        \frac1n\calL_{SU(d)}
        +
        \Lrest,
    \end{equation}
    where $\Lrest$ is another negative-semidefinite Davies generator. Hence
    \begin{equation}
        -\Lloc
        \geq
        -\frac1n\calL_{SU(d)},
    \end{equation}
    which proves~\cref{eq:main-su-sector-lower-bound}. The details are given in \cref{app:sud-sector}.
    
    For part~(ii), we introduce the normalized $S_n$ group mixer
    \begin{equation}
        \calL_{S_n}(O)
        :=
        \frac2n
        \sum_{1\leq u<v\leq n}
        \left(T_{uv}OT_{uv}-O\right).
    \end{equation}
    Its kernel is $\mathrm{comm}(S_n)$ and its spectral gap is exactly two. We then prove the uniform comparison
    \begin{equation}
        \langle O,-\calL_{S_n}(O)\rangle_\rho
        \leq
        C_1e^{C_2\beta}
        \langle O,-\Lloc(O)\rangle_\rho.
    \end{equation}
    Combining these two statements gives \cref{eq:main-Sn-sector-lower-bound}. The complete argument appears in \cref{app:Sn-sector}.
\end{proof}

Thus neither of the two symmetry-controlled sectors can produce an exponentially small gap at fixed $d$ and $\beta$. Any exponential slowdown must come from the Young diagram sector $\calA_d$.

\subsection{The Young Diagram Sector \texorpdfstring{$\calA_d$}{A d}}\label{subsec:young-diagram-sector}

The restriction of the Davies dynamics of $\Lloc$ to $\calA_d$ gives rise to a classical Markov chain called the master equation. This is a consequence of the following general result (cf.~\cite[Theorem~4.2]{carlen2017gradient}).

\begin{proposition}\label{prop:commutative-subalgebra-markov}
    Let $\calA \subset \calB(\calH)$ be a finite-dimensional commutative $*$-subalgebra, and suppose
    \begin{equation}
        \calL(\calA)\subseteq \calA,
    \end{equation}
    where $\calL$ generates a unital completely positive semigroup $\{e^{t\calL}\}_{t\geq0}$ and is KMS self-adjoint with respect to a unique full-rank stationary state $\rho$. Moreover, suppose that
    \begin{equation}
        [\rho,A] = 0
        \qquad
        \text{for every }A\in\calA.
    \end{equation}
    Then $\calL\vert_{\calA}$ gives rise to a reversible continuous-time Markov chain on the spectrum of $\calA$.
\end{proposition}

In the setting of \cref{prop:commutative-subalgebra-markov}, let $\{\Pi_1,\dots,\Pi_k\}$ be the nonzero, mutually orthogonal minimal projections of $\calA$, with $\sum_i\Pi_i=I$. We let $L \in \R^{k\times k}$ denote the generator matrix for the Markov chain $\calL\vert_{\calA}$. Its coefficients are given by
\begin{equation}
    L_{i\to j} = \frac{\Tr(\Pi_i\calL(\Pi_j))}{\Tr(\Pi_i)}.
\end{equation}

\begin{corollary}\label{A_d-is-markov}
    The Davies generator $\Lloc$ restricted to $\calA_d = \mathrm{span}\{\Pi_\lambda\}_{\lambda \in \Omega_{n,d}}$ gives rise to a classical reversible Markov chain on $\Omega_{n,d}$.
\end{corollary}

Our first contribution is a representation-theoretic approach to finding the exact master equation on the Young diagram sector. This work is collected in \cref{appdx:young-diagram-master-proof}. The master equation governs transitions between subspaces labeled by Young diagrams where the allowed transitions correspond exactly to one-box moves. In such moves, a corner box is removed and then added back so that the resulting diagram is a legal Young diagram. This induces a natural graph structure on $\Omega_{n,d}$, with the adjacency relation defined in \cref{subsec:schur-weyl}: two distinct Young diagrams $\lambda$ and $\mu$ are connected by an edge if and only if $\lambda \sim \mu$.

\begin{center}
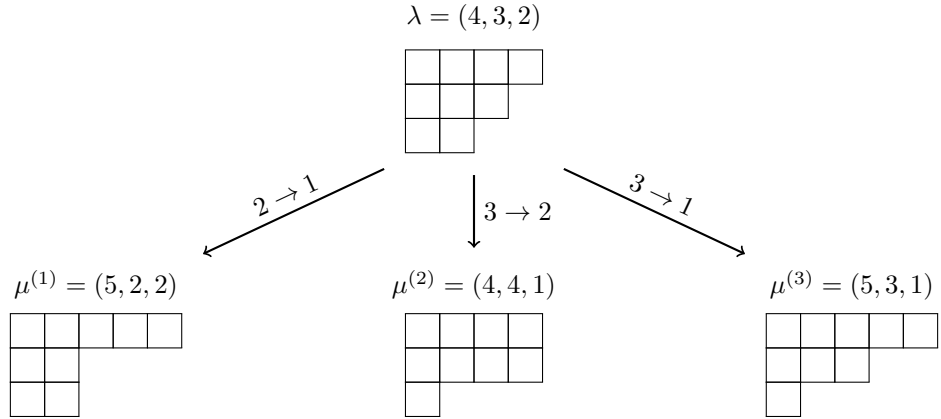

\begin{minipage}{0.95\textwidth}
    \centering

    \begin{tikzpicture}[
        diagram/.style={
            align=center,
            inner sep=4pt
        },
        move/.style={
            ->,
            thick,
            shorten >=4pt,
            shorten <=4pt
        }
    ]

        % Initial diagram
        \node[diagram] (lambda) at (0,3.5) {
            $\lambda=(4,3,2)$\\[0.6em]
            \ydiagram{4,3,2}
        };

        % Resulting diagrams
        \node[diagram] (muone) at (-5,0) {
            $\mu^{(1)}=(5,2,2)$\\[0.4em]
            \ydiagram{5,2,2}
        };

        \node[diagram] (mutwo) at (0,0) {
            $\mu^{(2)}=(4,4,1)$\\[0.4em]
            \ydiagram{4,4,1}
        };

        \node[diagram] (muthree) at (5,0) {
            $\mu^{(3)}=(5,3,1)$\\[0.4em]
            \ydiagram{5,3,1}
        };

        % Downward arrows
        \draw[move]
            (lambda.south west) --
            node[midway,above,sloped] {$2\to1$}
            (muone.north east);

        \draw[move]
            (lambda.south) --
            node[midway,right] {$3\to2$}
            (mutwo.north);

        \draw[move]
            (lambda.south east) --
            node[midway,above,sloped] {$3\to1$}
            (muthree.north west);

    \end{tikzpicture}

    \captionof{figure}{Here we depict three legal one-box moves from $\lambda=(4,3,2)$. The label $r\to s$ indicates that a corner box is moved from row $r$ to row $s$. Each resulting diagram $\mu$ shares its own eight-box subdiagram with $\lambda$ in which the box in row $r$ has been removed.}
    \label{fig:young-diagram-one-box-moves}
\end{minipage}
\end{center}

\begin{theorem}[Young diagram Markov chain generator]
\label{thm:master-eqn-coefficients}
The induced generator
$L\in\mathbb R^{\vert\Omega_{n,d}\vert\times\vert\Omega_{n,d}\vert}$ has entries
\begin{equation}
    L_{\lambda \to \mu}
    =
    \begin{cases}
        \gamma(E_\mu-E_\lambda)
        \frac nd
        \frac{\dim(V^\mu)}{\dim(V^\nu)}
        \frac{\dim(W^\nu)}{\dim(W^\lambda)}
        &\text{if}\ \lambda \overset{\nu}{\sim} \mu\ \text{and}\ \lambda\neq\mu,\\
        - \sum_{\eta \neq \lambda} L_{\lambda \to \eta} &\text{if}\ \lambda = \mu\\
        0 &\text{otherwise}
    \end{cases}
\end{equation}
where, in the first case, $\nu$ is the unique diagram in
$\lambda^-\cap\mu^-$.
\end{theorem}
\begin{proof}[Proof Sketch]
    The coefficients are given by
    \begin{equation}
        L_{\lambda \to \lambda'} = \frac{\Tr(\Pi_\lambda \Lloc(\Pi_{\lambda'}))}{\Tr(\Pi_\lambda)}.
    \end{equation}
    By \cref{eq:davies-twirl}, we can express this as a sum of terms of the form $\Tr(\Pi_\lambda \calT(\Pi_{\lambda'}))$. To determine $L_{\lambda \to \lambda'}$, we compute a representation-theoretic expression for the Lie algebra twirl of an isotypic projector, i.e. $\calT(\Pi_{\lambda'})$, in \cref{partialtraceprojector} and \cref{prop:schur-weyl-expansion-twirl}. With this expression in hand, we can easily derive all the coefficients of $L$. See \cref{appdx:young-diagram-master-proof} for the details.
\end{proof}

The $d=2$ case has a particularly nice physical interpretation. In this case, the 2-row Young diagrams $\lambda = (n-k,k)$, with $k=0,\ldots,\lfloor n/2\rfloor$, correspond exactly to spin sectors where the total spin is $s = \frac{n-2k}{2}$. The Wigner-Eckart Theorem tells us that the one-box moves correspond to changing the total spin by at most one: $s \to \{s-1,s,s+1\}$. Thus, the Markov chain generator $L$ becomes tridiagonal, or equivalently, a birth-death chain for which there are known results for controlling the mixing time. This is the setting studied in~\cite{BassoBergamaschiLinRagoneStubbs2026}.

\subsection{The Spectral Free Energy Landscape}

The stationary measure on Young diagrams is governed, to leading exponential order, by the free energy
\begin{equation}\label{eq:free-energy-1}
    F_\beta(x)
    =
    \frac{\beta}{2}\sum_{i=1}^d x_i^2
    -
    \sum_{i=1}^d x_i\log x_i,\qquad
    \text{for }
    x \in \Delta_d^\downarrow
    :=
    \left\{y\in\mathbb R^d: y_1\geq\cdots\geq y_d\geq0,\  \sum_{i=1}^d y_i=1 \right\}.
\end{equation}
Note that $\frac{1}{n}\Omega_{n,d} \subset \Delta_d^\down$. For $r>1$, define
\begin{equation}
    \beta_d(r)
    :=
    \frac{r+d-1}{r-1}\log r,\quad \text{and} \quad
    \betamin(d)
    :=
    \min_{r>1}\beta_d(r).
\end{equation}
as well as the critical inverse temperature
\begin{equation}
    \beta_c(d)
    :=
    \frac{2(d-1)}{d-2}\log(d-1).
\end{equation}
The points $\beta = \betamin(d)$ and $\beta = d$ are the spinodal thresholds for metastability, since, as we shall see in \cref{thm:full-dynamics}, these points define the boundaries of the slow-mixing region for dynamics, i.e. the metastable window. We also determine tight bounds on the temperature range for metastability.

\begin{proposition}\label{prop:betamin-asymptotics-and-inequalities}
    For $\betamin(d)$, $d \geq 3$ defined as above, we have
    \begin{equation}
        \betamin(d) = \log d + \log\log d + 1 + o_d(1).
    \end{equation}
    Moreover,
    \begin{equation}
        \betamin(d) < \beta_c(d) < d.
    \end{equation}
\end{proposition}
\begin{proof}
    See \cref{app:beta-min-asymptotics} for the asymptotic formula and \cref{subsec:proof-spectral-free-energy-landscape} for the inequalities.
\end{proof}

The following result is known in the literature~\cite{Bjornberg2016}; we provide a new proof that instead relies on our explicit computation for the stationary distribution -- and hence the free energy -- via the exact Markov chain matrix.

\begin{theorem}[Static Phase Diagram]\label{thm:spectral-free-energy-landscape}
    Assume $d\geq3$. Then, the global maximizers of $F_\beta$ on $\Delta_d^\downarrow$ are as follows.
    \begin{enumerate}[label=(\roman*)]
        \item If $0<\beta<\beta_c(d)$, then the unique global maximizer is the uniform distribution
        \begin{equation}
            u
            =
            \left(\frac1d,\ldots,\frac1d\right).
        \end{equation}
    
        \item If $\beta=\beta_c(d)$, then the two global maximizers are
        \begin{equation}
            u
            =
            \left(\frac1d,\ldots,\frac1d\right)\qquad \text{and}\qquad
            x_c
            =
            \left(
                \frac{d-1}{d},
                \frac{1}{d(d-1)},
                \ldots,
                \frac{1}{d(d-1)}
            \right).
        \end{equation}
    
        \item If $\beta>\beta_c(d)$, then the unique global maximizer is
        \begin{equation}
            x_\beta
            =
            \left(
                \frac{r_\beta}{r_\beta+d-1},
                \frac{1}{r_\beta+d-1},
                \ldots,
                \frac{1}{r_\beta+d-1}
            \right),
        \end{equation}
        where $r_\beta$ is the solution of $\beta=\beta_d(r)$ on the increasing branch of $\beta_d$. For $\beta\geq d$, $x_\beta$ is also the only local maximizer.
    \end{enumerate}
    Finally, the transition at $\beta_c(d)$ is first order: the largest coordinate of the global maximizer jumps discontinuously from $1/d$ to $(d-1)/d$.
\end{theorem}
\begin{proof}
    See \cref{subsec:proof-spectral-free-energy-landscape}.
\end{proof}

\begin{center}
\begin{minipage}{0.9\textwidth}
\centering

\includegraphics[width=.99\linewidth]{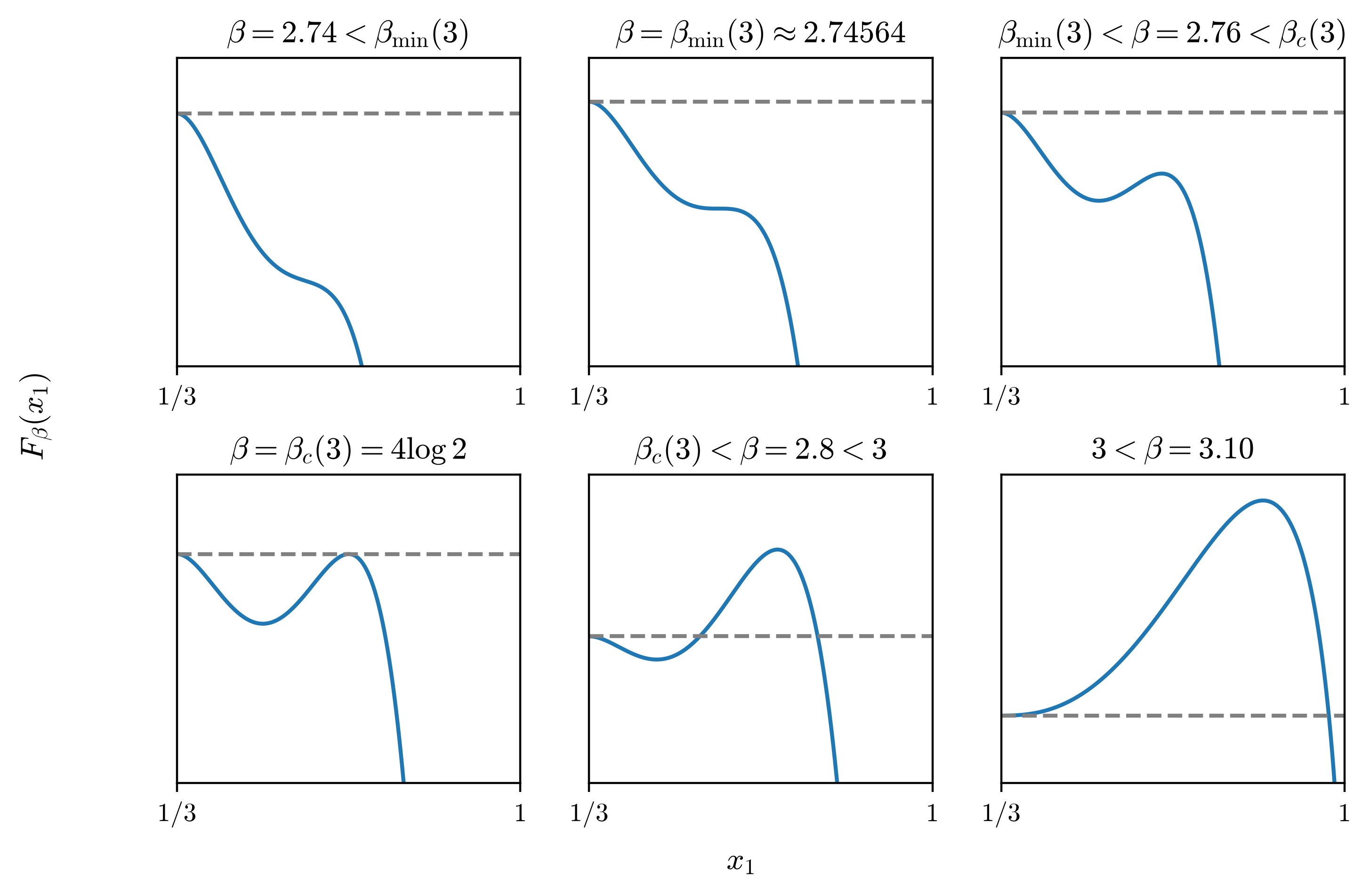}
\captionof{figure}{The free energy distribution for $d=3$ along the symmetry-breaking branch $(x_1,x_2,x_3) = \left(x_1, \frac{1-x_1}{2}, \frac{1-x_1}{2} \right)$. At $\betamin(3)$, a degenerate non-uniform critical point appears. Just above this threshold, it splits into a local maximum and a saddle. For $\betamin(3)<\beta<\beta_c(3)$, the uniform distribution remains the global maximizer while a non-uniform local maximizer is present. At $\beta_c(3)=4\log 2$, the uniform and non-uniform maxima have equal free energy. For $\beta_c(3)<\beta<3$, the non-uniform maximum is global and the uniform distribution remains locally stable. At $\beta=3$, the uniform distribution loses local stability. Dashed lines indicate the free energy of the uniform distribution. The panels use different vertical scales.}
\label{fig:free-energy-profiles-d3}

\end{minipage}
\end{center}

The next result summarizes the fast and slow mixing regimes of the commutative sector $\calA_d$. Here $L$ denotes the Young diagram generator from \cref{thm:master-eqn-coefficients}. The intuition for why the theorem holds is captured by \cref{fig:free-energy-profiles-d3}. Roughly speaking, we should expect this result because in the slow mixing window $(\betamin(d), d)$ the free energy admits an energy barrier separating two local maxima. We show that a neighborhood of one of the local maxima has exponentially small conductance. This gives an exponentially small upper bound on the gap of the classical Markov chain.

Outside of the metastable window the free energy admits only a single local maximum and thus has no metastable states to slow down the mixing. A conductance argument gives a polynomial lower bound on that gap, and hence a polynomial upper bound on the mixing time.

\begin{theorem}\label{thm:spectral-sector-gap-bounds}
Fix $d\geq3$ and use Metropolis weights.

\begin{enumerate}[label=(\roman*)]
    \item If $0<\beta<\betamin(d)$ or $\beta>d$, then there exists $C_{\beta,d}<\infty$ such that, for all sufficiently large $n$,
    \begin{equation}
        \gap(L)
        =
        \gap(\Lloc\vert_{\calA_d})
        \geq
        n^{-C_{\beta,d}}.
    \end{equation}

    \item If $\betamin(d)<\beta<d$, then there exist $0<c_{\beta,d}<C_{\beta,d}<\infty$ such that, for all sufficiently large $n$,
    \begin{equation}
        e^{-C_{\beta,d}n}
        \leq
        \gap(L)
        =
        \gap(\Lloc\vert_{\calA_d})
        \leq
        e^{-c_{\beta,d}n}.
    \end{equation}
\end{enumerate}
\end{theorem}

\begin{proof}
    Part~\textup{(i)} is \cref{thm:young-chain-fast-mixing}. Part~\textup{(ii)} is \cref{thm:young-chain-slow-mixing}.
\end{proof}

\begin{remark}
    The theorem says that the only exponentially small gaps in the Young diagram sector occur in the metastable window $\betamin(d)<\beta<d$. Outside this window, the spectral free energy landscape has a single well. Since there are no barriers to mixing, the Young diagram chain has at worst a polynomially small gap.
\end{remark}

\subsection{Full Dynamic Phase Diagram (Proof of \cref{thm:full-dynamics})}\label{subsec:full-dyn-phase-diagram}

The key quantity in the static picture is only the critical inverse temperature $\beta_c(d)$. In contrast, the dynamics is beholden to a pair of spinodal inverse temperatures $\betamin(d) < \beta_c(d) < d$ which are the thresholds defining the window for fast and slow mixing. Combining the symmetry decomposition of \cref{prop:invariant-symmetry-sectors} with the sector gap bounds of \cref{thm:symmetry-controlled-sector-bounds,thm:spectral-sector-gap-bounds} gives our main result---a full dynamical phase diagram for the Davies evolution of the interchange model. We give the proof here.

\begin{proof}
    By \cref{prop:invariant-symmetry-sectors},
    \begin{equation}
        \gap(\Lloc)
        =
        \min\left\{
            \gap(\Lloc\vert_{\calA_d}),
            \lambda_{\min}(-\Lloc\vert_{\calK_d}),
            \lambda_{\min}
            \left(
                -\Lloc\vert_{\mathrm{comm}(S_n)^{\perp_\rho}}
            \right)
        \right\}.
    \end{equation}
    
    If $0<\beta<\betamin(d)$ or $\beta>d$, the Young diagram sector estimate in \cref{thm:spectral-sector-gap-bounds}\textup{(i)} and lower bounds presented in \cref{thm:symmetry-controlled-sector-bounds} give polynomial lower bounds on all three terms. Hence $\gap(\Lloc) \geq n^{-c_{\beta,d}}$ for suitably large $c_{\beta,d}$.
    
    Now suppose $\betamin(d)<\beta<d$. \Cref{thm:spectral-sector-gap-bounds}\textup{(ii)} gives $e^{-C_{\beta,d}n}\leq\gap(\Lloc\vert_{\calA_d})\leq e^{-c_{\beta,d}n}$. For sufficiently large $n$, the polynomial lower bounds on the other two sectors also exceed $e^{-C_{\beta,d}n}$, after increasing $C_{\beta,d}$ if necessary. Taking the minimum proves both bounds.
\end{proof}

\section*{Acknowledgments}
This material is based upon work supported by the U.S. Department of Energy, Office of Science, National Quantum Information Science Research Centers, Quantum Systems Accelerator. Additional support is acknowledged from the Simons Targeted Grants in Mathematics and Physical Sciences on Moir\'e Materials Magic (K.D.S., L.L.), and the Simons Investigator in Mathematics award through Grant No. 825053 (L.L.). We thank Ehud Altman, Joao Basso, Thiago Bergamaschi, Pablo Sala, and Jack Spilecki for insightful discussions.

We used ChatGPT 5.6 and ChatGPT 6 to assist with editing, identifying minor errors, and (in appendix D) developing arguments from sketches provided by the authors. The authors wrote up and verified the arguments and take full responsibility for the content of this work.

\newpage
\appendix

\crefalias{section}{appendix}

\section{Representation Theory Background}\label{app:rep-theory-background}

This appendix collects the representation-theoretic definitions and facts used throughout the paper. These can all be found in standard texts on representation theory and the representation theory of $S_n$ such as, e.g.~\cite{FultonHarris1991} or~\cite{Sagan2001}.

\subsection{Representations and Equivariance}
Let $G$ be a group and let $V$ be a finite-dimensional complex vector space. A \textbf{representation} of $G$ on $V$ is a group homomorphism
$R:G\to GL(V)$. If $V$ is a Hilbert space and $R(g)$ is unitary for every $g\in G$, we call $R$ a \textbf{unitary representation}. A subspace $W\subseteq V$ is called \textbf{$G$-invariant} if $R(g)W\subseteq W$ for every $g\in G$. If $W\subseteq V$ is $G$-invariant, then the restriction $R\vert_W:G\to GL(W)$ given by $R\vert_W(g):=R(g)\vert_W$, is called a \textbf{subrepresentation} of $R$. Every representation has two trivial subrepresentations, corresponding to the trivial subspaces $\{0\}$ and $V$. A representation is said to be \textbf{irreducible} if it has no nontrivial subrepresentations.

Let $(R_1,V_1)$ and $(R_2,V_2)$ be representations of $G$. A linear map $T:V_1\to V_2$ is said to be \textbf{$G$-equivariant} if $TR_1(g) = R_2(g)T$ for every $g\in G$. We say that two representations are \textbf{isomorphic} and write $V_1\cong V_2$ if there exists an invertible $G$-equivariant map $T$ between them. When these representations are both unitary, then $T$ can be chosen to be unitary. The space of $G$-equivariant maps is denoted $\mathrm{Hom}_G(V_1,V_2)$. We recall the invaluable Schur's lemma: if $(R_1,V_1)$ and $(R_2,V_2)$ are finite-dimensional irreducible unitary representations of $G$ and $T\in\mathrm{Hom}_G(V_1,V_2)$, then either $T=0$ or $T$ is an     isomorphism. In particular, if $V_1=V_2$ and $R_1 = R_2$, then $T=cI$ for some $c\in\mathbb C$.

\begin{definition}[Commutant]
    If $G$ acts on $V$ through a representation $R$, its commutant is
    \begin{equation}
        \mathrm{comm}(G)
        :=
        \{X\in\calB(V):XR(g) = R(g)X\text{ for every }g\in G\}.
    \end{equation}
\end{definition}

\begin{fact}
    Let $G$ and $H$ be two groups acting on a finite-dimensional vector space $V$, and suppose that their actions commute. Then,
    \begin{equation}
    \mathrm{comm}(G\times H)
    =
    \mathrm{comm}(G)\cap \mathrm{comm}(H).
\end{equation}
\end{fact}
Finite dimensional unitary representations are \textbf{completely reducible}, meaning that every unitary representation $(R,V)$ is isomorphic to an orthogonal direct sum of irreducible subrepresentations. Every representative $R(g)$ is simultaneously block-diagonalized in any orthogonal basis realizing this direct sum: 
\begin{equation}
    V \cong \bigoplus_\lambda V_\lambda^{\oplus n_\lambda} = \bigoplus_\lambda V_\lambda\otimes \C^{n_\lambda}, \qquad R(g) = \bigoplus_\lambda R_\lambda(g)\otimes I_{n_\lambda}, 
\end{equation} where each $(V_\lambda, R_\lambda)$ is irreducible and $n_{\lambda}$ is its corresponding multiplicity.

\subsection{Representations of the Symmetric Group}

Irreducible complex representations of the symmetric group $S_n$ are indexed by partitions $\lambda\vdash n$, or equivalently Young diagrams with $n$ boxes. We denote the irreducible representation corresponding to $\lambda$ by $W^\lambda$ -- we shall also refer to these as Specht modules. If $\lambda=(\lambda_1,\lambda_2,\ldots)$, then $\lambda_i$ is the length of the $i$-th row. We write $\ell(\lambda)$ for the number of non-zero rows and $\vert \lambda \vert$ for the number of boxes $n$. For a box $b$ in row $i$ and column $j$, its \textbf{content} is defined as $c(b) := j-i$.

\begin{definition}[Hook length]
For a Young diagram $\lambda$, let $b=(i,j)\in\lambda$ be the box in row $i$ and column $j$. The hook of $b$ consists of the box $b$ itself, together with all boxes in $\lambda$ directly to the right of $b$ in row $i$ and all boxes directly below $b$ in column $j$. The hook length of $b$ is the number of boxes in the hook of $b$.
\end{definition}

\begin{fact}[Hook-length formula]\label{fact:hook-length-formula}
The dimension of the Specht module $W^\lambda$ is
\begin{equation}
    \dim(W^\lambda)
    =
    \frac{n!}{\prod_{b\in\lambda}h(b)},
\end{equation}
where $h(b)$ is the hook length of the box $b$.
\end{fact}

For partitions with at most $d$ rows, the hook product may equivalently be written as
\begin{equation}
    \prod_{b\in\lambda}h(b)
    =
    \frac{
        \prod_{i=1}^d(\lambda_i+d-i)!
    }{
        \prod_{1\leq i<j\leq d}(\lambda_i-\lambda_j+j-i)
    }.
\end{equation}

\subsection{Representations of \texorpdfstring{$SU(d)$}{SU(d)}}\label{subsec:su-representations}
Irreducible polynomial representations of $GL(d)$ are indexed by partitions $\lambda$ with at most $d$ rows. Restricting these representations to $SU(d)$ gives every finite-dimensional irreducible representation of $SU(d)$. We denote the corresponding $SU(d)$-representation by $V^\lambda$.

\begin{fact}[Weyl dimension formula]\label{fact:weyl-dimension}
For $\lambda=(\lambda_1,\ldots,\lambda_d)$ with $\lambda_1\geq\cdots\geq\lambda_d\geq0$, the dimension of the corresponding $SU(d)$-irrep is
\begin{equation}
    \dim(V^\lambda)
    =
    \prod_{1\leq i<j\leq d}
    \frac{\lambda_i-\lambda_j+j-i}{j-i}.
\end{equation}
\end{fact}

We will also use the action of scalar matrices on these
representations. The center of $SU(d)$ consists of the matrices
$zI$ with $z^d=1$. On $V^\lambda$ with $\lambda \vdash n$, such a matrix acts as multiplication by $z^n$. Indeed, on the tensor power
containing $V^\lambda$ its action is $(zI)^{\otimes n} = z^nI$ and restricting to $V^\lambda$ gives the claimed action.

\section{Young Diagram Master Equation}\label{appdx:young-diagram-master-proof}

We derive the transition coefficients using the multiplicity-free branching rules for $S_n$ and $SU(d)$, proving \cref{thm:master-eqn-coefficients}. Recall that the coefficients are defined as 
\begin{equation}
    L_{\lambda\to\lambda'}
    =
    \frac{\Tr\left(\Pi_\lambda\Lloc(\Pi_{\lambda'})\right)}{\Tr(\Pi_\lambda)}.
\end{equation}

\subsection{The Branching Rule and One-Box Moves}

As we described in the dialogue above \cref{thm:master-eqn-coefficients}, the allowed transitions in the Markov chain are between irreps labeled by diagrams that differ by a single box-move. This interpretation is derived from two key representation-theoretic facts: Young's branching rule and Pieri's rule.

In what follows, we will often isolate a single qudit $i \in [n]$, work on the space of $n-1$ qudits with qudit $i$ traced out, and then embed this back into the original space. This game will naturally lead us to consider representations of the group $S_n$ and its subgroup $S_{n-1}$ corresponding to permutations fixing this qudit, as well as the interplay between Schur-Weyl duality on the Hilbert space thought of as $(\C^d)^{\otimes n}$ and as $(\C^d)^{\otimes (n-1)} \otimes \C^d$.

For $S\subseteq[n]$, write $\calH_S:=\bigotimes_{j\in S}\mathbb C^d$. For $i\in[n]$, let $\ihat:=[n]\setminus\{i\}$ and use the natural tensor-factor identification $\calH_{[n]} \cong\calH_{\ihat}\otimes\mathbb C_i^d$. We embed operators on $\calH_{\ihat}$ into operators on $\calH_{[n]}$ by
\begin{equation}
    \iota_i:\calB(\calH_{\ihat}) \rightarrow \calB(\calH_{[n]}),
    \qquad
    \iota_i(A):=A \otimes I_i.
\end{equation}
When an operator on $\calH_{\ihat}$ appears in an expression with operators on $\calH_{[n]}$, this embedding is understood. As usual, $\Tr_i:\calB(\calH_{[n]})\to\calB(\calH_{\ihat})$ denotes the partial trace over the $i^{th}$ tensor factor. We use the superscript $\ihat$ for an operator acting on all sites except $i$, e.g. if $A$ is an operator which acts trivially on qudit $i$, then we write it as $A^{\ihat} \otimes I_i$ or just $A^{\ihat}$.

Given a group $G$ with subgroup $H$ and a representation $R:G\to GL(V)$, the \textbf{restricted representation} of $R$ from $G$ to $H$ is denoted by $R\down_H(h) = R(h)$ for all $h \in H$. In our setting we will consider $G=S_n$ and $H=S_{n-1}$. Young's branching rule and Pieri's rule will allow us to identify adding or removing a qudit from the system with adding or removing one box from the Young diagram. More specifically, Young's branching rule will imply that $\Tr_i(\Pi_\lambda)$ is supported on diagrams in $\lambda^-$, while the Pieri rule will imply that $\sum_i \Pi_\nu^{\ihat}\otimes I_i$ is supported on diagrams in $\nu^+$. These rules are dual to each other and in fact can be proved from one another. Since it will provide valuable intuition for understanding the twirl, we show how to prove the branching rule from Pieri's rule.

\begin{fact}[Young's Branching Rule]\label{fact:Young-branching-rule}
    If $\lambda \in \Omega_{n,d}$, then
    \begin{equation}
        W^\lambda \down_{S_{n-1}} \cong \bigoplus_{\nu \in \lambda^-} W^\nu.
    \end{equation}
\end{fact}

\begin{fact}[Pieri's Rule]\label{fact:Pieri-rule}
    Let $\nu \in \Omega_{n-1,d}$ and $V^\nu$ be an irrep of $SU(d)$. Then,
    \begin{equation}
        V^\nu \otimes \C^d \cong \bigoplus_{\mu \in \nu^+} V^\mu.
    \end{equation}
\end{fact}

It is crucial to point out that the decomposition appearing in both rules are multiplicity-free.

\begin{proof}[Proof of \cref{fact:Young-branching-rule}]

We start by applying Schur--Weyl duality on $(\C^d)^{\otimes (n-1)}$ to obtain the representation of $SU(d)\times S_{n-1}$
    \begin{equation}
        (\C^d)^{\otimes (n-1)} \cong \bigoplus_{\nu \in \Omega_{n-1,d}} V^\nu \otimes W^\nu.
    \end{equation}
    In the case of a single qudit $\C^d$, the permutation group acts trivially, so we write $\C^d\cong \C^d\otimes W^{\square}$. 
    Thus, we may decompose the space of $n$ qudits as 
    \begin{align}
        (\C^d)^{\otimes (n-1)} \otimes \C^d &\cong \bigoplus_{\nu \in \Omega_{n-1,d}} (V^\nu \otimes W^\nu) \otimes (\C^d\otimes W^{\square})\\
        &\cong \bigoplus_{\nu \in \Omega_{n-1,d}} (V^\nu \otimes \C^d) \otimes W^\nu\\
        &\cong \bigoplus_{\nu \in \Omega_{n-1,d}} \bigoplus_{\lambda \in \nu^+} V^\lambda \otimes W^\nu,
    \end{align} where in the second line we have used that $W^{\nu}\otimes W^{\square}\cong W^{\nu}$ and in the third we have used Pieri's rule. Continuing, we may swap the order of the finite direct sums:
    
    \begin{align}
        (\C^d)^{\otimes (n-1)} \otimes \C^d &\cong \bigoplus_{\lambda \in \Omega_{n,d}} V^\lambda \otimes \bigoplus_{\nu \in \lambda^-} W^\nu.
    \end{align}
    Meanwhile, by Schur--Weyl duality on $(\C^d)^{\otimes n} = (\C^d)^{\otimes (n-1)} \otimes \C^d$ we have as a representation of $SU(d)\times S_n$
    \begin{align}
        (\C^d)^{\otimes (n-1)} \otimes \C^d \cong \bigoplus_{\lambda \in \Omega_{n,d}} V^\lambda \otimes W^\lambda
    \end{align}
    The restriction of $S_n$ to the subgroup $S_{n-1}$ given by permutations fixing qudit $i$ does not affect the action of $SU(d)$, and so comparing both decompositions proves Young's branching Rule.
\end{proof}

\subsection{Understanding the Twirl}
Recall that the Lindbladian can be expressed in terms of the Lie algebra twirl, which furthermore can be expressed in terms of the partial trace~\cref{eq:lie-algebra-twirl-partial-trace}. For the Markov chain coefficients we will need to compute $\Lloc(\Pi_\lambda)$ and hence $\calT(\Pi_\lambda)$. To do this, we shall first find a useful expression for $\Tr_i(\Pi_\lambda)$, then leverage that to develop a representation-theoretic expression for $\calT(\Pi_\lambda)$.

\begin{lemma}[Partial trace of a Schur--Weyl projector]\label{partialtraceprojector}
    For $\lambda\in\Omega_{n,d}$, the linear operator $\Tr_i(\Pi_\lambda)\in \calB(\calH_{\ihat})$ may be written as
    \begin{equation}
        \Tr_i(\Pi_\lambda)
        =
        \sum_{\nu\in\lambda^-}
        \frac{\dim(V^\lambda)}{\dim(V^\nu)}
        \Pi_\nu^{\ihat}.
    \end{equation}
\end{lemma}
\begin{proof}
    First, we show that $\Tr_i(\Pi_\lambda)$ commutes with the actions of both $SU(d)$ and $S_{n-1}$ on $(\C^d)^{\otimes(n-1)}$. For any $U \in U(d)$ we have
    \begin{align}
        U^{\otimes(n-1)} &\Tr_i(\Pi_\lambda) = \Tr_i((U^{\otimes(n-1)} \otimes I_i)\Pi_\lambda) = \Tr_i(U^{\otimes n}\Pi_\lambda(I^{\otimes(n-1)} \otimes U_i^\dagger))\\ 
        &= \Tr_i(\Pi_\lambda (U^{\otimes(n-1)} \otimes I_i)) = \Tr_i(\Pi_\lambda) U^{\otimes(n-1)}
    \end{align}
    and for any $\sigma \in S_{n-1}$ we have
    \begin{align}
        P^{\ihat}(\sigma)\Tr_i(\Pi_\lambda) = \Tr_i((P^{\ihat}(\sigma) \otimes I_i)\Pi_\lambda) = \Tr_i(\Pi_\lambda(P^{\ihat}(\sigma)\otimes I_i)) = \Tr_i(\Pi_\lambda)P^{\ihat}(\sigma),
    \end{align}
    Thus, $\Tr_i(\Pi_\lambda)$ lies in the commutant of $SU(d) \times S_{n-1}$ so by \cref{lem:commutant-is-span-projectors} has the form
    \begin{equation}
        \Tr_i(\Pi_\lambda) = \sum_{\nu \in \Omega_{n-1,d}} c_{\nu} \Pi_\nu^{\ihat}.
    \end{equation}
    We now compute the coefficients $c_\nu$. 
    Fix $\mu \in \Omega_{n-1,d}$ and consider the isotypic projector on $\calH_{\ihat}$ (thought of as embedded into $\calH$ via $\iota_i$) associated to diagram $\mu$. We have
    \begin{equation}
        \Pi_\mu^{\ihat}\Tr_i(\Pi_\lambda) = \Pi_\mu^{\ihat} \sum_{\nu \in \Omega_{n-1,d}} c_\nu \Pi_\nu^{\ihat} = c_\mu \Pi_\mu^{\ihat}.
    \end{equation}
    Taking the trace of the right-hand side gives
    \begin{equation}\label{eq:partial-trace-coefficient}
        c_\mu \dim(V^\mu)\dim(W^\mu).
    \end{equation}
    Taking the trace of the left-hand side gives
    \begin{align}
        \Tr(\Pi_\mu^{\ihat}\Tr_i(\Pi_\lambda)) &= \Tr((\Pi_\mu^{\ihat} \otimes I_i) \Pi_\lambda)
    \end{align}
    
    We compute this trace as follows. Since this trace is unchanged by relabeling the tensor factors, we may assume WLOG that $i=n$. Write
    \begin{equation}
        \widetilde\Pi_\mu
        :=
        \Pi_\mu^{[n-1]}\otimes I_n.
    \end{equation}

    Applying Schur--Weyl duality and Young's branching rule onto $\calH = (\C^d)^{\otimes n}$ gives an orthogonal decomposition
    \begin{equation}\label{eq:Cdn-decomp}
    \calH = \bigoplus_{\alpha\in\Omega_{n,d}} \bigoplus_{\nu\in\alpha^-}\calH_{\alpha,\nu},
        \qquad
        \calH_{\alpha,\nu}\cong V^\alpha\otimes W^\nu.
    \end{equation}
    Moreover, applying Schur--Weyl duality and Pieri's rule onto $(\C^d)^{\otimes(n-1)} \otimes \C^d$ gives an equivalent form
    \begin{equation}\label{eq:cdn-1-decomp}
        \calH = \bigoplus_{\nu \in \Omega_{n-1,d}} (V^\nu \otimes \C^d) \otimes W^\nu = \bigoplus_{\nu \in \Omega_{n-1,d}} \bigoplus_{\alpha \in \nu^+} \calH_{\alpha,\nu}
    \end{equation}
    In both decompositions, the summands can be chosen to be orthogonal because the group actions are unitary. The projector $\Pi_\lambda$ selects the summands in \cref{eq:Cdn-decomp} with $\alpha=\lambda$, while $\widetilde\Pi_\mu$, the $S_{n-1}$-isotypic projector, selects those in \cref{eq:cdn-1-decomp} with $\nu=\mu$. Thus,
    \begin{equation}
        \mathrm{Ran}(\Pi_\lambda)
        =
        \bigoplus_{\nu\in\lambda^-}
        \calH_{\lambda,\nu},\qquad
        \text{and}
        \qquad
        \mathrm{Ran}(\widetilde\Pi_\mu)
        =
        \bigoplus_{\alpha\in\mu^+}
        \calH_{\alpha,\mu}.
    \end{equation}
    Consequently, the two projectors commute and the intersection of their ranges is given by
    \begin{equation}\label{eq:intersection-of-schur-weyl-sectors}
        \mathrm{Ran}(\Pi_\lambda)\cap\mathrm{Ran}(\widetilde\Pi_\mu)
        =
        \begin{cases}
            \calH_{\lambda,\mu},
            &\mu\in\lambda^-,\\
            \{0\},
            &\mu\notin\lambda^-.
        \end{cases}
    \end{equation}
    For $\mu\in\lambda^-$, we have $\dim(\calH_{\lambda,\mu})=\dim(V^\lambda)\dim(W^\mu)$. Since $\widetilde\Pi_\mu\Pi_\lambda$ is the orthogonal projector onto the intersection in \cref{eq:intersection-of-schur-weyl-sectors}, we conclude that
    \begin{equation}
        \Tr\left(\widetilde\Pi_\mu\Pi_\lambda\right)
        =
        \begin{cases}
            \dim(V^\lambda)\dim(W^\mu),
            &\mu\in\lambda^-,\\
            0,
            &\mu\notin\lambda^-.
        \end{cases}
    \end{equation}
    Comparing with \cref{eq:partial-trace-coefficient} gives
    \begin{equation}
            c_\mu
            =
            \mathbf 1_{\{\mu\in\lambda^-\}}
            \frac{\dim(V^\lambda)}{\dim(V^\mu)},
    \end{equation}
    as desired.
\end{proof}

The next identity is a selection rule for the depolarizing channel: the twirl connects only diagrams related through a one-box move.

\begin{proposition}[Schur--Weyl Expansion of the Twirl]\label{prop:schur-weyl-expansion-twirl}
    The Lie algebra twirl of an isotypic Schur--Weyl projector can be expressed as
    \begin{align}
        \calT(\Pi_\lambda) = \frac{n}{d} \sum_{\nu \in \lambda^-} \frac{\dim(V^\lambda)}{\dim(V^\nu)} \sum_{\mu \in \nu^+} \frac{\dim(W^\nu)}{\dim(W^\mu)} \Pi_\mu.
    \end{align}
\end{proposition}
\begin{proof}
    Applying~\cref{partialtraceprojector} to ~\cref{eq:lie-algebra-twirl-partial-trace} and swapping the sums,
    \begin{align}\label{eq:raw-twirl}
        \calT(\Pi_\lambda) &= \frac{1}{d}\sum_{i=1}^n \Tr_i(\Pi_\lambda) \otimes I_i\\
        &= \frac{1}{d} \sum_{\nu \in \lambda^-} \frac{\dim(V^\lambda)}{\dim(V^\nu)} \left(\sum_{i=1}^n \Pi_\nu^{\ihat} \otimes I_i\right)
    \end{align}
    Define
    \begin{equation}
        A_\nu = \sum_{i=1}^n \Pi_\nu^{\ihat} \otimes I_i.
    \end{equation}
    Note that $A_\nu$ commutes with the actions of $SU(d)$ and $S_n$, so by \cref{lem:commutant-is-span-projectors}, we have
    \begin{equation}
        A_\nu = \sum_{\kappa \in \Omega_{n,d}} a_\kappa \Pi_\kappa,
    \end{equation}
    and we can determine the coefficients $a_\kappa$ by taking traces in two ways:
    \begin{equation}
        \Tr(A_\nu \Pi_\mu) = a_\mu \Tr(\Pi_\mu)
    \end{equation}
    and
    \begin{align}
        \Tr(A_\nu \Pi_\mu) &= \Tr\left(\left(\sum_{i=1}^n \Pi_\nu^{\ihat} \otimes I_i\right) \Pi_\mu\right)\\
        &= \sum_{i=1}^n \Tr(\Tr_i((\Pi_\nu^{\ihat} \otimes I_i)\Pi_\mu))\\
        &= \sum_{i=1}^n \Tr(\Pi_\nu^{\ihat}\Tr_i(\Pi_\mu))\\
        &= \sum_{i=1}^n \Tr\left(\Pi_\nu^{\ihat}\sum_{\alpha \in \mu^-} \frac{\dim(V^\mu)}{\dim(V^\alpha)} \Pi_\alpha^{\ihat}\right)\\
        &= \sum_{i=1}^n \sum_{\alpha \in \mu^-}\delta_{\alpha\nu} \frac{\dim(V^\mu)}{\dim(V^\alpha)}\Tr(\Pi_\nu^{\ihat})\\
        &= n\mathbf 1_{\{\nu \in \mu^-\}} \frac{\dim(V^\mu)}{\dim(V^\nu)} \Tr(\Pi_\nu^{[n]\setminus \{1\}}).
    \end{align}
    In the third equality we used the defining property of the partial trace, in the fourth equality we applied the previous lemma, and in the final equality we used the fact that all $\Pi_\nu^{\ihat}$ have the same trace. Thus, 
    \begin{align}
        a_\mu &= n\mathbf 1_{\{\mu \in \nu^+\}} \frac{\dim(V^\mu)}{\dim(V^\nu)} \frac{\Tr(\Pi_\nu^{[n]\setminus \{1\}})}{\Tr(\Pi_\mu)}\\
        &= n\mathbf 1_{\{\mu \in \nu^+\}} \frac{\dim(W^\nu)}{\dim(W^\mu)}.
    \end{align}
    Substitution into \cref{eq:raw-twirl} gives
    \begin{equation}
        \calT(\Pi_\lambda) = \frac{n}{d} \sum_{\nu \in \lambda^-} \frac{\dim(V^\lambda)}{\dim(V^\nu)} \sum_{\mu \in \nu^+} \frac{\dim(W^\nu)}{\dim(W^\mu)} \Pi_\mu,
    \end{equation}
    as desired.
\end{proof}

\begin{proof}[Proof of \cref{thm:master-eqn-coefficients}]
We compute $L_{\lambda\to\lambda'}$ for a fixed pair
$\lambda,\lambda'$. From \cref{eq:davies-twirl} applied to $\Lloc$,
\begin{equation}
    \Lloc(\Pi_{\lambda'})
    =
    \sum_{E\in\mathrm{spec}(H)}
    \gamma(E_{\lambda'}-E)
    \Pi_E\calT(\Pi_{\lambda'})\Pi_E
    -
    \sum_{E\in\mathrm{spec}(H)}
    \gamma(E-E_{\lambda'})
    \Pi_{\lambda'}\calT(\Pi_E)\Pi_{\lambda'}.
\end{equation}
For $\lambda\neq\lambda'$, the Schur--Weyl expansion of \cref{prop:schur-weyl-expansion-twirl} gives
\begin{align}
    \Pi_{\lambda}\Lloc(\Pi_{\lambda'})
    &=
    \mathbf 1_{\{\lambda\sim\lambda'\}}
    \gamma(E_{\lambda'}-E_\lambda)
    \frac nd
    \frac{\dim(V^{\lambda'})}{\dim(V^\nu)}
    \frac{\dim(W^\nu)}{\dim(W^\lambda)}
    \Pi_\lambda .
\end{align}
Therefore
\begin{align}
    L_{\lambda\to\lambda'}
    &=
    \frac{
        \Tr\!\left(\Pi_\lambda\Lloc(\Pi_{\lambda'})\right)
    }{
        \Tr(\Pi_\lambda)
    } \\
    &=
    \mathbf 1_{\{\lambda\sim\lambda'\}}
    \gamma(E_{\lambda'}-E_\lambda)
    \frac nd
    \frac{\dim(V^{\lambda'})}{\dim(V^\nu)}
    \frac{\dim(W^\nu)}{\dim(W^\lambda)}.
\end{align}
The diagonal is determined by the row-sum condition $L_{\lambda\to\lambda} = -\sum_{\lambda'\neq\lambda}L_{\lambda\to\lambda'}.$
\end{proof}

\section{Stationary Measure and Free Energy}\label{app:stationary-measure}

In this appendix we find the stationary measure $\pi_\beta(\lambda)$ and derive a uniform polynomial equivalence between the stationary measure and the exponential of the free energy $F_\beta$, \cref{eq:free-energy-1}. We then prove \cref{thm:spectral-free-energy-landscape}.

\subsection{The Free Energy and Polynomial Estimates}
Detailed balance requires
\begin{equation}
    \frac{L_{\lambda \to \lambda'}}{L_{\lambda' \to \lambda}} = \frac{\pi_\beta(\lambda')}{\pi_\beta(\lambda)}.
\end{equation}

For $\lambda\neq\lambda'$ with $\lambda\overset{\nu}{\sim}\lambda'$, write $\omega=E_{\lambda'}-E_\lambda$. Then
\begin{align}
    \frac{L_{\lambda\to\lambda'}}{L_{\lambda'\to\lambda}} =
    \frac{\gamma(\omega)}{\gamma(-\omega)}
    \frac{\dim(V^{\lambda'})}{\dim(V^\lambda)}
    \frac{\dim(W^{\lambda'})}{\dim(W^\lambda)}=
    e^{-\beta(E_{\lambda'}-E_\lambda)}
    \frac{\dim(V^{\lambda'})\dim(W^{\lambda'})}
    {\dim(V^\lambda)\dim(W^\lambda)} =
    \frac{\pi_\beta(\lambda')}{\pi_\beta(\lambda)}.
\end{align}
Thus
\begin{equation}\label{eq:stationary-measure}
    \pi_\beta(\lambda)
    =
    \frac1{Z_\beta}
    e^{-\beta E_\lambda}
    \dim(V^\lambda)\dim(W^\lambda).
\end{equation}

\begin{lemma}[Schur--Weyl block dimension]\label{lem:dim-Pi_lambda}
    For $\lambda=(\lambda_1,\ldots,\lambda_d)\in\Omega_{n,d}$, we have
    \begin{align}
      \dim(\Pi_{\lambda})
      &=
      \frac{n!}{\lambda_{1}!\cdots\lambda_d!}
      \frac{\displaystyle\prod_{1\leq i<j\leq d}
      (\lambda_i-\lambda_j+j-i)^2}
      {\displaystyle
      \left(\prod_{1\leq i<j\leq d}(j-i)\right)
      \left(\prod_{i=1}^d\prod_{r=1}^{d-i}(\lambda_i+r)\right)}\\
      &= \frac{n!}{\lambda_1! \cdots \lambda_d!} \frac{p_1(\lambda;d)}{p_2(\lambda;d)},
    \end{align}
    where, for fixed $d$, $p_i(\lambda;d)$ are polynomials in the row lengths $\lambda_1,\ldots,\lambda_d$.
    Moreover, uniformly in $\lambda\in\Omega_{n,d}$,
    \begin{equation}\label{eq:schur-weyl-dimension-bounds}
        \frac{e^{nH(x_\lambda)}}{(n+1)^d(n+d)^{\binom d2}}
        \leq \dim(\Pi_\lambda)
        \leq (n+d)^{\binom d2}e^{nH(x_\lambda)},
    \end{equation}
    where $x_\lambda=\lambda/n$ and $H(x):=-\sum_i x_i\log x_i$, with $0\log0=0$.
\end{lemma}
\begin{proof}
    Our starting point is
    \begin{equation}
        \dim(\Pi_\lambda) = \dim(V^\lambda)\dim(W^\lambda).
    \end{equation}
    By Weyl's dimension formula~\cite{FultonHarris1991} for $\dim(V^\lambda)$ and the hook-length formula~\cite{Sagan2001} for the Specht module $W^\lambda$ we have
    \begin{equation}
        \dim(V^\lambda) = \prod_{1\leq i<j\leq d} \frac{\lambda_i-\lambda_j+j-i}{j-i},\qquad
    \dim(W^\lambda) = \frac{n!}{\prod_{b\in\lambda}h(b)}.
    \end{equation}
    For a partition with at most $d$ rows, the hook product may be written as
    \begin{equation}
        \prod_{b\in\lambda}h(b) = \frac{\prod_{i=1}^d(\lambda_i+d-i)!}{\prod_{1\leq i<j\leq d}(\lambda_i-\lambda_j+j-i)}.
    \end{equation}
    Therefore
    \begin{equation}
        \dim(W^\lambda) = n! \frac{\prod_{1\leq i<j\leq d}(\lambda_i-\lambda_j+j-i)}{\prod_{i=1}^d(\lambda_i+d-i)!}.
    \end{equation}
    Now, factor
    \begin{equation}
        (\lambda_i + d-i)! = \lambda_i! \prod_{m=1}^{d-i} (\lambda_i + m),
    \end{equation}
    so that 
    \begin{equation}
        \prod_{i=1}^d (\lambda_i + d - i)! = \left(\prod_{i=1}^d \lambda_i!\right)\left(\prod_{i=1}^d \prod_{m=1}^{d-i} (\lambda_i + m)\right).
    \end{equation}
    Substituting this into the formula for $\dim(W^\lambda)$ gives
    \begin{align}
        \dim(\Pi_\lambda) &= \dim(V^\lambda)\dim(W^\lambda)\\
        &= \frac{n!}{\lambda_1!\cdots\lambda_d!}\frac{\prod_{1\leq i<j\leq d}(\lambda_i-\lambda_j+j-i)^2}{\left(\prod_{1\leq i<j\leq d}(j-i)\right)\left(\prod_{i=1}^d\prod_{m=1}^{d-i} (\lambda_i+m)\right)}
    \end{align}
    as desired. To obtain the bounds, write
    \begin{equation}
        \frac{p_1(\lambda;d)}{p_2(\lambda;d)}
        =\prod_{1\leq i<j\leq d}
        \frac{(\lambda_i-\lambda_j+j-i)^2}{(j-i)(\lambda_i+j-i)}.
    \end{equation}
    For each factor, set $a=\lambda_i-\lambda_j+j-i$ and $b=\lambda_i+j-i$.
    Since $1\leq j-i\leq a\leq b\leq n+d$, each factor lies between $(n+d)^{-1}$ and $n+d$. Hence
    \begin{equation}\label{eq:dimension-prefactor-bounds}
        (n+d)^{-\binom d2}
        \leq\frac{p_1(\lambda;d)}{p_2(\lambda;d)}
        \leq(n+d)^{\binom d2}.
    \end{equation}
    The multinomial coefficient satisfies
    \begin{equation}\label{eq:multinomial-entropy-bounds}
        (n+1)^{-d}e^{nH(x_\lambda)}
        \leq\frac{n!}{\lambda_1!\cdots\lambda_d!}
        \leq e^{nH(x_\lambda)}.
    \end{equation}
    Combining \cref{eq:dimension-prefactor-bounds} and \cref{eq:multinomial-entropy-bounds} proves \cref{eq:schur-weyl-dimension-bounds}, including diagrams with empty rows.
\end{proof}

\begin{lemma}[Energy eigenvalues]\label{lem:energy-eigenvalues}
The interchange Hamiltonian
\begin{equation}
    H=-\frac1n\sum_{1\leq i<j\leq n}T_{ij}.
\end{equation}
acts as the scalar
\begin{equation}
    E_\lambda = \frac{1}{2n} \left(-\sum_{i=1}^d\lambda_i^2 + \sum_{i=1}^d(2i-1)\lambda_i \right)
\end{equation}
on the Schur--Weyl block $V^\lambda\otimes W^\lambda$.
\end{lemma}

\begin{proof}
Since $\sum_{i<j}T_{ij}$ is central in the group algebra of $S_n$, this operator acts as a scalar on each Specht module $W^\lambda$. This scalar is given by~\cite{Sagan2001}
\begin{equation}
    h_\lambda = \sum_{b\in\lambda}c(b),
\end{equation}
where $c(b)$ denotes the content of the box $b$. If $b$ lies in row $i$ and column $j$, then $c(b)=j-i$. Hence
\begin{align}
    h_\lambda = \sum_{i=1}^d\sum_{j=1}^{\lambda_i}(j-i) = \sum_{i=1}^d \left(\frac{\lambda_i(\lambda_i+1)}{2} - i\lambda_i \right) = \frac12 \sum_{i=1}^d \lambda_i(\lambda_i-2i+1).
\end{align}
Since $H=-\frac1n\sum_{i<j}T_{ij}$, the energy on the $\lambda$-block is $E_\lambda=-\frac1n h_\lambda$. Therefore
\begin{align}
    E_\lambda = -\frac{1}{2n} \sum_{i=1}^d \lambda_i(\lambda_i-2i+1) = \frac{1}{2n}
    \left(-\sum_{i=1}^d\lambda_i^2 + \sum_{i=1}^d(2i-1)\lambda_i \right).
\end{align}
\end{proof}

\begin{lemma}[Stationary Measure]\label{lem:stationary-measure-free-energy}
    Define the free energy $F_\beta: \Delta_d^\downarrow \to \R$ via
    \begin{equation}
        F_\beta(x)
        =
        \frac{\beta}{2}\sum_{i=1}^d x_i^2
        -
        \sum_{i=1}^d x_i\log x_i.
      \end{equation}
    Let $x_\lambda=\lambda/n$ and set $r = \ell(\lambda)$,
    \begin{equation}\label{eq:exact-young-weight}
      \pi_\beta(\lambda) = \frac{\Theta(1)}{Z_{\beta}} n^{-\frac{r-1}{2}}\prod_{i=1}^r x_{\lambda,i}^{-\frac{1}{2}} e^{-(\beta/2) \sum_{i} (2i - 1) x_{\lambda,i}} \frac{p_{1}(\lambda;d)}{p_{2}(\lambda;d)} e^{nF_\beta(x_\lambda)}.
    \end{equation}
\end{lemma}

\begin{proof}
    From \cref{eq:stationary-measure}, \cref{lem:dim-Pi_lambda}, and \cref{lem:energy-eigenvalues} we get
    \begin{equation}\label{eq:stationary-measure-simple}
        \pi_\beta(\lambda)
        =
        \frac1{Z_\beta} \frac{p_1(\lambda;d)}{p_2(\lambda;d)} \frac{n!}{\lambda_1!\cdots\lambda_r!} \exp\left\{\frac{\beta n}{2}\sum_{i=1}^d x_{\lambda,i}^2 -\frac\beta2\sum_{i=1}^d(2i-1)x_{\lambda,i}\right\}.
    \end{equation}
    
    Using Stirling's approximation on the $r=\ell(\lambda)$ non-empty rows and recalling $\sum_{i=1}^r \lambda_i = n$,
    \begin{align}
        \ln \frac{n!}{\lambda_1!\lambda_2! \cdots \lambda_r!}
        & = \left( n \ln(n) - n + \frac{1}{2} \ln(2 \pi n) \right) - \sum_{i=1}^r \left( \lambda_{i} \ln(\lambda_{i}) - \lambda_{i} + \frac{1}{2} \ln(2 \pi \lambda_{i}) \right) + O_d(1) \\
        & = \left\{ n \ln(n) - \sum_{i=1}^r \lambda_{i} \ln(\lambda_{i}) \right\} + \left\{ \frac{1}{2} \ln(2 \pi n) - \frac{1}{2} \sum_{i=1}^r  \ln\left( 2 \pi \lambda_{i} \right) \right\} + O_d(1) \\
        & = - \sum_{i=1}^r \lambda_{i} \ln\left(\frac{\lambda_{i}}{n}\right) + \frac{1}{2} \ln\left( \frac{n}{\lambda_{1} \cdots \lambda_r} \right) + O_d(1)
    \end{align}
    Therefore,
    \begin{align}
        \frac{n!}{\lambda_1!\lambda_2! \cdots \lambda_r!}
        & = \Theta(1) \sqrt{\frac{n}{\lambda_{1} \cdots \lambda_{r}} } \exp\left( -\sum_{i=1}^{r} \lambda_{i} \ln\left(\frac{\lambda_{i}}{n}\right) \right)\\
        & = \Theta(1) n^{-\frac{r - 1}{2}} \sqrt{\frac{n^{r}}{\lambda_{1} \cdots \lambda_{r}} }  \exp\left( -n \sum_{i=1}^{r} \frac{\lambda_{i}}{n} \ln\left(\frac{\lambda_{i}}{n}\right) \right)\\
        & = \Theta(1) n^{-\frac{r - 1}{2}} \prod_{i=1}^r x_{\lambda,i}^{-\frac{1}{2}}  \exp\left(n H(x_{\lambda,1}, \cdots, x_{\lambda,d}) \right)
    \end{align}
    where $H$ is the entropy from \cref{lem:dim-Pi_lambda}, so the empty rows contribute zero. Plugging this into \cref{eq:stationary-measure-simple} gives 
    \begin{equation}
      \pi_\beta(\lambda) = \frac{\Theta(1)}{Z_{\beta}} n^{-\frac{r-1}{2}}\prod_{i=1}^r x_{\lambda,i}^{-\frac{1}{2}} e^{-(\beta / 2) \sum_{i} (2i - 1) x_{\lambda,i}} \frac{p_{1}(\lambda;d)}{p_{2}(\lambda;d)}  \exp\left\{n \left(\frac{\beta}{2} \sum_{i=1}^{d} x_{\lambda,i}^{2} + H(x_\lambda)\right)\right\}
    \end{equation}
    as desired.
\end{proof}

\begin{corollary}[Stationary probability ratios]\label{cor:stationary-probability-ratios}

For any $\lambda,\eta\in\Omega_{n,d}$,
\begin{equation}\label{eq:stationary-ratio-bounds}
    \frac{\pi_\beta(\lambda)}{\pi_\beta(\eta)}
    \leq e^{\beta(d-1)}(n+1)^d(n+d)^{d(d-1)}
    e^{n(F_\beta(x_\lambda)-F_\beta(x_\eta))}.
\end{equation}
\end{corollary}
\begin{proof}
Cancel the partition function in \cref{eq:stationary-measure} and apply the dimension bounds~\cref{eq:schur-weyl-dimension-bounds} and \cref{lem:energy-eigenvalues}. The entropy and quadratic terms give $n(F_\beta(x_\lambda)-F_\beta(x_\eta))$, while the remaining exponent is
\begin{equation}
    \frac\beta2\sum_{i=1}^d(2i-1)(x_{\eta,i}-x_{\lambda,i})
    \leq\beta(d-1),
\end{equation}
since $1\leq\sum_i(2i-1)x_i\leq2d-1$ on the simplex.
\end{proof}

\begin{lemma}[Stationary Measure Estimates]\label{lem:stationary-measure-estimates}
    For fixed $d$ and $\beta$, uniformly in $\lambda\in\Omega_{n,d}$, we have the following polynomial equivalence between the stationary measure $\pi_\beta$ and the exponential of the free energy:
\begin{equation}\label{eq:stationary-measure-bounds}
n^{-c} e^{n f_\beta(x_\lambda)} \leq \pi_\beta(\lambda) \leq n^c e^{n f_\beta(x_\lambda)}
\end{equation}
where $f_\beta(x_\lambda) = F_\beta(x_\lambda) - F_0$, and $F_0 = \max_{x\in\Delta_d^\downarrow}F_\beta(x)$. Equivalently,
\begin{equation}\label{eq:stationary-log-asymptotic}
    \log\pi_\beta(\lambda)
    =
    n (F_\beta(x_\lambda)-F_0) + O_{\beta,d}(\log n).
\end{equation}
\end{lemma}
\begin{proof}
    We decompose the stationary measure as $\pi(\lambda) = \frac{1}{Z_\beta} P(n, x_\lambda) e^{n (f_\beta(x_\lambda) + F_0)}$, where, writing $r=\ell(\lambda)$ for the number of non-empty rows, the prefactor is given by
    \begin{equation}
    P(n, x_\lambda) = \Theta(1) n^{-\frac{r-1}{2}}\prod_{i=1}^r x_{\lambda,i}^{-1/2} e^{-\frac{\beta}{2} \sum_{i=1}^d (2i - 1) x_{\lambda,i}} \frac{p_{1}(\lambda;d)}{p_{2}(\lambda;d)}.
    \end{equation}
    We bound each component of $P(n, x_\lambda)$ uniformly over the state space $\Delta_d^\downarrow$:
    \begin{enumerate}
        \item Since $1 \leq \lambda_i \leq n$ for $1\leq i\leq r$, $1 \leq x_{\lambda,i}^{-1} \leq n$. Thus, using $1\leq r\leq d$, $n^{-\frac{d-1}{2}} \leq n^{-\frac{r-1}{2}}\prod_{i=1}^r x_{\lambda,i}^{-1/2} \leq n^{1/2}$. Empty rows do not appear in the product or in the Stirling expansion.
        \item Since $\sum_i x_i = 1$ and $x_i \in [0,1]$, the exponent is bounded by constants: $e^{-\beta(2d-1)/2} \leq e^{-\frac{\beta}{2} \sum_{i} (2i-1)x_i} \leq e^{-\beta/2}$.
        \item By \cref{eq:dimension-prefactor-bounds}, $n^{-c_p} \leq p_1/p_2 \leq n^{c_p}$ for some $c_p>0$ and all sufficiently large $n$.
    \end{enumerate}
    Combining these yields $n^{-c_0} \leq P(n, x_\lambda) \leq n^{c_0}$ for a sufficiently large constant $c_0 > 0$.
    
    Next, we bound the partition function $Z_\beta = e^{n F_0} \sum_{\lambda \in \Omega_{n,d}} P(n,x_\lambda) e^{n f_\beta(x_\lambda)}$. The cardinality of the state space $\Omega_{n,d}$ (partitions of $n$ into at most $d$ parts) is bounded by stars-and-bars: $\vert \Omega_{n,d} \vert \leq \binom{n+d-1}{d-1} \leq (n+1)^{d-1}$. Thus,
    \begin{equation}
        e^{n F_0} n^{-c_0} \leq Z_\beta \leq e^{n F_0} \cdot \vert\Omega_{n,d} \vert \cdot n^{c_0} \cdot \max_{x_\lambda \in \Delta_d^\downarrow} \left(e^{n f_\beta(x_\lambda)} \right) \leq e^{n F_0} n^{d-1+c_0}
    \end{equation}
    where we used the fact that $\max_{x_\lambda} f_\beta(x_\lambda) = 0$. Substituting the bounds for $P(n, x_\lambda)$ and $Z_\beta$ back into $\pi(\lambda)$ yields:
    \begin{equation}
    \frac{n^{-c_0} e^{n F_0} e^{n f_\beta(x_\lambda)}}{e^{n F_0} n^{d-1+c_0}} \leq \pi(\lambda) \leq \frac{n^{c_0} e^{n F_0} e^{n f_\beta(x_\lambda)}}{e^{n F_0} n^{-c_0}} 
    \end{equation}
    which simplifies directly to
    \begin{equation}
    n^{-(2c_0 + d - 1)} e^{n f_\beta(x_\lambda)} \leq \pi(\lambda) \leq n^{2c_0} e^{n f_\beta(x_\lambda)}.
    \end{equation}
    Setting $c = 2c_0 + d - 1$ completes the proof.
\end{proof}

\subsection{The Spectral Free Energy Landscape}\label{subsec:proof-spectral-free-energy-landscape}

This section is dedicated to proving the second claim in \cref{prop:betamin-asymptotics-and-inequalities} and \cref{thm:spectral-free-energy-landscape}. We shall do this in three parts. First, we establish some facts about the relationship between the inverse temperature $\beta$, the structure of non-uniform local maxima (captured by the parameter $r$ defined below), and the local dimension $d$. This will be used to prove that $\betamin(d) < \beta_c(d) < d$. Then, we use that analysis to classify the global maximizers of the free energy at all temperatures. Lastly, we show that the transition at the critical inverse temperature $\beta_c(d)$ is first order.

Since $F_\beta$ is invariant under permutations of the coordinates, maximizing $F_\beta$ on the ordered simplex $\Delta_d^\downarrow$ is equivalent to maximizing it on the full simplex
\begin{equation}
    \Delta_d
    :=
    \left\{
        x\in\mathbb R_{\geq0}^d:
        \sum_{i=1}^d x_i=1
    \right\}
\end{equation}
and then sorting the maximizer in decreasing order. No maximizer lies on a boundary face where some coordinate is zero. Indeed, the entropy derivative $-\log t-1$ tends to $+\infty$ as $t\downarrow0$, while the quadratic part has bounded derivative on $\Delta_d$. Thus moving an infinitesimal amount of mass into a zero coordinate strictly increases
$F_\beta$. Hence every maximizer is an interior critical point. By Lagrange multipliers, an interior critical point satisfies
\begin{equation}\label{eq:free-energy-critical-equation}
    \beta x_i-\log x_i=c
    \qquad
    \text{for all }i.
\end{equation}
The function $t\mapsto \beta t-\log t$ is strictly convex on $(0,\infty)$, so \cref{eq:free-energy-critical-equation} has at most two positive solutions. If all coordinates are equal, we obtain the uniform point $u=(1/d,\ldots,1/d)$. Otherwise, after sorting, every non-uniform critical point has the form
\begin{equation}
    x
    =
    (\underbrace{a,\ldots,a}_{k\text{ times}}, \underbrace{b,\ldots,b}_{d-k\text{ times}}),
    \qquad
    a>b>0,
    \qquad
    d > k \geq 1
\end{equation}
We show below that only $k=1$ can give a non-uniform local maximizer of $F_\beta$. Writing $r=a/b>1$, the constraint $ka+(d-k)b=1$ gives
\begin{equation}
    a=\frac{r}{d-k+kr},
    \qquad
    b=\frac{1}{d-k+kr}.
\end{equation}
The critical-point \cref{eq:free-energy-critical-equation} gives
\begin{equation}\label{eq:beta-k-branch}
    \beta
    =
    \frac{\log(a/b)}{a-b}
    =
    \frac{d-k+kr}{r-1}\log r.
\end{equation}

\begin{lemma}
    When $k \geq 2$, $x = (\underbrace{a,\dots, a}_\text{$k$ times}, \underbrace{b, \dots, b}_\text{$d-k$ times})$ is an unstable critical point of $F_\beta(x)$. It cannot be a local maximum.
\end{lemma}
\begin{proof}
    Assume $k \geq 2$ and consider the Hessian of $F_\beta$:
    \begin{equation}
    \nabla^2 F_\beta(x) = \mathrm{diag}\left(\beta - \frac{1}{x_1}, \dots, \beta - \frac{1}{x_d}\right).
    \end{equation}
    Recall that a necessary condition for a point to be a local maximum is that the Hessian is negative semi-definite (NSD) on the tangent space $\sum_i v_i=0$. Let $v = (v_1, \dots, v_d)$ be a valid perturbation of $(x_1, \dots, x_d)$, i.e. $\sum_i v_i = 0$. Then,
    \begin{equation}
    v^T \nabla^2F_\beta(x) v = \sum_{i=1}^d \left(\beta - \frac{1}{x_i}\right) v_i^2.
    \end{equation}
    Take $v_1 = \eps, v_k = -\eps$ and all other $v_i = 0$, i.e. $v = (\underbrace{\eps, 0, \dots, 0, -\eps}_\text{$k$}, \underbrace{0, \dots, 0}_\text{$d-k$})$. This is a valid perturbation for sufficiently small $\eps>0$. Then,
    \begin{equation}
    v^T \nabla^2F_\beta(x) v = 2\eps^2\left(\beta - \frac{1}{a}\right).
    \end{equation}
    This is non-positive only if $\beta \leq 1/a$. However, this is not possible:
    \begin{align}
        \beta = \frac{\log(a/b)}{a-b} \leq \frac{1}{a}
        \Leftrightarrow a\log(a/b) \leq a-b
        \Leftrightarrow \log(r) \leq 1 - \frac{1}{r},
    \end{align}
    but for $r > 1$ we have $\log(r) > 1-\frac{1}{r}$. Thus, $\beta > 1/a$, so the Hessian of $F_\beta$ is not NSD at those critical points.
\end{proof}

\begin{corollary}
    All non-uniform local maxima of $F_\beta(x)$ occur at
    \begin{equation}
    x_r = (a,b,\dots, b) = \left(\frac{r}{r+d-1}, \frac{1}{r+d-1}, \dots, \frac{1}{r+d-1}\right),
    \end{equation}
    where for fixed $\beta > 0$, the parameter $r>1$ must satisfy
\begin{equation}
\beta = \beta_d(r) = \frac{d + r - 1}{r-1}\log(r).
\end{equation}
\end{corollary}

For the rest of this section we shall keep track of the explicit dependence of $\beta$ on $r$ and $d$. First, we analyze the shape of $\beta_d(r)$ and show that there is a unique minimum $\betamin(d) := \beta_d(r_*) < d$. Observe that for $\eps > 0$ small enough,
\begin{equation}
    \beta_d(1 + \eps) = \frac{d+\eps}{\eps}\log(1+\eps) \approx d + \frac{\eps}{2}(2-d) + O(\eps^2).
\end{equation}
For $d \geq 3$, this estimate tells us that $\beta_d(r)$ is decreasing near $r = 1$. Moreover,
\begin{equation}\label{eq:beta-prime-proof}
    \beta_d'(r)
    =
    \frac{r^2-dr\log r+(d-2)r-(d-1)}{r(r-1)^2}
    =
    \frac{(r-1)^2 - d(r\log r-r+1)}{r(r-1)^2},
\end{equation}
which is 0 if and only if $r$ satisfies
\begin{equation}
    d = f(r) := \frac{(r-1)^2}{r\log(r)-r+1}.
\end{equation}
But $f' > 0$ for all $r > 1$, so $\beta_d'(r)$ can vanish at most once and so $\beta_d(r)$ has at most one minimum. Hence
there is a unique $r_*=r_*(d)>1$ such that
\begin{equation}
    \beta_d'(r)<0
    \quad\text{for }1<r<r_*,
    \qquad
    \beta_d'(r)>0
    \quad\text{for }r>r_*.
\end{equation}
Therefore
\begin{equation}
    \betamin(d)
    =
    \beta_d(r_*)
    =
    \min_{r>1}\beta_d(r).
\end{equation}

\begin{center}
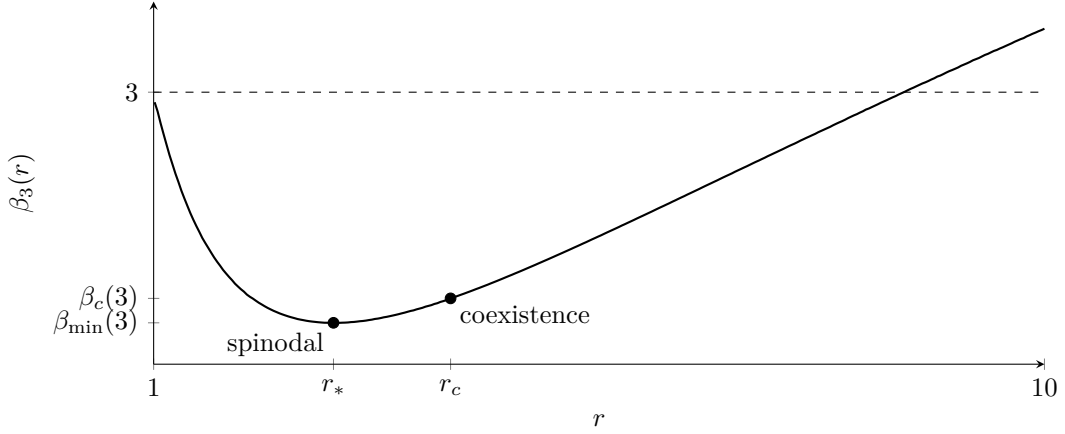

\begin{minipage}{0.92\textwidth}
\centering
\begin{tikzpicture}
\begin{axis}[
    width=0.88\textwidth,
    height=0.42\textwidth,
    axis lines=left,
    xlabel={$r$},
    ylabel={$\beta_3(r)$},
    xmin=1,
    xmax=10,
    ymin=2.70,
    ymax=3.10,
    domain=1.01:10,
    samples=300,
    xtick={1,2.8169605,4,10},
    xticklabels={$1$,$r_*$,$r_c$,$10$},
    ytick={2.7456436,2.7725887,3},
    yticklabels={$\betamin(3)$,$\beta_c(3)$,$3$},
    clip=false
]
\addplot[thick] {(x+2)/(x-1)*ln(x)};
\addplot[dashed] coordinates {(1,3) (10,3)};
\addplot[only marks,mark=*] coordinates
    {(2.8169605,2.7456436) (4,2.7725887)};
\node[anchor=north east] at (axis cs:2.8169605,2.7456436)
    {spinodal};
\node[anchor=south west] at (axis cs:4,2.735)
    {coexistence};
\end{axis}
\end{tikzpicture}
\captionof{figure}{Plot of $\beta_d(r)$ for $d=3$. Its unique minimum gives the spinodal threshold $\betamin(3)$; the point $r_c$ gives the critical inverse temperature $\beta_c(3)$; and the limiting value $\lim_{r\downarrow1}\beta_3(r)=3$ is the threshold for the stability of the uniform distribution. The same qualitative geometry holds for every $\beta_d(r)$, $d\geq3$.}
\label{fig:spectral-critical-branch}
\end{minipage}
\end{center}

Observe that for fixed $\beta$, there are either 0, 1, or 2 solutions in $r$ for $r > 1$ (i.e. non-uniform solutions). For $\beta < \betamin(d)$, no non-uniform solution exists. When $\beta > d$, there is a unique non-uniform solution. In the window $\beta \in (\betamin(d), d)$ there are two solutions for $r$ satisfying $\beta = \beta_d(r)$. Call them $r_1$ and $r_2$ ordered so that $r_1 < r_* < r_2$. We determine which of these points correspond to local maxima of $F_\beta$. At $x_r$, the Hessian is diagonal with entries
\begin{equation}
    \beta-\frac1a,
    \qquad
    \beta-\frac1b,
    \ldots,
    \beta-\frac1b.
\end{equation}
On the tangent subspace supported on the last $d-1$ coordinates with coordinate sum zero, the Hessian has value
\begin{equation}
    \beta-\frac1b
    =
    \beta-(r+d-1)
    =
    (r+d-1)
    \left(
        \frac{\log r}{r-1}-1
    \right)
    <0\quad \text{for } r>1.
\end{equation}
The remaining tangent direction is spanned by
\begin{equation}
    \left(
        1,
        -\frac1{d-1},
        \ldots,
        -\frac1{d-1}
    \right).
\end{equation}
Along this direction, the sign of the Hessian is the sign of
\begin{equation}
    d\beta_d(r)-\frac{(r+d-1)^2}{r} = - (r-1)(r+d-1)\beta_d'(r),
\end{equation}
using \cref{eq:beta-prime-proof}. Therefore the Hessian is negative definite on the tangent space exactly when $r>r_*$. Thus the smaller solution $1<r_1<r_*$ is an unstable saddle point, while the larger solution $r_2>r_*$ is a local maximum.

Consider the uniform point $u = (\tfrac{1}{d}, \dots, \tfrac{1}{d})$. Its Hessian on the tangent space is $(\beta-d)I$. Thus the uniform point is a strict local maximum for $\beta<d$, degenerate at $\beta=d$, and not a local maximum for $\beta>d$. At $\beta=d$, the feasible perturbation $u+t(1,-1/(d-1),\ldots,-1/(d-1))$, with $t>0$ small, increases $F_d$ through a positive cubic term, so $u$ is not a local maximum at this endpoint either.

It remains to compare the free energy of the uniform point with the local maximum $x_{r_2}$. Henceforth we refer to this point as simply $x_r$. Along the branch $\beta=\beta_d(r)$, define
\begin{equation}
    \varphi(r)
    :=
    F_{\beta_d(r)}(x_r)-F_{\beta_d(r)}(u).
\end{equation}
A direct calculation gives
\begin{equation}\label{eq:varphi-formula-proof}
    \varphi(r)
    =
    \log\left(\frac{r+d-1}{d}\right)
    -
    \frac{r(d+1)+d-1}{2d(r+d-1)}\log r.
\end{equation}
Differentiating,
\begin{equation}\label{eq:varphi-prime-proof}
    \varphi'(r)
    =
    \frac{d-1}{2dr(r+d-1)^2}
    \left[
        r^2+(d-2)r-(d-1)-dr\log r
    \right].
\end{equation}
The bracket in \cref{eq:varphi-prime-proof} is exactly the numerator of $\beta_d'(r)$. Hence $\varphi'$ has the same sign as $\beta_d'$:
\begin{equation}
    \varphi'(r)<0
    \quad\text{for }1<r<r_*,
    \qquad
    \varphi'(r)>0
    \quad\text{for }r>r_*.
\end{equation}
We now locate the unique crossing on the local-maximum branch. Let $r_c := (d-1)^2$. Substituting $r_c$ into \cref{eq:varphi-formula-proof} gives $\varphi(r_c) = 0$. Moreover,
\begin{equation}
    \beta_d(r_c)
    =
    \frac{2(d-1)}{d-2}\log(d-1)
    =
    \beta_c(d).
\end{equation}
This is the critical inverse temperature at which the uniform point $u$ and non-uniform point $x_r$ swap dominance of the free energy. We can also easily see that $\beta_c(d)<d$. Indeed, writing $x=d-1>1$, then $\beta_c(d) < d$ is equivalent to
\begin{equation}
    2x\log x<x^2-1.
\end{equation}
To see that this inequality holds for $x>1$, note that at $x=1$, the two sides are equal and the derivative of the LHS is $2(1 + \log x)$, while the derivative on the RHS is $2x$, which is larger for all $x >1$. The same inequality shows that $r_c>r_*$. Indeed, the numerator of
$\beta_d'(r)$ at $r=r_c$ is
\begin{align}
    (r_c-1)(r_c+d-1)-dr_c\log r_c &= d(d-1) [d(d-2)-2(d-1)\log(d-1)]\\
    &= (x+1)x[(x+1)(x-1) - 2x\log x],
\end{align}
which is positive by the inequality just proved. Hence $\beta_d'(r_c)>0$, and therefore $r_c$ lies on the increasing interval $r>r_*$. Since $\varphi$ is strictly increasing on this interval and $\varphi(r_c)=0$, the crossing at $r_c$ is the unique crossing between the uniform point and the non-uniform local maximum branch.

We can now classify the global maximizers of the free energy $F_\beta(x)$. If $0<\beta<\beta_c(d)$, then the uniform point $u$ is the global maximum. This is because for $0 < \beta<\betamin(d)$ the uniform point is the only local maximum. At $\beta=\betamin(d)$, the non-uniform critical point $x_{r_*}$ has a nonzero cubic term along the symmetry-breaking direction and is not a local maximum, so $u$ remains the only local maximum. For $\betamin(d) < \beta < \beta_c(d)$, the non-uniform point $x_r$ is a local maximum with parameter $r < r_c$, and so the uniform point is the higher local maximum.

If $\beta=\beta_c(d)$, then at $r = r_c$ the points $x_r$ and $u$ have the same free energy. Thus, there are two distinct global maximizers on
$\Delta_d^\downarrow$,
\begin{equation}
    u
    =
    \left(\frac1d,\ldots,\frac1d\right)
\end{equation}
and
\begin{equation}
    x_c
    =
    x_{r_c}
    =
    \left(
        \frac{d-1}{d},
        \frac{1}{d(d-1)},
        \ldots,
        \frac{1}{d(d-1)}
    \right).
\end{equation}

If $\beta>\beta_c(d)$, then the non-uniform local maximum has parameter $r_\beta>r_c$ and hence has strictly larger free energy than the uniform point. All other non-uniform critical points are saddles. Thus the unique global maximizer on $\Delta_d^\downarrow$ is
\begin{equation}
    x_\beta
    =
    \left(
        \frac{r_\beta}{r_\beta+d-1},
        \frac{1}{r_\beta+d-1},
        \ldots,
        \frac{1}{r_\beta+d-1}
    \right),
\end{equation}
where $r_\beta>r_*$ is the solution of $\beta = \beta_d(r)$ on the increasing branch of $\beta_d$. For $\beta\geq d$, this solution is unique because $\lim_{r\downarrow1}\beta_d(r)=d$, $\beta_d$ initially decreases below $d$, and then increases strictly to infinity on the branch $r>r_*$.

Finally, since $r_c>r_*$,
\begin{equation}
    \betamin(d)=\beta_d(r_*)<\beta_d(r_c)=\beta_c(d).
\end{equation}
Together with $\beta_c(d)<d$, this proves
\begin{equation}
    \betamin(d)<\beta_c(d)<d.
\end{equation}
At $\beta_c(d)$, the largest coordinate of the global maximizer jumps from $\frac1d$ to
\begin{equation}
    \frac{r_c}{r_c+d-1}
    =
    \frac{(d-1)^2}{(d-1)^2+d-1}
    =
    \frac{d-1}{d}.
\end{equation}
Therefore the transition is first order for $d\geq3$. This proves all the claims in \cref{thm:spectral-free-energy-landscape}.

\section{Mixing in the Young Diagram Sector}
\label{app:young-diagram-mixing}

This appendix proves \cref{thm:spectral-sector-gap-bounds}. Throughout, $d\geq 3$ and $\beta>0$ are fixed, the Davies weights are Metropolis weights, and $L$ denotes the Young diagram generator from \cref{thm:master-eqn-coefficients}. For $\lambda\in\Omega_{n,d}$, recall that $x_\lambda:=\frac{\lambda}{n}$. Recall the graph structure on $\Omega_{n,d}$ where two vertices $\lambda,\mu$ are connected by an edge iff $\lambda \sim \mu$.

\subsection{Markov Chain Preliminaries}

First we will provide the relevant definitions and results needed from the Markov chain literature. These results and their proofs can be found, for instance, in~\cite{LevinPeresWilmer2006}. The continuous-time versions below follow by applying the discrete-time bounds to $P=I+L/r$, where $r\geq\max_x(-L(x,x))$ and $r>0$, and rescaling by $r$.

Let $\Omega$ be a finite state space and let $L$ be the generator of an irreducible reversible continuous-time Markov chain on $\Omega$ with stationary distribution $\pi$. For $x\neq y$, define the \textbf{edge measure} $Q(x,y) := \pi(x)L(x,y)$, and set $Q(x,x):=0$. By reversibility, $Q(x,y)=Q(y,x)$. We call $(x,y)$ a directed edge whenever $L(x,y)>0$, and write $\calE := \{(x,y)\in\Omega^2:x\neq y,\ L(x,y)>0\}$. For a directed edge $e=(x,y)$, we also write $Q(e):=Q(x,y)$. For subsets $A,B\subseteq\Omega$, define
\begin{equation}
    Q(A,B)
    :=
    \sum_{x\in A}
    \sum_{y\in B}
    Q(x,y).
\end{equation}
For a non-empty set $S\subseteq\Omega$, its \textbf{conductance} is
\begin{equation}\label{eq:set-conductance}
    \Phi(S)
    :=
    \frac{Q(S,S^c)}{\pi(S)}.
\end{equation}
The conductance of the chain is
\begin{equation}
    \Phi
    :=
    \min_{\substack{S\subseteq\Omega\\
    0<\pi(S)\leq 1/2}}
    \Phi(S).
\end{equation}

\begin{fact}[Conductance bounds {\cite[Theorem~13.14 and its proof]{LevinPeresWilmer2006}}]
\label{fact:conductance-gap-upper-bound}
For every $S\subseteq\Omega$ satisfying $0<\pi(S)\leq\frac12$, one has
\begin{equation}\label{eq:conductance-gap-upper-bound}
    \gap(L)
    \leq
    \frac{Q(S,S^c)}
    {\pi(S)\pi(S^c)}
    \leq
    2\Phi(S).
\end{equation}
If $r\geq\max_x\sum_{y\neq x}L(x,y)$ and $r>0$, then Cheeger's inequality gives
\begin{equation}\label{eq:conductance-gap-lower-bound}
    \frac{\Phi^2}{2r}\leq\gap(L)\leq 2\Phi.
\end{equation}
\end{fact}

It will be convenient to distinguish the states inside a set which can exit the set in one step.

\begin{definition}[Boundary]
\label{def:internal-vertex-boundary}
For $S\subseteq\Omega$, define its (internal) vertex boundary by
\begin{equation}
    \partial S
    :=
    \left\{
        x\in S:
        \text{there exists }y\in S^c
        \text{ with }L(x,y)>0
    \right\}.
\end{equation}
Thus,
\begin{equation}\label{eq:boundary-flow}
    Q(S,S^c)
    =
    \sum_{x\in\partial S}
    \pi(x)
    \sum_{y\in S^c}L(x,y).
\end{equation}
\end{definition}

The second standard tool we use is the canonical-path method. For every ordered pair $x,y\in\Omega$, choose a simple path $\Gamma_{x,y} = (x=x_0,x_1,\ldots,x_\ell=y)$ such that $L(x_i,x_{i+1})>0$ for $i=0,\ldots,\ell-1$. We write $\vert \Gamma_{x,y} \vert :=\ell$ for the length of the path. The \textbf{congestion} of the collection
$\Gamma=\{\Gamma_{x,y}\}_{x,y\in\Omega}$ is
\begin{equation}\label{eq:canonical-path-congestion}
    \rho(\Gamma)
    :=
    \max_{e\in\calE}
    \frac1{Q(e)}
    \sum_{\substack{x,y\in\Omega\\
    \Gamma_{x,y} \ni e}}
    \pi(x)\pi(y)
    \vert\Gamma_{x,y}\vert.
\end{equation}

\begin{fact}[Canonical-path bound {\cite[Theorem~5]{Sinclair1992}}]
\label{fact:canonical-path-bound}
For any choice of paths
$\Gamma=\{\Gamma_{x,y}\}_{x,y\in\Omega}$,
\begin{equation}\label{eq:canonical-path-gap-bound}
    \gap(L)^{-1}
    \leq
    \rho(\Gamma).
\end{equation}
\end{fact}

We apply these facts below to the Young diagram chain with
\begin{equation}
    \Omega=\Omega_{n,d},
    \qquad
    L(\lambda,\mu)=L_{\lambda\to\mu},
    \qquad
    \pi=\pi_\beta.
\end{equation}
The slow-mixing argument will construct a set $\mathcal W_n$ with exponentially small conductance, while the fast-mixing argument will use paths to the unique maximum to bound conductance from below. The canonical-path bound supplies the universal exponential lower bound.

\subsection{Useful Estimates}

We begin with elementary results that will be used at all temperatures.

\begin{lemma}[Size and connectivity of the state space]
\label{lem:young-state-space-geometry}
The state space satisfies $\vert\Omega_{n,d}\vert \leq (n+1)^{d-1}$. The legal one-box move graph on $\Omega_{n,d}$ is connected, every state has at most $d(d-1)$ neighbors, and any two states can be connected by a path of legal one-box moves of length at most $2n$.
\end{lemma}

\begin{proof}
A diagram is determined by its first $d-1$ row lengths, each of which lies between $0$ and $n$. This proves the cardinality bound. A legal move is specified by picking one of $d$ source rows and one of $d-1$ target rows. This gives the degree bound. To prove connectivity, set $\lambda^\star=(n,0,\ldots,0)$. If $\lambda\neq\lambda^\star$, remove a box from the lowest non-empty row of $\lambda$ and add it to the first row. The removed box is a corner box, and adding a box to the first row always produces a Young diagram. Repeating this operation reaches $\lambda^\star$ after $n-\lambda_1\leq n$ moves. Joining two states through $\lambda^\star$ gives a path of length at most $2n$.
\end{proof}

\begin{lemma}[Polynomial transition rates]\label{lem:young-polynomial-rates}
For $n\geq d$, every legal non-zero transition rate satisfies
\begin{equation}\label{eq:young-uniform-rates}
    a_{\beta,d}:=\frac{e^{-2\beta}}{d^2}
    \leq L_{\lambda\to\mu}\leq n.
\end{equation}
\end{lemma}

\begin{proof}
    Recall that
    \begin{equation}\label{eq:transition-coefficients}
        L_{\lambda \to \mu}
        =
        \gamma(E_\mu - E_\lambda)
        \frac nd
        \frac{\dim(V^\mu)}{\dim(V^\nu)}
        \frac{\dim(W^\nu)}{\dim(W^\lambda)}\quad
    \text{if}\ \lambda \overset{\nu}{\sim} \mu\ \text{and}\ \lambda\neq\mu.
    \end{equation}
    Suppose that $\mu=\lambda-e_r+e_s$. That is, $\mu$ is attained from $\lambda$ by removing a box from row $r$ and adding it back into row $s$. By \cref{lem:energy-eigenvalues},
    \begin{equation}
        \vert E_\mu-E_\lambda \vert
        =
        \frac{\vert \lambda_r-\lambda_s+s-r-1\vert}{n}
        \leq
        \frac{n+d-1}{n}
        \leq
        1 + \frac{d}{n}.
    \end{equation}
    Hence the Metropolis factor $\gamma(\omega) = \min\{1,e^{-\beta\omega}\}$ lies between $e^{-\beta(1+d/n)}$ and $1$.
    
    We bound the $SU(d)$ dimension ratios and the $S_n$ dimension ratios independently to get uniform upper and lower bounds. Let $\nu=\lambda-e_r = \mu - e_s$ be the common $(n-1)$ box diagram. From the Weyl dimension formula (\cref{fact:weyl-dimension}), defining $x_i = \mu_i - i$, we get
    \begin{equation}
        \frac{\dim(V^\mu)}{\dim(V^\nu)}
        =
        \prod_{\substack{i=1\\ i\neq s}}^d \frac{x_s - x_i}{x_s-x_i-1}
        =
        \prod_{i < s} \frac{x_s - x_i}{x_s-x_i-1} \prod_{i > s} \frac{x_s - x_i}{x_s-x_i-1}.
    \end{equation}
    If $s<d$, then $\mu_s \geq \mu_{s+1} + 1$ since we are able to remove a box from row $s$ in $\mu$. Hence, $x_s - x_{s+1} \geq 2$, and $x_s-x_i\geq2$ for every $i>s$. Therefore
    \begin{equation}
        \prod_{i > s}\frac{x_s-x_i}{x_s-x_i-1} \geq 1
    \end{equation}
    since each factor is at least one. If $s=d$, this is an empty product and equals one.
    
    For $i < s$ we can write the terms in the product as $\tfrac{y_i}{y_i + 1}$ where $y_i= x_i - x_s$. From the condition $\mu_1 \geq \mu_2 \geq \dots, \geq \mu_d$ we can see that the $y_i$ are distinct positive integers. So, we minimize the product by taking the $y_i$ as small as possible, e.g. take $y_i = 1, 2, \dots, s-1$. Then,
    \begin{equation}
        \prod_{i < s} \frac{y_i}{y_i+1} \geq \frac{1}{2}\cdot \frac{2}{3} \cdots \frac{s-1}{s} = \frac{1}{s} \geq \frac{1}{d}.
    \end{equation}
    Putting these together gives
    \begin{equation}
        \frac{\dim(V^\mu)}{\dim(V^\nu)} \geq \frac{1}{d}.
    \end{equation}
    For the upper bound, Pieri's rule (\cref{fact:Pieri-rule}) shows that $V^\mu$ is a summand of $V^\nu\otimes\mathbb C^d$, since $\mu\in\nu^+$. Taking dimensions gives
    \begin{equation}
        \frac{\dim(V^\mu)}{\dim(V^\nu)} \leq d.
    \end{equation}
    Now we turn to the $S_n$ dimension ratio. By the hook-length formula (\cref{fact:hook-length-formula}),
    \begin{equation}
        \frac{\dim(W^\nu)}{\dim(W^\lambda)} = \frac{(n-1)!/\prod_{b \in \nu} h_\nu(b)}{n!/\prod_{b\in \lambda} h_\lambda(b)},
    \end{equation}
    where $h_\lambda(b)$ is the hook-length of the box $b$ in the diagram $\lambda$. The removed corner has hook-length one. Every remaining box has the same hook-length in both diagrams, except those in the same row or column as the removed box, whose hook-length decreases by one. Thus,
    \begin{equation}
        \frac{\dim(W^\nu)}{\dim(W^\lambda)} = \frac{1}{n} \prod_{\substack{b\in\nu\\ h_\nu(b)=h_\lambda(b)-1}} \frac{h_\lambda(b)}{h_\lambda(b)-1},
    \end{equation}
    since for the affected boxes we have $h_\nu(b) = h_\lambda(b) - 1$. Any box that is not the removed box but whose hook-length is affected by the removal must have had a hook length of at least 2. Hence, each term in the product is $\geq 1$ giving a lower bound of
    \begin{equation}
        \frac{\dim(W^\nu)}{\dim(W^\lambda)} \geq \frac{1}{n}.
    \end{equation}
    An upper bound of 1 follows from Young's branching Rule (\cref{fact:Young-branching-rule}). Combining these estimates with the bounds on the Metropolis factor gives
    \begin{equation}
        \frac{e^{-\beta(1+d/n)}}{d^2}
        \leq L_{\lambda\to\mu}\leq n.
    \end{equation}
    For $n\geq d$, this gives \cref{eq:young-uniform-rates}.
\end{proof}

We now give a rough lower bound that supplies the lower half of the gap estimate in the metastable regime.

\begin{lemma}[A universal exponential lower bound]
\label{lem:young-universal-exponential-lower-bound}
For fixed $d$ and $\beta>0$, there is $C_{\beta,d}<\infty$ such that
\begin{equation}
    \gap(L)\geq e^{-C_{\beta,d}n}
\end{equation}
for every sufficiently large $n$.
\end{lemma}
\begin{proof}
Every Schur--Weyl block is nonzero, so its Gibbs probability is at least the smallest eigenvalue of $\rho$. Thus the Gibbs estimate~\cref{eq:gibbs-min-eigenvalue-bound} gives
\begin{equation}
    \pi_{\min}:=\min_\lambda\pi_\beta(\lambda)
    \geq\lambda_{\min}(\rho)
    \geq e^{-C_0n}
\end{equation}
for some $C_0=C_0(\beta,d)$.

By \cref{lem:young-state-space-geometry}, choose simple legal paths between all pairs of states with length at most $2n$. For $n\geq d$, \cref{lem:young-polynomial-rates} gives $Q(e)\geq a_{\beta,d}\pi_{\min}$ on every edge. Since the sum of all pair weights $\pi_\beta(\lambda)\pi_\beta(\mu)$ is one, the canonical-path bound (\cref{fact:canonical-path-bound}) yields
\begin{equation}
    \gap(L)^{-1}
    \leq\rho(\Gamma)
    \leq\frac{2n}{a_{\beta,d}\pi_{\min}}
    \leq e^{C_{\beta,d}n}
\end{equation}
for a sufficiently large constant $C_{\beta,d}$.
\end{proof}

\subsection{The Metastable Window}

We shall also use the following rounding map to bridge the gap between the continuous free energy defined on $\Delta_d^\down$ and the discrete chain on $\Omega_{n,d}$. For $x\in\Delta_d$, define an integer composition of $n$ by
\begin{equation}
\label{eq:composition-rounding-map}
    \widehat R_n(x)_i
    :=
    \lfloor nx_i\rfloor
    \quad (1\leq i<d),
    \qquad
    \widehat R_n(x)_d
    :=
    n-\sum_{i=1}^{d-1}\lfloor nx_i\rfloor.
\end{equation}
Then
\begin{equation}
\label{eq:composition-rounding-error}
    \left\Vert\frac{\widehat R_n(x)}n-x\right\Vert_1
    \leq
    \frac{2(d-1)}n.
\end{equation}
If $a^\downarrow$ denotes the decreasing rearrangement of $a$, set $R_n(x) := \widehat R_n(x)^\downarrow$. Decreasing rearrangement is non-expansive in $\ell^1$, hence whenever $x\in\Delta_d^\downarrow$,
\begin{equation}
\label{eq:young-diagram-rounding-error}
    R_n(x)\in\Omega_{n,d},
    \qquad
    \left\Vert\frac{R_n(x)}n-x\right\Vert_1
    \leq
    \frac{2(d-1)}n.
\end{equation}

We shall also repeatedly use the following modulus of continuity. If $x,y\in\Delta_d$ and $\varepsilon=\Vert x-y\Vert_1\leq1/2$, then
\begin{equation}
\label{eq:free-energy-modulus}
    \vert F_\beta(x)-F_\beta(y)\vert
    \leq
    \beta\varepsilon
    +
    \varepsilon\log\left(\frac{ed}{\varepsilon}\right).
\end{equation}
Indeed, $\frac\beta2 \left\vert \sum_i x_i^2-\sum_i y_i^2 \right\vert \leq \beta\Vert x-y\Vert_1$, while $\vert-a\log a+b\log b\vert \leq \vert a-b \vert \log\left(\frac{e}{\vert a-b \vert}\right)$ for $a,b\in[0,1]$. Summing the latter inequality and using concavity of $t\mapsto t\log(e/t)$ gives \cref{eq:free-energy-modulus}. In particular, there is $C_{\beta,d}<\infty$ such that any legal one-box move satisfies
\begin{equation}
\label{eq:one-box-free-energy-modulus}
    \vert F_\beta(x_\lambda)-F_\beta(x_\mu)\vert
    \leq
    C_{\beta,d}\frac{\log(n+1)}n.
\end{equation}

\begin{theorem}[Slow mixing in the metastable window]\label{thm:young-chain-slow-mixing}
    If $\betamin(d)<\beta<d$, then there exist $0<c_{\beta,d}<C_{\beta,d}<\infty$ such that
    \begin{equation}
        e^{-C_{\beta,d}n}
        \leq
        \gap(L)
        \leq
        e^{-c_{\beta,d}n}
    \end{equation}
    for every sufficiently large $n$.
\end{theorem}
\begin{proof}
    The lower bound follows from \cref{lem:young-universal-exponential-lower-bound}. For the upper bound, we'll construct for each $n$ a set $\mathcal W_n \subset \Omega_{n,d}$ whose conductance is exponentially small. This provides an exponentially small upper bound on the gap via \cref{fact:conductance-gap-upper-bound}. First, for $\beta \in (\betamin(d), d)$ we know from the local-maxima classification in \cref{subsec:proof-spectral-free-energy-landscape} that the free energy $F_\beta$ admits two distinct strict local maxima on $\Delta_d^\down$. When $\beta \neq \beta_c$, denote the smaller of the two maxima by $m$ and the global maximum by $g$. Choose a sufficiently small open ball $B$ around $m$ and also $\delta > 0$ so that
    \begin{equation}\label{eq:metastable-ball-properties}
        F_\beta(x) \leq F_\beta(m) \quad \text{for all } x \in \overline{B} \qquad \text{and} \qquad \sup_{x \in \partial B} F_\beta(x) \leq F_\beta(m) - 4\delta.
    \end{equation}
    Here $\overline{B}$ and $\partial B$ denote the closure and boundary of $B$ in $\Delta_d^\down$, respectively. Such a choice is possible because $m$ is a strict local maximum (choose $\delta$ to be smaller than a quarter of the height from $m$ to the trough separating $m$ and $g$). At $\beta = \beta_c(d)$, denote the two local maxima by $m_1$ and $m_2$ and define disjoint balls $B_1$ and $B_2$ around them. Choose a common $\delta>0$ small enough that both triples $(m_i,B_i,\delta)$ satisfy \cref{eq:metastable-ball-properties}. In either case, fix these choices as functions of $\beta,d$, independently of $n$.

    We now round these to a set $\mathcal W_n \subset \Omega_{n,d}$ satisfying $0 < \pi_\beta(\mathcal W_n) \leq \frac12$ by rounding $m$ and $B$ to Young diagrams. First suppose that $\beta \neq \beta_c(d)$ and set
    \begin{equation}
        \mathcal W_n = \{\lambda \in \Omega_{n,d}: x_\lambda \in B\},
    \end{equation}
    this takes the points in $B$ which correspond to Young diagrams. Set $\lambda^{(g)} = R_n(g)$ so that $\Vert x_{\lambda^{(g)}} - g \Vert_1 < \frac{2(d-1)}{n}$ which implies $F_\beta(x_{\lambda^{(g)}}) \geq F_\beta(g) - C_{\beta,d} \frac{\log(n+1)}{n}$. Let $\Delta = F_\beta(g) - F_\beta(m) > 0$ be the difference between the two well heights. By \cref{eq:metastable-ball-properties} and \cref{cor:stationary-probability-ratios}, for $\lambda \in \mathcal W_n$ we have
    \begin{align}
        \frac{\pi_\beta(\lambda)}{\pi_\beta(\lambda^{(g)})} &\leq n^{c'}\exp\{n(f_\beta(x_\lambda) - f_\beta(x_{\lambda^{(g)}}))\}\\
        &\leq n^{c'} \exp\{n(f_\beta(m) - f_\beta(g) + C_{\beta,d}\frac{\log(n+1)}{n})\}\\
        &\leq n^Ce^{-n\Delta} \leq e^{-n\Delta/2}
    \end{align}
    for all large enough $n$. Since $\vert\Omega_{n,d}\vert$ is polynomial in $n$, we see that
    \begin{equation}
        \begin{aligned}
        \pi_\beta(\mathcal W_n)
        &= \sum_{\lambda \in \mathcal W_n} \pi_\beta(\lambda)
        \leq \vert\Omega_{n,d}\vert e^{-\Delta n/2}\pi_\beta(\lambda^{(g)})\\
        &\leq (n+1)^{d-1}e^{-\Delta n/2}<\frac12,
        \end{aligned}
    \end{equation}
    for all sufficiently large $n$.
    At $\beta = \beta_c(d)$, define
    \begin{equation}
        \mathcal W_{i,n} = \{\lambda \in \Omega_{n,d} : x_\lambda \in B_i\} \quad \text{for } i = 1,2.
    \end{equation}
    These sets are non-empty and disjoint for large enough $n$, thus at least one of them has stationary mass $\leq \frac{1}{2}$. For each $n$ choose $i = i(n)$ so that $\pi_\beta(\mathcal W_{i,n}) \leq \frac12$ and for the remainder of the proof set $\mathcal W_n := \mathcal W_{i,n}, B := B_i, m := m_i$. This completes the construction of the desired set with mass $\leq \frac12$ in every case.

    We now show that this set has small conductance. Set $\lambda^{(m)} = R_n(m)$ and note that $\lambda^{(m)} \in \mathcal W_n$ for all large enough $n$. Moreover, $\Vert x_{\lambda^{(m)}}-m\Vert_1 \leq \frac{2(d-1)}n$ and the modulus estimate gives $F_\beta(x_{\lambda^{(m)}}) \geq F_\beta(m)-C_{\beta,d}\frac{\log(n+1)}n$. Pick $\lambda \in \partial \mathcal W_n$. Then, $x_\lambda \in B$ and $x_\lambda$ has a neighbor $x_\mu \in B^c$. The line segment between $x_\lambda$ and $x_\mu$ crosses $\partial B$ at a point whose $\ell^1$-distance from $x_\lambda$ is at most $2/n$. Hence, \cref{eq:free-energy-modulus,eq:metastable-ball-properties} imply that for all sufficiently large $n$,
    \begin{equation}
    \label{eq:metastable-boundary-height}
        F_\beta(x_\lambda)
        \leq
        F_\beta(m)-2\delta
        \qquad
        \text{for every }\lambda\in\partial \mathcal W_n.
    \end{equation}
    It follows from \cref{cor:stationary-probability-ratios} and \cref{eq:metastable-boundary-height} that for every $\lambda \in \partial \mathcal W_n$,
    \begin{align}
        \frac{\pi_\beta(\lambda)}
             {\pi_\beta(\lambda^{(m)})}
        &\leq n^{c'} \exp\{n(f_\beta(m) - 2\delta - f_\beta(m))\}\\
        &\leq e^{-n\delta}
    \end{align}
    for all sufficiently large $n$, where the polynomial factor also absorbs the rounding error at $m$. Therefore
    \begin{equation}
    \label{eq:metastable-boundary-mass}
        \frac{\pi_\beta(\partial \mathcal W_n)}{\pi_\beta(\mathcal W_n)}
        \leq
        \frac{\vert \partial \mathcal W_n\vert e^{-\delta n}
              \pi_\beta(\lambda^{(m)})}
             {\pi_\beta(\lambda^{(m)})}
        \leq (n+1)^{d-1}e^{-\delta n}
        \leq e^{-\delta n/2}.
    \end{equation}
    By \cref{lem:young-state-space-geometry} and the rate bound~\cref{eq:young-uniform-rates}, the total exit rate from any state is at most $d(d-1)n$ for $n\geq d$. Thus, by \cref{fact:conductance-gap-upper-bound}
\begin{equation}
    \begin{aligned}
    \gap(L) &\leq
    \frac{Q(\mathcal W_n,\mathcal W_n^c)}
         {\pi_\beta(\mathcal W_n)\pi_\beta(\mathcal W_n^c)}
    \leq
    2\frac{\sum_{\lambda \in \mathcal W_n}\sum_{\mu \in \mathcal W_n^c}\pi_\beta(\lambda)L_{\lambda \to \mu}}{\pi_\beta(\mathcal W_n)}
    \\
    &\leq
    2d(d-1)n\frac{\pi_\beta(\partial \mathcal W_n)}{\pi_\beta(\mathcal W_n)} \leq e^{-\delta n/4}.
    \end{aligned}
\end{equation}
where we used $\pi_\beta(\mathcal W_n^c) \geq \frac12$. Both final exponential bounds hold for all sufficiently large $n$. Since $\delta$ was fixed in terms of $\beta,d$, this proves the upper bound with $c_{\beta,d}=\delta/4$.
\end{proof}

\subsection{The One-Well Regimes}

By the classification of local maxima in
\cref{subsec:proof-spectral-free-energy-landscape}, $F_\beta$
has a unique local maximum $m_\beta$ on $\Delta_d^\downarrow$ whenever
$0<\beta<\betamin(d)$ or $\beta>d$. We first construct continuous
paths to this maximum and then approximate them by legal one-box moves.

\begin{lemma}[Discrete paths with controlled free energy loss]\label{lem:one-well-discrete-paths}
    Suppose that $0<\beta<\betamin(d)$ or $\beta>d$, and set $\lambda^{(m)}:=R_n(m_\beta)$. For $n\geq14(d-1)$ and every $\lambda\in\Omega_{n,d}$, there is a simple legal path $\Gamma_\lambda$ from $\lambda$ to $\lambda^{(m)}$ such that
    \begin{equation}\label{eq:one-well-path-height}
        F_\beta(x_\eta)
        \geq F_\beta(x_\lambda)
           -\frac{7(d-1)\log n+C_{\beta,d}}{n}
        \qquad(\eta\in\Gamma_\lambda),
    \end{equation}
    where $C_{\beta,d}$ is independent of $n$.
\end{lemma}
\begin{proof}
    We first choose a continuous path to $m_\beta$ and then approximate it by a sequence of legal box moves. We will keep every diagram close to the continuous path, so that the modulus of continuity controls the resulting loss in free energy.
    
    First, every nonempty strict superlevel set $U_a:=\{x\in\Delta_d^\downarrow:F_\beta(x)>a\}$ is path connected. Indeed, let $C$ be a path component of $U_a$, and maximize $F_\beta$ on its compact closure, obtaining a point $z$. Since $C$ contains points of value greater than $a$, $F_\beta(z)>a$. A sufficiently small relative convex neighborhood of $z$ lies in $U_a$ and intersects $C$, so it lies entirely in $C$. Consequently $z\in C$, and $z$ is a local maximum of $F_\beta$ on $\Delta_d^\downarrow$. Hence $z=m_\beta$. Every path component therefore contains $m_\beta$, proving the claim.

    Fix $\lambda\in\Omega_{n,d}$. Both $x_\lambda$ and $m_\beta$ belong to $U_a$ for $a:=F_\beta(x_\lambda)-1/n$, so there is a continuous path $\gamma:[0,1]\to\Delta_d^\downarrow$ between them satisfying
    \begin{equation}
        F_\beta(\gamma(t))>F_\beta(x_\lambda)-\frac1n\qquad(0\leq t\leq1).
    \end{equation}
    By uniform continuity, we can choose finitely many points along $\gamma$, including both endpoints, so that consecutive points are at most $1/n$ apart in $\ell^1$ distance. Round each point to an integer composition using $\widehat R_n$ from \cref{eq:composition-rounding-map}. The rounding error bound \cref{eq:composition-rounding-error} shows that consecutive compositions are at most $4(d-1)+1$ apart in integer $\ell^1$ distance.

    To connect two consecutive compositions, move one box at a time from a coordinate that is larger than its target value to one that is smaller. Each coordinate moves only toward its target, so every intermediate composition stays within $4(d-1)+1$ of the starting composition. Including rounding error, this puts every intermediate composition within $6(d-1)+1\leq7(d-1)$ of the preceding continuous path point.
    
    We now sort all these compositions in decreasing order to obtain Young diagrams. The continuous path is already ordered, so sorting does not increase the distance to it, as in \cref{eq:young-diagram-rounding-error}. Consecutive sorted compositions have the same sum and integer $\ell^1$ distance at most $2$. They therefore either coincide or differ by a legal one-box move. The endpoints are $R_n(x_\lambda)=\lambda$ and $R_n(m_\beta)=\lambda^{(m)}$. Removing repeated consecutive diagrams and erasing loops gives a simple legal path between them. This only removes vertices, so the distance bound is preserved.

    It remains to check the free energy along this path. Every rescaled vertex $x_\eta$ is within $A_d/n$ of a point on $\gamma$, where $A_d:=7(d-1)$. Since $n\geq2A_d$, this distance is at most $1/2$, and we can apply \cref{eq:free-energy-modulus}. Together with the lower bound on $F_\beta$ along $\gamma$, this gives
    \begin{equation}
        n\bigl(F_\beta(x_\lambda)-F_\beta(x_\eta)\bigr)
        \leq1+A_d\left[\beta+\log\left(\frac{edn}{A_d}\right)\right]
        \leq A_d\log n+C_{\beta,d},
    \end{equation}
    where the terms independent of $n$ are absorbed into $C_{\beta,d}$. This proves \cref{eq:one-well-path-height}.
\end{proof}

\begin{theorem}[Fast mixing in the one-well regimes]\label{thm:young-chain-fast-mixing}
    There is a universal constant $c>0$ such that, whenever $0<\beta<\betamin(d)$ or $\beta>d$,
    \begin{equation}
        \gap(L)
        =\gap\!\left(\Lloc\big\vert_{\calA_d}\right)
        \geq n^{-cd^2}
    \end{equation}
    for all sufficiently large $n$, with the threshold depending on $\beta,d$.
\end{theorem}
\begin{proof}
    Take the paths from \cref{lem:one-well-discrete-paths}. Their free energy never falls far below its starting value, so \cref{cor:stationary-probability-ratios} gives
    \begin{equation}\label{eq:one-well-path-weights}
        \pi_\beta(\lambda)\leq K_{n,\beta,d}\pi_\beta(\eta)
        \qquad(\eta\in\Gamma_\lambda),
    \end{equation}
    where
    \begin{equation}\label{eq:one-well-weight-loss}
        K_{n,\beta,d}\leq C_{\beta,d}n^{d^2+7(d-1)}.
    \end{equation}
    Here the dimension bounds contribute $n^{d^2}$, while exponentiating the free energy loss in \cref{eq:one-well-path-height} contributes $n^{7(d-1)}$. All factors independent of $n$ are absorbed into $C_{\beta,d}$.

    We use these paths to bound the conductance from below. Write $N:=|\Omega_{n,d}|$, and consider any set $S$ with $0<\pi_\beta(S)\leq1/2$. Let $A$ be whichever of $S$ and $S^c$ does not contain $\lambda^{(m)}$. In either case, $\pi_\beta(A)\geq\pi_\beta(S)$. Choose a state $\lambda\in A$ of largest stationary probability. Since $A$ contains at most $N$ states,
    \begin{equation}
        \pi_\beta(\lambda)\geq\frac{\pi_\beta(A)}N
        \geq\frac{\pi_\beta(S)}N.
    \end{equation}
    The path from $\lambda$ to $\lambda^{(m)}$ must cross from $A$ to $A^c$ at some edge $(\eta,\eta')$. By reversibility, this edge contributes to $Q(S,S^c)$ regardless of which side contains $\lambda^{(m)}$. The rate bound~\cref{eq:young-uniform-rates} and \cref{eq:one-well-path-weights} therefore give 

    \begin{equation}
        Q(S,S^c)\geq Q(\eta,\eta')
        \geq a_{\beta,d}\pi_\beta(\eta)
        \geq\frac{a_{\beta,d}}{NK_{n,\beta,d}}\pi_\beta(S).
    \end{equation}
    Since $S$ was arbitrary, $\Phi\geq a_{\beta,d}/(NK_{n,\beta,d})$.

    By \cref{lem:young-state-space-geometry}, $N\leq(n+1)^{d-1}$ and each state has at most $d(d-1)$ neighbors. Together with $L_{\lambda\to\mu}\leq n$, this bounds the total exit rate by $d(d-1)n$. Cheeger's inequality \cref{eq:conductance-gap-lower-bound} now yields
    \begin{equation}
        \gap(L)\geq\frac{\Phi^2}{2d(d-1)n}
        \geq\frac{a_{\beta,d}^2}
          {2d(d-1)nN^2K_{n,\beta,d}^2}.
    \end{equation}
    Substituting the bounds on $N$ and $K_{n,\beta,d}$ gives a power of $2d^2+16(d-1)+1=O(d^2)$ in the denominator. Absorbing the factors independent of $n$ into a slightly larger power gives the claimed bound for all sufficiently large $n$.
\end{proof}

\Cref{thm:young-chain-fast-mixing} proves \cref{thm:spectral-sector-gap-bounds}\textup{(i)}, while \cref{thm:young-chain-slow-mixing} proves \cref{thm:spectral-sector-gap-bounds}\textup{(ii)}.

\section{Asymptotic Bounds on \texorpdfstring{$\betamin(d)$}{beta min(d)}}\label{app:beta-min-asymptotics}

Recall that, for $d\geq3$,
\begin{equation}
    \betamin(d)
    =
    \min_{r>1}
    \frac{r+d-1}{r-1}\log r.
\end{equation}
Equivalently, after setting $y=r-1$,
\begin{equation}
    \betamin(d)
    =
    \min_{y>0}
    \frac{y+d}{y}\log(1+y).
\end{equation}

\begin{proposition}[Asymptotics of $\betamin(d)$]\label{prop:beta-min-asymptotics}
    For $d\geq3$,
    \begin{equation}
        \log(d\log d)+1
        < \betamin(d)
        \leq \left(1+\frac{1}{\log d}\right)\log(1+d\log d).
    \end{equation}
    Consequently,
    \begin{equation}
        \betamin(d)
        =
        \log d+\log\log d+1+o_d(1)
        \qquad
        \text{as }d\to\infty.
    \end{equation}
\end{proposition}

\begin{proof}
    We first prove the upper bound. Evaluating the objective at $y=d\log d$ gives
    \begin{equation}
        \betamin(d)
        \leq \frac{d\log d+d}{d\log d}\log(1+d\log d)
        = \left(1+\frac{1}{\log d}\right)\log(1+d\log d).
    \end{equation}

    We now prove the lower bound. Set $x:=\log(1+y)$, and $y=e^x-1$. Then
    \begin{equation}
        \betamin(d)
        =
        \min_{x>0}
        \left(1+\frac{d}{e^x-1}\right)x.
    \end{equation}
    Define
    \begin{equation}
        a:=\log d,
        \qquad
        A:=\log(d\log d)+1=a+\log a+1.
    \end{equation}
    We show that, for every $x>0$,
    \begin{equation}
        \left(1+\frac{d}{e^x-1}\right)x
        \geq
        A.
    \end{equation}
    First, if $x\geq A$, then $\left(1+\frac{d}{e^x-1}\right)x > x \geq A$, so it remains to consider $0<x<A$. Set $u := A-x$. Then $0<u<A$. The desired inequality is
    \begin{equation}
        A
        \leq
        x+\frac{dx}{e^x-1}.
    \end{equation}
    Since $d=e^a$ and $A=a+\log a+1$, this is equivalent to
    \begin{align}
        A
        \leq
        x+\frac{e^a x}{e^x-1}
        \Longleftrightarrow
        (A-x)(e^x-1)
        \leq
        e^a x
        \Longleftrightarrow
        u(e^{A-u}-1)
        \leq
        e^a(A-u).
    \end{align}
    Using $e^A = e^{a+\log a+1} = eae^a$, the last inequality becomes
    \begin{align}
        u(eae^{a-u}-1)
        \leq
        e^a(A-u)
        \Longleftrightarrow
        eaue^{-u}-ue^{-a}
        \leq
        A-u
        \Longleftrightarrow
        u(1+eae^{-u}-e^{-a})
        \leq
        A.
    \end{align}
    Thus it suffices to prove $p(u)\leq A$ for $0\leq u\leq A$, where $p(u) := u(1+eae^{-u}-e^{-a})$. We have $p(0)=0$ and
    \begin{align}
        p(A)
        =
        A(1+eae^{-A}-e^{-a})
        =
        A(1+e^{-a}-e^{-a})
        =
        A.
    \end{align}
    It remains to control possible interior maxima. Differentiating gives,
    \begin{equation}
        p'(u)
        =
        1-e^{-a}+eae^{-u}(1-u),\qquad \text{and}\qquad
        p''(u)
        =
        eae^{-u}(u-2).
    \end{equation}
    Hence $p'$ decreases on $(0,2)$ and increases on $(2,\infty)$. Since $p'(u)>0$ for $0\leq u\leq1$, any interior local maximum must occur at some $u_0\in(1,2)$. If no such local maximum exists, then the maximum of $p$ on $[0,A]$ is attained at an endpoint, and the claim follows. Suppose, then, that such a local maximum exists. Write
    \begin{equation}
        u_0=1+w,
        \qquad
        0<w<1.
    \end{equation}
    Since $p'(u_0)=0$,
    \begin{align}
        0
        =
        1-e^{-a}+eae^{-(1+w)}(1-(1+w))
        =
        1-e^{-a}-awe^{-w}.
    \end{align}
    Thus $1-e^{-a} = awe^{-w}$. At this point,
    \begin{align}
        p(u_0)
        =
        (1+w)(1-e^{-a}+eae^{-(1+w)})
        =
        (1+w)(awe^{-w}+ae^{-w})
        =
        ae^{-w}(1+w)^2.
    \end{align}
    We now bound this expression. First suppose $1<a\leq e$. The function $w\mapsto e^{-w}(1+w)^2$ is increasing on $0<w<1$. Therefore $p(u_0) \leq \frac{4a}{e}$ which is $\leq a+\log a + 1 = A$ for $1\leq a\leq e$. Indeed, define
    \begin{equation}
        \psi(a)
        :=
        a+\log a+1-\frac{4a}{e}.
    \end{equation}
    Then $\psi$ is concave on $[1,e]$, so its minimum on this interval is attained at an endpoint. Since
    \begin{equation}
        \psi(1)=2-\frac4e>0,
        \qquad
        \psi(e)=e-2>0,
    \end{equation}
    we have $\psi(a)>0$ on $[1,e]$. Hence $p(u_0)\leq A$. Now suppose $a\geq e$. Since $u_0\in(1,2)$ and $u_0e^{1-u_0}\leq1$, we get
    \begin{align}
        p(u_0)
        =
        u_0(1-e^{-a})+a u_0e^{1-u_0}
        \leq
        u_0+a
        <
        a+2.
    \end{align}
    Since $a\geq e$, we have $\log a\geq1$, and therefore $a+2 \leq a+\log a+1 = A$. Thus $p(u_0)\leq A$ also in this case.

    We have shown that every possible interior local maximum of $p$ is at most $A$, and the endpoints are also at most $A$. Therefore
    \begin{equation}
        p(u)\leq A
        \qquad
        \text{for all }0\leq u\leq A.
    \end{equation}
    The inequalities above are strict for $0<u<A$: at a possible interior maximum this follows from the strict endpoint estimates for $\psi$, while the only endpoint at which $p(u)=A$ is $u=A$, corresponding to the excluded value $x=0$. Hence,
    \begin{equation}
        \betamin(d) = \min_{x > 0}\left(1+\frac{d}{e^x-1}\right)x
        >
        A.
    \end{equation}
    Finally,
    \begin{equation}
        \log\left(1+\frac{1}{d\log d}\right)=o_d(1),
        \qquad
        \frac{\log\log d}{\log d}=o_d(1),
    \end{equation}
    so the upper bound is
    \begin{equation}
        \betamin(d)
        \leq
        \log(d\log d)+1+o_d(1).
    \end{equation}
    Since $\log(d\log d)+1 = \log d+\log\log d+1$, we conclude
    \begin{equation}
        \betamin(d)
        =
        \log d+\log\log d+1+o_d(1).
    \end{equation}
\end{proof}

\section{The Nontrivial \texorpdfstring{$SU(d)$}{SU(d)} Sector}\label{app:sud-sector}

For a single Hermitian jump operator $A$, let $\calL_A$ denote the Davies generator with the same Hamiltonian and inverse temperature, generated by the single jump $A$. For $\alpha=1,\ldots,d^2-1$, define the total Gell--Mann generator
\begin{equation}
    \Lambda_\alpha^{\tot}
    :=
    \sum_{u=1}^n (\lambda^\alpha)^{(u)}.
\end{equation}
We define the $SU(d)$ group-mixer generator by
\begin{equation}
    \calL_{SU(d)}(O)
    :=
    \sum_{\alpha=1}^{d^2-1}
    \calL_{\Lambda_\alpha^{\tot}/\sqrt{2d}}(O).
\end{equation}
Since each $\Lambda_\alpha^{\tot}$ commutes with the interchange Hamiltonian, these jumps have only Bohr frequency zero. With the Metropolis weight $\gamma(0)=1$, \cref{eq:zero-frequency-davies} gives
\begin{equation}\label{eq:su-group-mixer-double-commutator}
    \calL_{SU(d)}(O)
    =
    -\frac{1}{4d} \sum_{\alpha=1}^{d^2-1} [\Lambda_\alpha^{\tot}, [\Lambda_\alpha^{\tot},O]].
\end{equation}

\begin{lemma}\label{lem:group-mixer-gap}
    The kernel of $\calL_{SU(d)}$ is
    \begin{equation}
        \ker(\calL_{SU(d)})
        =
        \mathrm{comm}(SU(d)).
    \end{equation}
    Moreover,
    \begin{equation}
        -\calL_{SU(d)}
        \big\vert_{\mathrm{comm}(SU(d))^{\perp_\rho}}
        \geq
        \idty
        \big\vert_{\mathrm{comm}(SU(d))^{\perp_\rho}}.
    \end{equation}
    In fact, $\gap(\calL_{SU(d)}) = 1$.
\end{lemma}
\begin{proof}[Proof of \cref{lem:group-mixer-gap}]
     By \cref{eq:dirichlet-form-commutator}, its Dirichlet form is
    \begin{equation}
        \langle O,-\calL_{SU(d)}(O)\rangle_\rho
        =
        \frac{1}{4d}
        \sum_{\alpha=1}^{d^2-1}
        \Vert[\Lambda_\alpha^{\tot},O]\Vert_\rho^2.
    \end{equation}
    Therefore $\calL_{SU(d)}(O)=0$ if and only if $[\Lambda_\alpha^{\tot},O] = 0$ for every $\alpha=1,\ldots,d^2-1$. Equivalently, $O$ commutes with the global $SU(d)$ action. Hence, $\ker(\calL_{SU(d)}) = \mathrm{comm}(SU(d))$.
    
     It remains to compute the first non-zero eigenvalue. For a unitary representation $R$ of $SU(d)$, let $G_\alpha$ be its Hermitian generators, defined by $R(e^{it\lambda^\alpha})=e^{itG_\alpha}$. The quadratic Casimir $\sum_\alpha G_\alpha^2$ commutes with the representation and hence acts as a scalar on each irreducible summand, by Schur's lemma. For the $SU(d)$ observable representation of \cref{def:observable-representation}, the Hermitian generators are $G_\alpha(O)=[\Lambda_\alpha^{\tot},O]$, so $-\calL_{SU(d)}=\frac{1}{4d}\sum_\alpha G_\alpha^2$. We decompose the observable representation into irreducible $SU(d)$-modules. Schur--Weyl duality gives
    \begin{equation}
        \calB(\calH)
        \cong
        \bigoplus_{\lambda,\mu\in\Omega_{n,d}}
        \left(V^\lambda\otimes(V^\mu)^*\right)
        \otimes
        \left(W^\lambda\otimes(W^\mu)^*\right).
    \end{equation}
    Under the $SU(d)$ action, the Specht factors are multiplicity spaces. Note that the tensor product $V^\lambda\otimes(V^\mu)^*$ need not be irreducible. Decomposing these tensor products further gives the irreducible summands $V^\nu$ of the observable representation. Choose a partition label $\nu$ with at most $d$ rows for each such summand and set $m:=|\nu|$. On $V^\nu$, the quadratic Casimir acts as the scalar
    \begin{equation}\label{eq:su-casimir-partition}
        C_2(\nu)
        =2\left[\sum_{i=1}^d\nu_i(\nu_i+d+1-2i)-\frac{m^2}{d}\right],
        \qquad m=\sum_{i=1}^d\nu_i.
    \end{equation}
    This is the Casimir eigenvalue formula written in terms of the partition labeling the representation; see~\cite[Proposition~10.6]{Hall2015}. The factor of $2$ corresponds to our normalization $\Tr(\lambda^\alpha\lambda^\beta)=2\delta_{\alpha\beta}$. Thus $-\calL_{SU(d)}$ acts on this summand as $C_2(\nu)/(4d)$.

    We now identify a restriction on the partitions that can occur in the observable representation. If $z^d=1$, then $zI\in SU(d)$ acts on the physical Hilbert space $\calH=(\mathbb C^d)^{\otimes n}$ as $(zI)^{\otimes n} = z^nI$. The phase cancels under conjugation, so $zI$ acts as the identity on $\calB(\calH)$ and therefore on every irreducible summand $V^\nu$. On the other hand, Schur--Weyl duality realizes the representation $V^\nu$ inside $(\mathbb C^d)^{\otimes m}$, where $zI$ acts as $z^mI$. Consequently, $z^m=1$ for every $z$ satisfying $z^d=1$. Thus $m=qd$ for an integer $q$. Consider
    \begin{equation}
    a_i:=\nu_i-q,
    \qquad
    a_1\geq\cdots\geq a_d,
    \qquad
    a_i\in\mathbb Z,
    \qquad
    \sum_{i=1}^da_i=0.
    \end{equation}
    
    If $V^\nu$ is non-trivial, its rows are not all equal, so $a\neq0$. The partition formula~\cref{eq:su-casimir-partition} gives
    \begin{equation}
        C_2(\nu)
        =
        2\left[
            \sum_{i=1}^da_i^2
            -
            2\sum_{i=1}^d ia_i
        \right].
    \end{equation}
    Set
    \begin{equation}
        S_k:=\sum_{i=1}^ka_i,
        \qquad
        k=1,\ldots,d.
    \end{equation}
    Since $a$ is non-increasing and has total sum zero, each $S_k\geq0$. If $S_k=0$ for some $k<d$, then the first $k$ entries and the last $d-k$ entries both have average zero. Since $a$ is non-increasing, this forces all entries of $a$ to be zero, contradicting $a\neq0$. Since $a$ is a non-zero integer vector, in fact $S_k\geq1$ for every $k=1,\ldots,d-1$. Summation by parts gives
    \begin{equation}
        -2\sum_{i=1}^d ia_i
        =
        2\sum_{k=1}^{d-1}S_k
        \geq
        2(d-1).
    \end{equation}
    Also, since $a$ is a non-zero integer vector with total sum zero,
    \begin{equation}
        \sum_{i=1}^da_i^2\geq2.
    \end{equation}
    Therefore
    \begin{equation}
        C_2(\nu)
        \geq
        2\left(2+2(d-1)\right)
        =
        4d.
    \end{equation}
    Thus every non-trivial irreducible summand has eigenvalue at least $1$ under $-\calL_{SU(d)}$. Moreover, equality is attained. The adjoint representation occurs inside the observable algebra, for example in the span of the total generators $\mathrm{span}\{\Lambda_1^{\tot},\ldots,\Lambda_{d^2-1}^{\tot}\}$. The Gell--Mann completeness relation gives $\sum_\alpha[\lambda^\alpha,[\lambda^\alpha,X]]=4dX-4\Tr(X)I$. Since operators on different sites commute, this implies
    \begin{equation}
        \sum_{\alpha=1}^{d^2-1}
        [\Lambda_\alpha^{\tot},[\Lambda_\alpha^{\tot},\Lambda_b^{\tot}]]
        =4d\,\Lambda_b^{\tot},
        \qquad b=1,\ldots,d^2-1.
    \end{equation}
    Since each nonzero total generator $\Lambda_b^{\tot}$ is an eigenvector of $-\calL_{SU(d)}$ with eigenvalue $1$, we conclude that $\gap(\calL_{SU(d)})=1$. Thus
    \begin{equation}
        -\calL_{SU(d)}
        \big\vert_{\mathrm{comm}(SU(d))^{\perp_\rho}}
        \geq
        I
        \big\vert_{\mathrm{comm}(SU(d))^{\perp_\rho}}.
    \end{equation}
\end{proof}

\begin{lemma}[Unitary freedom of Davies jump operators]\label{lem:unitary-freedom-davies}
    Let $\{A_a\}_{a=1}^m$ and $\{B_b\}_{b=1}^m$ be two families of Hermitian jump operators related by a unitary change of basis in the jump index:
    \begin{equation}
        B_b
        =
        \sum_{a=1}^m U_{ba}A_a,
        \qquad
        U\in U(m).
    \end{equation}
    Then the corresponding Davies generators are equal, $\calL_{\{A_a\}_{a=1}^m} = \calL_{\{B_b\}_{b=1}^m}.$
\end{lemma}
\begin{proof}[Proof of \cref{lem:unitary-freedom-davies}]
    For every Bohr frequency $\omega$, the Bohr components satisfy
    \begin{equation}
        B_b(\omega)
        =
        \sum_{a=1}^m U_{ba}A_a(\omega).
    \end{equation}
    Therefore, by unitarity of $U$,
    \begin{align}
        \sum_{b=1}^m B_b(\omega)^\dagger X B_b(\omega)
        &=
        \sum_{a,c=1}^m
        \left(\sum_{b=1}^m \overline{U_{ba}}U_{bc}\right)
        A_a(\omega)^\dagger X A_c(\omega)\\
        &=
        \sum_{a=1}^m A_a(\omega)^\dagger X A_a(\omega).
    \end{align}
    and so the completely positive parts of their associated Davies generators agree. Taking $X=I$ one sees that their dissipative terms likewise agree.
\end{proof}

\begin{proposition}[Single-site jumps simulate the $SU(d)$ group mixer]\label{prop:single-site-simulates-su-mixer}
    There exists a Davies generator $\Lrest$ such that
    \begin{equation}
        \Lloc
        =
        \frac1n\calL_{SU(d)} + \Lrest .
    \end{equation}
    Consequently,
    \begin{equation}
        -\Lloc
        \geq
        -\frac1n\calL_{SU(d)}
    \end{equation}
    as quadratic forms in the KMS inner product.
\end{proposition}
\begin{proof}[Proof of \cref{prop:single-site-simulates-su-mixer}]
    Fix $\alpha\in\{1,\ldots,d^2-1\}$. Choose a real orthogonal matrix $R\in O(n)$ whose first row is the normalized all-ones vector:
    \begin{equation}
        R_{1u}=\frac1{\sqrt n},
        \qquad
        u=1,\ldots,n.
    \end{equation}
    Apply this change of basis to the $n$ normalized single-site jumps
    \begin{equation}
        J_1^\alpha,\ldots,J_n^\alpha,
        \qquad
        J_u^\alpha
        =
        \frac{1}{\sqrt{2d}}(\lambda^\alpha)^{(u)}.
    \end{equation}
    Define
    \begin{equation}
        B_{\alpha,k}
        :=
        \sum_{u=1}^n R_{ku}J_u^\alpha,
        \qquad
        k=1,\ldots,n.
    \end{equation}
    Since $R$ is real and the $J_u^\alpha$ are Hermitian, each $B_{\alpha,k}$ is Hermitian. Moreover,
    \begin{equation}
        B_{\alpha,1}
        =
        \frac1{\sqrt n}
        \sum_{u=1}^n J_u^\alpha
        =
        \frac{1}{\sqrt{2dn}}
        \Lambda_\alpha^{\tot}.
    \end{equation}
    By \cref{lem:unitary-freedom-davies},
    \begin{equation}
        \sum_{u=1}^n \calL_{J_u^\alpha}
        =
        \sum_{k=1}^n \calL_{B_{\alpha,k}}.
    \end{equation}
    Separating the first term and using $\calL_{cA}=\vert c\vert^2\calL_A$ gives
    \begin{equation}
        \sum_{u=1}^n \calL_{J_u^\alpha}
        =
        \frac1n
        \calL_{\Lambda_\alpha^{\tot}/\sqrt{2d}}
        +
        \sum_{k=2}^n\calL_{B_{\alpha,k}}.
    \end{equation}
    Summing over $\alpha=1,\ldots,d^2-1$ yields
    \begin{equation}
        \Lloc
        =
        \frac1n\calL_{SU(d)}
        +
        \Lrest,
    \end{equation}
    where
    \begin{equation}
        \Lrest
        :=
        \sum_{\alpha=1}^{d^2-1}
        \sum_{k=2}^n
        \calL_{B_{\alpha,k}}.
    \end{equation}
    Since $\Lrest$ is a Davies generator generated by Hermitian jumps, it is negative semidefinite in the KMS inner product. Therefore
    \begin{equation}
        -\Lloc
        =
        -\frac1n\calL_{SU(d)}
        -
        \Lrest
        \geq
        -\frac1n\calL_{SU(d)}
    \end{equation}
    as quadratic forms.
\end{proof}

\section{The Permutation-Nonsymmetric Sector \texorpdfstring{$\mathrm{comm}(S_n)^{\perp_\rho}$}{}}\label{app:Sn-sector}

In this appendix we control the permutation-non-symmetric sector $\mathrm{comm}(S_n)^{\perp_\rho}$ by comparison with an $S_n$ group mixer. Under the observable representation (\cref{def:observable-representation}), the auxiliary generator is a central element of the group algebra $\C[S_n]$. Thus, we may compute its spectral gap directly from the content formula for irreducible representations of $S_n$. Here, for $\sigma\in S_n$, this representation is
\begin{equation}
    \mathcal R(\sigma)(O)
    =
    P(\sigma)O P(\sigma)^\dagger.
\end{equation}
Because $H$, and therefore $\rho$, commutes with every permutation, $\mathcal R$ is a unitary representation of $S_n$ for the KMS inner product. Define
\begin{align}\label{eq:S_n-mixer}
    \mathcal L_{S_n}(O)
    &:=
    \frac{2}{n}
    \sum_{1\leq u<v\leq n}
    \left(T_{uv}OT_{uv}-O\right)\\
    &=
    \frac2n \sum_{1 \leq u< v\leq n} (\calR((uv))(O) - O)
\end{align}
This is the Davies generator with zero-frequency jumps $\sqrt{2/n}\,T_{uv}$, so \cref{eq:zero-frequency-davies} applies. The normalization in \cref{eq:S_n-mixer} is chosen so that the gap is independent of $n$. We proceed in three steps. First, we compute the exact gap of the $S_n$ group mixer. Second, we control products of local jumps using their one-site bandwidth and modular norms. Finally, we compare the Dirichlet form of the $S_n$ group mixer to the local Davies Dirichlet form.

\begin{remark}
    The analogous group-mixer bound in~\cite{BassoBergamaschiLinRagoneStubbs2026} uses transpositions associated with an expander graph. The representation-theoretic argument below gives the exact gap of $2$ directly, without an expander construction. This argument provides a self-contained alternative for obtaining the system-size-independent bound for all $d \geq 2$.
\end{remark}

\begin{lemma}[Gap of the $S_n$ mixer]\label{lem:complete-transposition-gap}
    Assume $d\geq2$ and $n\geq2$. Then $\ker(\mathcal L_{S_n}) = \mathrm{comm}(S_n)$, and the spectral gap of $-\mathcal L_{S_n}$ is exactly $2$. Equivalently,
    \begin{equation}\label{eq:complete-transposition-poincare}
        \langle O,-\mathcal L_{S_n}(O)\rangle_\rho
        \geq
        2\Vert O\Vert_\rho^2,
        \qquad
        O\in\mathrm{comm}(S_n)^{\perp_\rho}.
    \end{equation}
\end{lemma}

\begin{proof}
    Since $T_{uv}$ is a self-adjoint involution commuting with $\rho$,
    \begin{equation}\label{eq:complete-mixer-dirichlet}
        \langle O,-\mathcal L_{S_n}(O)\rangle_\rho
        =
        \frac1n\sum_{u<v}\Vert[T_{uv},O]\Vert_\rho^2.
    \end{equation}
    Thus, $O$ lies in the kernel precisely when it commutes with every transposition. Since the set of transpositions generates $S_n$, this proves the kernel statement.
    
    It remains to compute the smallest non-zero eigenvalue. Let
    \begin{equation}
        C_n:=\sum_{u<v}(uv)\in\mathbb C[S_n].
    \end{equation}
    Since transpositions form a conjugacy class, $C_n$ is a central element of the group algebra $\C[S_n]$. By Schur's Lemma, on an irreducible summand $W^\mu$ of the conjugation representation $\mathcal R$, $\mathcal R(C_n)$ acts as a scalar. We can compute the scalar in the following way. Write $C_n$ in terms of the so-called Jucys--Murphy elements
    \begin{equation}
        J_k = \sum_{i=1}^{k-1} (ik) \Rightarrow C_n = \sum_{k=2}^n J_k.
    \end{equation}
    In the Young seminormal basis indexed by standard Young tableaux $T$ of shape $\mu$, the Jucys--Murphy element $J_k$ acts diagonally with eigenvalue $c_T(k)$, the content of the box containing $k$~\cite[Chapter 3.3.3]{Ceccherini-Silberstein_Scarabotti_Tolli_2010}. If $k$ is in row $a$ and column $b$, then the content is $c_T(k) = b-a$. So,
    \begin{equation}\label{eq:content-scalar-Sn-mixer}
        c_\mu
        =
        \sum_{(a,b)\in\mu}(b-a)
        =
        \sum_{1 \leq a \leq \ell(\mu)}
        \left[\binom{\mu_a}{2}-(a-1)\mu_a\right].
    \end{equation}
    The sum over all Jucys--Murphy elements of these contents depends only on the shape $\mu$, which gives \cref{eq:content-scalar-Sn-mixer}.
    
    It follows from \cref{eq:S_n-mixer} that $-\mathcal L_{S_n}$ acts on this irreducible summand as
    \begin{equation}\label{eq:Sn-mixer-irrep-eigenvalue}
        \frac{2}{n}\left(\binom{n}{2}-c_\mu\right)I.
    \end{equation}
    We now minimize over irreps for which this quantity is non-zero. The trivial representation corresponds to $\mu=(n)$, for which $c_{(n)} = \binom{n}{2}$ and so $\frac2n\left(\binom{n}{2}-c_\mu\right) = 0$. This is expected since the trivial isotypic component of the conjugation representation is precisely $\mathrm{comm}(S_n)$. Now let $\mu\neq(n)$ and write
    \begin{equation}
        m:=n-\mu_1,
        \qquad 1\leq m\leq n-1.
    \end{equation}
    Using $\sum_{r\geq2}\mu_r=m$, we obtain
    \begin{align}
        \binom{n}{2}-c_\mu
        &=
        \binom{n}{2}-\binom{n-m}{2}
        -\sum_{r\geq2}\binom{\mu_r}{2}
        +\sum_{r\geq2}(r-1)\mu_r \\
        &\geq
        \binom{n}{2}-\binom{n-m}{2}-\binom m2+m \\
        &=
        m(n-m+1) \\
        &=
        n+(m-1)(n-m) \\
        &\geq n.
    \end{align}
    In the first inequality we used $\sum_{r\geq2}\binom{\mu_r}{2}\leq\binom m2$ and $\sum_{r\geq2}(r-1)\mu_r\geq m$. Therefore every non-trivial irreducible summand has eigenvalue at least $2$ by \cref{eq:Sn-mixer-irrep-eigenvalue}.
    
    Finally, this lower bound is attained in the observable representation. Notice that in the calculation above, we have equality throughout by taking $\mu = (n-1,1)$, i.e. $m=1$. Explicitly, choose a non-zero traceless one-site operator $X$, and let $X^{(i)}$ be its copy on site $i$. The space $\left\{\sum_{i=1}^n a_iX^{(i)} : \sum_{i=1}^n a_i=0\right\}$ is an $S_n$-subrepresentation isomorphic to $W^{(n-1,1)}$. Indeed, the operators $X^{(i)}$ are linearly independent, for example by their Hilbert--Schmidt orthogonality when $\Tr(X)=0$, and permutations act by permuting the site label. For $\mu=(n-1,1)$, one has $\binom{n}{2}-c_\mu=n$, so \cref{eq:Sn-mixer-irrep-eigenvalue} equals $2$. 
    This proves that the eigenvalue lower bound of 2 is sharp.
\end{proof}

For an operator $A$, define its modular norm by
\begin{equation}
    \Gamma_\beta(A)
    :=
    \max\left\{
        \Vert\rho^{-1/4}A\rho^{1/4}\Vert,
        \Vert\rho^{1/4}A\rho^{-1/4}\Vert
    \right\}.
\end{equation}

\begin{lemma}[Product commutators]\label{lem:product-commutator}
    If $\Gamma_\beta(A_1),\Gamma_\beta(A_2)\leq\Gamma$, then
    \begin{equation}\label{eq:product-commutator}
        \Vert[A_1A_2,O]\Vert_\rho^2
        \leq
        2\Gamma^2
        \left(
            \Vert[A_1,O]\Vert_\rho^2
            +
            \Vert[A_2,O]\Vert_\rho^2
        \right).
    \end{equation}
    \end{lemma}
    
    \begin{proof}
    Use
    \begin{equation}
        [A_1A_2,O]
        =
        A_1[A_2,O]+[A_1,O]A_2
    \end{equation}
    and the inequalities
    \begin{equation}
        \Vert AX\Vert_\rho
        \leq
        \Vert\rho^{1/4}A\rho^{-1/4}\Vert\,\Vert X\Vert_\rho,
        \qquad
        \Vert XA\Vert_\rho
        \leq
        \Vert\rho^{-1/4}A\rho^{1/4}\Vert\,\Vert X\Vert_\rho.
    \end{equation}
    The triangle inequality and $(a+b)^2\leq2(a^2+b^2)$ give     \cref{eq:product-commutator}.
\end{proof}

\begin{lemma}[One-site bandwidth and modular norm]\label{lem:one-site-bandwidth}
    There are constants $\Omega_d,C_d<\infty$, independent of $n$, such that every normalized local jump $J_i^\alpha$ satisfies
    \begin{equation}
        J_i^\alpha(\omega)=0
        \quad\text{for }\vert\omega\vert>\Omega_d,
        \qquad
        \Gamma_\beta(J_i^\alpha)
        \leq
        C_de^{\beta\Omega_d/4}.
    \end{equation}
    One may take $\Omega_d=d+1$.
\end{lemma}

\begin{proof}
    Consider the operator $\Pi_\mu \calT(\Pi_\lambda) \Pi_\mu$. It can be expressed as
    \begin{equation}
        \Pi_\mu \calT(\Pi_\lambda) \Pi_\mu = \sum_{i,\alpha} (\Pi_\lambda J_i^\alpha \Pi_\mu)^\dagger (\Pi_\lambda J_i^\alpha \Pi_\mu)
    \end{equation}
    By \cref{prop:schur-weyl-expansion-twirl}, this quantity vanishes unless $\lambda \sim \mu$. Hence, $\Pi_\lambda J_i^\alpha \Pi_\mu = 0$ unless $\lambda \sim \mu$. Now suppose that $J_i^\alpha(\omega) \neq 0$. Then there exists a pair of diagrams $\lambda, \mu$ with $\lambda \sim \mu$ such that $E_\lambda - E_\mu = \omega$. If $\omega = 0$, the desired bound is immediate. So, suppose that $\omega \neq 0$ and write $\mu = \lambda - e_r + e_s$, to denote that $\mu$ is attained from $\lambda$ by removing a box from row $r$ and adding it to row $s$. Then
    \begin{align}
        \vert \omega \vert = \vert E_\mu-E_\lambda \vert
        &=
        \left\vert \frac{\lambda_r-\lambda_s+s-r-1}{n} \right\vert \\
        &\leq \frac{\vert \lambda_r - \lambda_s \vert + \vert s-r-1 \vert}{n}\\
        &\leq 1 + \frac{d}{n}
    \end{align}
    which is at most $1+d$. This proves the bandwidth claim.
    
    Now we prove the modular norm bound. Fix a pair of distinct rows $1\leq r \neq s \leq d$ and define the injective map $\lambda \to \mu(\lambda) = \lambda - e_r + e_s$ whenever $\mu(\lambda)$ is a valid Young diagram. Then, decompose a local jump as
    \begin{equation}
        J_i^\alpha
        =
        A^0+\sum_{r\neq s}A^{r\to s},
        \qquad
        A^0
        :=
        \sum_\lambda \Pi_\lambda J_i^\alpha \Pi_\lambda,
        \qquad
        A^{r\to s}
        :=
        \sum_{\lambda:\mu(\lambda)\ \mathrm{valid}}
        \Pi_{\mu(\lambda)}J_i^\alpha\Pi_\lambda.
    \end{equation}
    For fixed $r\neq s$, the source spaces are mutually orthogonal, and injectivity of $\lambda\mapsto\mu(\lambda)$ makes the target spaces mutually orthogonal as well. Thus, write $x = \sum_\lambda x_\lambda$ where $x_\lambda \in \mathrm{Ran}(\Pi_\lambda)$. Then,
    \begin{equation}
        \Vert A^{r\to s} x \Vert^2 = \sum_\lambda \Vert\Pi_{\mu(\lambda)}J_i^\alpha x_\lambda \Vert^2 \leq \Vert J_i^\alpha \Vert^2 \sum_\lambda \Vert x_\lambda \Vert^2
    \end{equation}
    and hence $\Vert A^{r \to s} \Vert \leq \Vert J_i^\alpha \Vert$. Similarly, the diagonal term $A^0$ satisfies $\Vert A^0 \Vert \leq \Vert J_i^\alpha \Vert$. By the bandwidth claim,
    \begin{equation}
        e^{-\beta H/4}\Pi_\mu J_i^\alpha \Pi_\lambda e^{\beta H/4} = e^{\beta(E_\lambda - E_\mu)/4} \Pi_\mu J_i^\alpha \Pi_\lambda,
    \end{equation}
    so conjugation by $e^{\pm\beta H/4}$ multiplies each block by a scalar of magnitude at most $e^{\beta\Omega_d/4}$. Applying the triangle inequality and summing the at most $1+d(d-1)$ components proves the modular-norm bound.
\end{proof}

\begin{proposition}[Comparison of Dirichlet forms]\label{prop:Sn-local-comparison}
    There are constants $C_1,C_2>0$, depending only on $d$ and the chosen jump normalization, such that
    \begin{equation}\label{eq:Sn-local-comparison}
        \langle O,-\mathcal L_{S_n}(O)\rangle_\rho
        \leq
        C_1e^{C_2\beta}
        \langle O,-\Lloc(O)\rangle_\rho
    \end{equation}
    for every observable $O$, uniformly in $n$.
\end{proposition}
\begin{proof}
    The swap expansion in the normalized jump basis is
    \begin{equation}
        T_{uv}
        =
        \frac1dI +
        d\sum_{\alpha=1}^{d^2-1}J_u^\alpha J_v^\alpha.
    \end{equation}
    Write $O=\sum_\nu O(\nu)$ for its Bohr decomposition. Since $[H,T_{uv}]=0$, the operators $[T_{uv},O(\nu)]$ belonging to distinct Bohr sectors are KMS-orthogonal. Thus, by Cauchy--Schwarz, \cref{lem:product-commutator}, and \cref{lem:one-site-bandwidth},
    \begin{align}
        \Vert[T_{uv},O]\Vert_\rho^2
        &=
        \sum_\nu\Vert[T_{uv},O(\nu)]\Vert_\rho^2 \\
        &\leq
        C_de^{\beta\Omega_d/2}
        \sum_{\alpha=1}^{d^2-1}\sum_\nu
        \left(
            \Vert[J_u^\alpha,O(\nu)]\Vert_\rho^2
            +
            \Vert[J_v^\alpha,O(\nu)]\Vert_\rho^2
        \right).
    \end{align}
    Bohr-sector orthogonality also gives
    \begin{equation}\label{eq:double-Bohr-orthogonality}
        \sum_\nu\Vert[J_u^\alpha,O(\nu)]\Vert_\rho^2
        =
        \sum_\omega\Vert[J_u^\alpha(\omega),O]\Vert_\rho^2.
    \end{equation}
    Indeed, both sides equal the double sum over $(\omega,\nu)$ of $\Vert[J_u^\alpha(\omega), O(\nu)]\Vert_\rho^2$, because the summands have distinct total Bohr frequencies whenever the index being summed out varies. In particular, $[H,[J_u^\alpha(\omega), O(\nu)]] = (\omega + \nu)[J_u^\alpha(\omega), O(\nu)]$.
    
    By \cref{lem:one-site-bandwidth}, only $\vert\omega\vert\leq\Omega_d$ occur. Consequently,
    \begin{equation}
        1
        \leq
        e^{\beta\Omega_d/2}h_\omega,
        \qquad
        h_\omega=e^{-\beta\vert\omega\vert/2},
    \end{equation}
    on the support of the sum. Combining this observation with \cref{eq:complete-mixer-dirichlet} and \cref{eq:double-Bohr-orthogonality} yields
    \begin{align}
        \langle O,-\mathcal L_{S_n}(O)\rangle_\rho
        &\leq
        \frac{C_de^{\beta\Omega_d}}{n}
        \sum_{u<v}\sum_{\alpha,\omega}h_\omega
        \left(
            \Vert[J_u^\alpha(\omega),O]\Vert_\rho^2
            +
            \Vert[J_v^\alpha(\omega),O]\Vert_\rho^2
        \right) \\
        &=
        C_de^{\beta\Omega_d}\frac{n-1}{n}
        \sum_{u=1}^n\sum_{\alpha,\omega}h_\omega
        \Vert[J_u^\alpha(\omega),O]\Vert_\rho^2 \\
        &\leq
        2C_de^{\beta\Omega_d}\mathcal E_{\loc}(O),
    \end{align}
    where the last step uses \cref{eq:dirichlet-form-commutator}. Absorbing the factor $2C_d$ and $\Omega_d$ into $C_1,C_2$ proves \cref{eq:Sn-local-comparison}.
\end{proof}

\begin{proof}[Proof of \cref{thm:symmetry-controlled-sector-bounds}\textup{(ii)}]
    For $O\in\mathrm{comm}(S_n)^{\perp_\rho}$, \cref{lem:complete-transposition-gap} and \cref{prop:Sn-local-comparison} imply
    \begin{align}
        \mathcal E_{\loc}(O)
        &\geq
        \frac{1}{C_1e^{C_2\beta}}
        \langle O,-\mathcal L_{S_n}(O)\rangle_\rho \\
        &\geq
        \frac{2}{C_1e^{C_2\beta}}\Vert O\Vert_\rho^2.
    \end{align}
    Absorbing the inessential factor $2$ into $C_1$ gives the desired result.
\end{proof}

\clearpage
\bibliographystyle{plainnat}
\bibliography{refs}

@article{Bjornberg2016,
  author  = {Bj{\"o}rnberg, Jakob E.},
  title   = {The Free Energy in a Class of Quantum Spin Systems and
             Interchange Processes},
  journal = {Journal of Mathematical Physics},
  volume  = {57},
  number  = {7},
  pages   = {073303},
  year    = {2016},
  doi     = {10.1063/1.4959238}
}

@article{bergamaschi2026,
author = {Thiago Bergamaschi and Reza Gheissari and Yunchao Liu},
title = {Rapid mixing for Gibbs states within a logical sector: a dynamical view of self-correcting quantum memories},
year = {2026},
journal = {Proceedings of the 2026 Annual ACM-SIAM Symposium on Discrete Algorithms (SODA)},
chapter = {},
pages = {3407-3422},
doi = {10.1137/1.9781611978971.123},
URL = {https://epubs.siam.org/doi/abs/10.1137/1.9781611978971.123},
eprint = {https://epubs.siam.org/doi/pdf/10.1137/1.9781611978971.123},
}

@article{Sinclair1992,
  author  = {Sinclair, Alistair},
  title   = {Improved Bounds for Mixing Rates of {Markov} Chains and
             Multicommodity Flow},
  journal = {Combinatorics, Probability and Computing},
  volume  = {1},
  number  = {4},
  pages   = {351--370},
  year    = {1992},
  doi     = {10.1017/S0963548300000390}
}

@article{CuffEtAl2012,
  author  = {Cuff, Paul and Ding, Jian and Louidor, Oren and Lubetzky, Eyal
             and Peres, Yuval and Sly, Allan},
  title   = {{Glauber} Dynamics for the Mean-Field {Potts} Model},
  journal = {Journal of Statistical Physics},
  volume  = {149},
  number  = {3},
  pages   = {432--477},
  year    = {2012},
  doi     = {10.1007/s10955-012-0599-2}
}

@article{KastoryanoTemme2013,
  author  = {Kastoryano, Michael J. and Temme, Kristan},
  title   = {Quantum Logarithmic {Sobolev} Inequalities and Rapid Mixing},
  journal = {Journal of Mathematical Physics},
  volume  = {54},
  number  = {5},
  pages   = {052202},
  year    = {2013},
  doi     = {10.1063/1.4804995}
}

@article{EllisWang1990,
  author  = {Ellis, Richard S. and Wang, Kongming},
  title   = {Limit Theorems for the Empirical Vector of the
             {Curie--Weiss--Potts} Model},
  journal = {Stochastic Processes and their Applications},
  volume  = {35},
  number  = {1},
  pages   = {59--79},
  year    = {1990},
  doi     = {10.1016/0304-4149(90)90122-9}
}

@book{Sagan2001,
  author    = {Sagan, Bruce E.},
  title     = {The Symmetric Group: Representations, Combinatorial Algorithms,
               and Symmetric Functions},
  edition   = {2},
  series    = {Graduate Texts in Mathematics},
  volume    = {203},
  publisher = {Springer},
  address   = {New York},
  year      = {2001},
  doi       = {10.1007/978-1-4757-6804-6}
}

@book{FultonHarris1991,
  author    = {Fulton, William and Harris, Joe},
  title     = {Representation Theory: A First Course},
  series    = {Graduate Texts in Mathematics},
  volume    = {129},
  publisher = {Springer},
  address   = {New York},
  year      = {1991},
  doi       = {10.1007/978-1-4612-0979-9}
}

@misc{BassoBergamaschiLinRagoneStubbs2026,
  author        = {Basso, Joao and Bergamaschi, Thiago and Lin, Lin and
                   Ragone, Michael and Stubbs, Kevin D.},
  title         = {Spectral Gap of the {Davies} Generator for the Mean-Field
                   {Heisenberg} Model},
  year          = {2026},
  eprint        = {2607.21798},
  archivePrefix = {arXiv},
  primaryClass  = {quant-ph},
  note          = {arXiv:2607.21798 [quant-ph]},
  doi           = {10.48550/arXiv.2607.21798},
  url           = {https://arxiv.org/abs/2607.21798}
}

@article{carlen2017gradient,
  author  = {Carlen, Eric A. and Maas, Jan},
  title   = {Gradient Flow and Entropy Inequalities for Quantum {Markov}
             Semigroups with Detailed Balance},
  journal = {Journal of Functional Analysis},
  volume  = {273},
  number  = {5},
  pages   = {1810--1869},
  year    = {2017},
  doi     = {10.1016/j.jfa.2017.05.003}
}

@incollection{hall2013lie,
  author    = {Hall, Brian C.},
  title     = {{Lie} Groups, {Lie} Algebras, and Representations},
  booktitle = {Quantum Theory for Mathematicians},
  series    = {Graduate Texts in Mathematics},
  volume    = {267},
  pages     = {333--366},
  publisher = {Springer},
  address   = {New York},
  year      = {2013},
  doi       = {10.1007/978-1-4614-7116-5_16}
}

@misc{basso2025quantumclassicalcomparisondavies,
  author        = {Basso, Joao and Ganguly, Shirshendu and Sinclair, Alistair
                   and Srivastava, Nikhil and Stier, Zachary and
                   Vuong, Thuy-Duong},
  title         = {On Quantum to Classical Comparison for {Davies} Generators},
  year          = {2025},
  eprint        = {2510.07267},
  archivePrefix = {arXiv},
  primaryClass  = {quant-ph},
  note          = {arXiv:2510.07267 [quant-ph]},
  doi           = {10.48550/arXiv.2510.07267},
  url           = {https://arxiv.org/abs/2510.07267}
}

@book{LevinPeresWilmer2006,
  author    = {Levin, David A. and Peres, Yuval and Wilmer, Elizabeth L.},
  title     = {Markov Chains and Mixing Times},
  publisher = {American Mathematical Society},
  address   = {Providence, RI},
  year      = {2009},
  url       = {https://bookstore.ams.org/mbk-58}
}

@book{Ceccherini-Silberstein_Scarabotti_Tolli_2010, place={Cambridge}, series={Cambridge Studies in Advanced Mathematics}, title={Representation Theory of the Symmetric Groups: The Okounkov-Vershik Approach, Character Formulas, and Partition Algebras}, publisher={Cambridge University Press}, author={Ceccherini-Silberstein, Tullio and Scarabotti, Fabio and Tolli, Filippo}, year={2010}, collection={Cambridge Studies in Advanced Mathematics}}

@book{Hall2015,
  author    = {Hall, Brian C.},
  title     = {{Lie} Groups, {Lie} Algebras, and Representations:
               An Elementary Introduction},
  edition   = {2},
  series    = {Graduate Texts in Mathematics},
  volume    = {222},
  publisher = {Springer},
  address   = {Cham},
  year      = {2015},
  doi       = {10.1007/978-3-319-13467-3},
  isbn      = {978-3-319-13466-6}
}

@article{DingLubetzkyPeres2009,
  author  = {Ding, Jian and Lubetzky, Eyal and Peres, Yuval},
  title   = {The Mixing Time Evolution of {Glauber} Dynamics for the
             Mean-Field {Ising} Model},
  journal = {Communications in Mathematical Physics},
  volume  = {289},
  number  = {2},
  pages   = {725--764},
  year    = {2009},
  doi     = {10.1007/s00220-009-0781-9}
}

@misc{GamarnikKianiZlokapa2024,
  author        = {Gamarnik, David and Kiani, Bobak T. and Zlokapa, Alexander},
  title         = {Slow Mixing of Quantum {Gibbs} Samplers},
  year          = {2024},
  eprint        = {2411.04300},
  archivePrefix = {arXiv},
  primaryClass  = {quant-ph},
  doi           = {10.48550/arXiv.2411.04300},
  url           = {https://arxiv.org/abs/2411.04300}
}

@misc{alicki2008thermalstabilitytopologicalqubit,
      title={On thermal stability of topological qubit in Kitaev's 4D model}, 
      author={R. Alicki and M. Horodecki and P. Horodecki and R. Horodecki},
      year={2008},
      eprint={0811.0033},
      archivePrefix={arXiv},
      primaryClass={quant-ph},
      url={https://arxiv.org/abs/0811.0033}, 
}

@book{BovierDenHollander2015,
  author    = {Bovier, Anton and den Hollander, Frank},
  title     = {Metastability: A Potential-Theoretic Approach},
  series    = {Grundlehren der mathematischen Wissenschaften},
  volume    = {351},
  publisher = {Springer},
  address   = {Cham},
  year      = {2015},
  isbn      = {978-3-319-24777-9},
  doi       = {10.1007/978-3-319-24777-9},
  url       = {https://doi.org/10.1007/978-3-319-24777-9}
}

@misc{rakovszky2024bottlenecksquantumchannelsfinite,
      title={Bottlenecks in quantum channels and finite temperature phases of matter}, 
      author={Tibor Rakovszky and Benedikt Placke and Nikolas P. Breuckmann and Vedika Khemani},
      year={2024},
      eprint={2412.09598},
      archivePrefix={arXiv},
      primaryClass={quant-ph},
      url={https://arxiv.org/abs/2412.09598}, 
}

@misc{chen2025efficientexactnoncommutativequantum,
      title={An efficient and exact noncommutative quantum Gibbs sampler}, 
      author={Chi-Fang Chen and Michael J. Kastoryano and András Gilyén},
      year={2025},
      eprint={2311.09207},
      archivePrefix={arXiv},
      primaryClass={quant-ph},
      url={https://arxiv.org/abs/2311.09207}, 
}

@article{Ding_2025,
   title={Efficient Quantum Gibbs Samplers with Kubo–Martin–Schwinger Detailed Balance Condition},
   volume={406},
   ISSN={1432-0916},
   url={http://dx.doi.org/10.1007/s00220-025-05235-3},
   DOI={10.1007/s00220-025-05235-3},
   number={3},
   journal={Communications in Mathematical Physics},
   publisher={Springer Science and Business Media LLC},
   author={Ding, Zhiyan and Li, Bowen and Lin, Lin},
   year={2025},
   month=Feb }

\end{document}